\documentclass[11pt,letterpaper]{article}
\usepackage{fullpage}

\usepackage[mathlines]{lineno}

\usepackage[T1]{fontenc}
\usepackage{lmodern}
\usepackage{amsmath,amssymb,amsthm,mathtools}
\usepackage{mathrsfs} 
\usepackage{thmtools,thm-restate}
\makeatletter
\@ifundefined{newcounteralias}{}{%
  \renewcommand\thmt@autorefsetup{\@xa\def\csname\thmt@envname autorefname\@xa\endcsname\@xa{\thmt@thmname}}%
}
\makeatother

\newif\ifintroproofrefs
\introproofrefstrue
\newcommand{\proofref}[1]{%
  \ifintroproofrefs
    ; \S\ref{#1}%
  \fi
}

\usepackage{booktabs,array,microtype,needspace}
\usepackage{flafter}
\usepackage{float}
\usepackage[only,llbracket,rrbracket]{stmaryrd}
\SetSymbolFont{stmry}{bold}{U}{stmry}{m}{n}

\usepackage{wrapfig}
\usepackage{tikz}
\usetikzlibrary{positioning,calc,arrows.meta}
\usepackage{enumitem}
\usepackage[mathlines]{lineno}
\newif\ifreviewversion
\reviewversionfalse
\ifreviewversion\linenumbers\fi
\newcommand{\citep}{\cite}
\usepackage{xcolor}
\usepackage{todonotes}
\usepackage{CJKutf8} 

\definecolor{linkgreen}{RGB}{0,110,65}
\usepackage[hypertexnames=false,colorlinks,citecolor=linkgreen,linkcolor=linkgreen,urlcolor=linkgreen]{hyperref}
\usepackage{bookmark}
\hypersetup{pdftitle={An Exponential Succinctness Gap between Three-Variable Logic and the Calculus of Relations},pdfauthor={Yuya Uezato and Yoshiki Nakamura}}
\newcommand{\doi}[1]{\href{https://doi.org/#1}{\nolinkurl{doi:#1}}}
\allowdisplaybreaks[2]
\newtheorem{theorem}{Theorem}[section]
\newtheorem{lemma}[theorem]{Lemma}
\newtheorem{proposition}[theorem]{Proposition}
\newtheorem*{proposition*}{Proposition}
\newtheorem{corollary}[theorem]{Corollary}
\theoremstyle{definition}

\theoremstyle{remark}
\newtheorem{remark}[theorem]{Remark}
\newcommand{\set}[1]{\{ #1 \}}

\newcommand{\customparag}[1]{{\noindent{}\textit{#1}}}
\newcommand{\customparagB}[1]{{\noindent{}\textbf{#1}}}
\newcommand{\CoR}{\ensuremath{\mathrm{CoR}}}
\newcommand{\FOtwo}{\ensuremath{\mathrm{FO}^{2}}}
\newcommand{\FOthree}{\ensuremath{\mathrm{FO}^{3}}}
\newcommand{\FOfour}{\ensuremath{\mathrm{FO}^{4}}}

\newcommand{\zero}{\mathbf{0}}
\newcommand{\one}{\mathbf{1}}
\newcommand{\id}{\mathbf{id}}
\newcommand{\cost}[1]{\#_{;}(#1)}
\newcommand{\gcost}[1]{g_{;}(#1)}
\newcommand{\tcost}{\operatorname{cost}}
\newcommand{\relatoms}[1]{\#_{\mathrm{rel}}(#1)}
\newcommand{\mvcount}[1]{\#_{\mathrm{mv}}(#1)}
\newcommand{\mmcount}[1]{\#_{\mathrm{mm}}(#1)}

\newcommand{\bool}{\mathbb B}
\newcommand{\agg}[1]{\mathbin{;_{#1}}}
\newcommand{\sem}[2]{\llbracket #1\rrbracket_{#2}}
\newcommand{\hub}{\hslash}
\newcommand{\pred}[2]{#1^{(#2)}}
\newcommand{\diagc}{\id}
\newcommand{\defn}[1]{{\boldmath\textbf{\textit{#1}}}}

\def\cordom{\mathbb{D}}

\title{An Exponential Succinctness Gap between\protect\\Three-Variable Logic and the Calculus of Relations}

\author{Yuya Uezato \\
CyberAgent, Inc. \\
National Institute of Informatics \\
uezato\_yuya@cyberagent.co.jp \\
uezato@nii.ac.jp
\and 
Yoshiki Nakamura \\
Chiba University \\
nakamura.yoshiki.ny@gmail.com
}

\date{}
\begin{document}

\hypersetup{pageanchor=false}
\begin{titlepage}
\maketitle
\thispagestyle{empty}

\begin{abstract}
Three-variable first-order logic ($\mathrm{FO}^3$) and the calculus of relations (CoR)
define the same binary queries, an
equivalence that goes back to Tarski's work in the 1940s. The classical
translation from $\mathrm{FO}^3$ to CoR is single-exponential, but it
has remained open whether this blowup is necessary. We prove that it is,
even over finite structures and even for circuits that share subterms.
Our proof introduces a new technique: we modify a single
finite structure in exponentially many local ways, each
changing the answer of a short $\mathrm{FO}^3$ formula $\varphi$.
In any CoR circuit $C$ equivalent to $\varphi$,
only compositions of binary relations can notice these modifications, and each composition can
notice at most two; thus, $C$ requires exponentially many compositions.
The same technique gives an exponential succinctness
gap between $\mathrm{FO}^3$ and $\mathrm{FO}^2$,
and exponential lower bounds for matrix query languages over the Boolean and tropical semirings, 
where the analogous question was also open.
\end{abstract}
\end{titlepage}
\setcounter{tocdepth}{1}
\tableofcontents\clearpage
\setcounter{page}{1}
\hypersetup{pageanchor=true}

\section{Introduction}

A classical result of Tarski from the 1940s states that,
over binary relational signatures,
\defn{three-variable first-order logic} (\FOthree) and the \defn{calculus of relations} (\CoR; a variable-free algebra of binary relations)~\citep{tarski1941} express the same binary queries,
i.e., queries whose output is a binary relation~\citep{tarski-givant1987,givant2006}.
The known translations between them have asymmetric size bounds:
\[
\varphi(x,y)\in\mathrm{FO}^3
\ \smash{\xRightarrow{\;|T|\,=\,2^{O(|\varphi|)}\;}}\ 
T \in\CoR,
\qquad
T \in\mathrm{CoR}
\ \smash{\xRightarrow{\;|\varphi|\,=\,O(|T|)\;}}\ 
\varphi(x,y)\in\FOthree.
\]
Whether this exponential blowup is unavoidable has remained open~\citep[Remark~4.5]{FletcherEtAl2015},
\citep[\S8]{nakamura2022}, \citep[Question~3]{nakamura-uezato2027}.
\textbf{We show that the exponential blowup from \FOthree\ to \CoR\ is unavoidable}.
Before stating our result precisely in \S\ref{sec:main-results},
we review the importance of \FOthree\ and this translation in logic and computation.

\smallskip
\customparag{Expressive power of \FOthree.}
\FOthree\ is singular among the finite-variable fragments $\mathrm{FO}^{\ell}$ (FO with at most $\ell$ reusable variables).
\FOthree\ can express the transitivity law
$\forall x\,\forall y\,\forall z\,(E(x,y)\land E(y,z)\to E(x,z))$,
and thereby, together with irreflexivity and the absence of maximal elements, force infinite models.
In contrast, \FOtwo{} cannot express transitivity and has the finite model property~\citep{mortimer1975}.
Over words, the variable hierarchy of FO collapses at three: every FO-definable property of words can be expressed in \FOthree~\citep{immerman-kozen1989}, whereas \FOtwo\ is strictly less expressive~\citep{etessami-vardi-wilke2002}.
Remarkably, three variables even suffice as a foundation for mathematics: Tarski and Givant showed that set theories such as ZFC can be formalized by \CoR\ equations~\citep{tarski-givant1987,givant2006,maddux2021}.

This expressive strength is also reflected in computational complexity.
For \FOthree, satisfiability and finite satisfiability are undecidable,
already for the $\forall\exists\forall$ prenex class~\citep{kahr-moore-wang1962,borger-gradel-gurevich1997}.
Separately, both problems are undecidable over a single binary relation~\citep{nakamura2019}.
For \FOtwo, both are decidable~\citep{mortimer1975}.
%
For a fixed sentence, FO model counting can be $\textsc{\#P}_1$-complete in \FOthree~\citep{beame2015} but takes polynomial time in \FOtwo~\citep{vandenbroeck-meert-darwiche2014}.
Asser's longstanding question of whether first-order spectra (i.e., the sets of finite-model sizes) are closed under complement~\citep{asser1955} is equivalent to the open problem $\textsc{NE} = \textsc{coNE}$~\citep{jones-selman1974}.
Resolving it for \FOthree\ over a single binary relation would already settle the general problem~\citep{kopczynski-tan2015,kopczynski-tan2018}.

\smallskip
\customparag{Algorithmic power of the \FOthree-to-\CoR\ translation.}
The \FOthree-to-\CoR\ translation also has an algorithmic consequence in model checking.
For fixed $\ell$, a sentence $\varphi \in \mathrm{FO}^{\ell}$ can be evaluated on an $n$-element structure in $O(|\varphi|n^\ell)$ time~\citep{vardi1995}.
For $\ell\ge4$, this exponent is conditionally optimal under the 3-uniform hyperclique hypothesis~\citep{bringmann-slusallek2021}, \citep[Corollary~32]{nakamura-uezato2027}.
For \FOthree, however, the cubic barrier can be broken: \FOthree\ model checking takes $2^{O(|\varphi|)}n^\omega$ time~\citep[Theorem~10]{nakamura-uezato2027}, where $\omega < 2.3714$ is the matrix multiplication exponent~\citep{alman2025}.
\textbf{The key to subcubic \FOthree\ model checking is the \FOthree-to-\CoR\ translation}.
We translate $\varphi$ into a \CoR\ term $T$ and evaluate $T$ in $O(|T|n^\omega)$ time,
since composition is Boolean matrix multiplication and every other \CoR\ operation takes $O(n^2)$ time. Thus, a polynomial-size translation would give compact algebraic evaluation plans executable in $\mathrm{poly}(|\varphi|)n^\omega$ time once compiled.
Our result rules out such a translation: some \FOthree\ queries require exponentially large \CoR\ terms and circuits.


\subsection{Main results}\label{sec:main-results}

Over a nonempty domain $\cordom$,
\defn{CoR} \defn{terms} are built from binary relation symbols,
the constants $\mathbf{0} \coloneqq \varnothing$, 
$\mathbf{1} \coloneqq \cordom^2$,
and $\id \coloneqq \set{(u,u) : u\in \cordom}$,
using union $X \cup Y$, intersection $X \cap Y$, 
complement $\overline X \coloneqq \cordom^2\setminus X$ (also written $\neg X$),
converse $X^\smile \coloneqq \set{(v,u):(u,v)\in X}$, and composition:
\[
X;Y \coloneqq \set{(u,v):\text{some }w\in\cordom\text{ satisfies }
 (u,w)\in X \text{ and } (w,v)\in Y}.
\]
Composition implements the $\exists$-elimination
$\exists z\,(P(x,z)\land Q(z,y)) \mapsto P;Q$.
A vertex $z$ satisfying $P(x,z)\land Q(z,y)$ is a \emph{witness} for $(x,y)$ in $P ; Q$.

\smallskip
\customparagB{Notation (see \S\ref{sec:aggregation} for full definitions).}
For a formula $\varphi$ and a \CoR\ term $T$,
$|\varphi|$ and $|T|$ denote the numbers of syntax-tree nodes, and $\cost T$ the number of compositions in $T$.
A \defn{\CoR\ circuit} $C$ is a DAG that may share subterms;
$|C|$ and $\gcost C$ count its distinct gates and composition gates.
We write $\Psi\equiv\Gamma$
when $\Psi$ and $\Gamma$,
each a \CoR\ term, a \CoR\ circuit, or an FO formula
with at most two free variables,
define the same binary query over all nonempty structures of their common input signature.
We write $\mathcal M\models\Psi(a,b)$ when $\Psi$ holds at $(a,b)$ in $\mathcal M$, and $\mathcal M\not\models\Psi(a,b)$ otherwise.

\smallskip
\customparagB{Our hard family and its translation into \CoR\ (upper bound).}
For $k \geq 1$,
we use the following single-quantifier positive \FOthree\ formula over $2k$ binary relation symbols $P_1,\ldots,P_k,Q_1,\ldots,Q_k$:
\[
 G_k(x,y)\coloneqq\exists z\,\smash{\bigwedge_{i=1}^k}
 \Bigl[\bigl(P_i(x,y)\land Q_i(x,z)\bigr)\lor\bigl(Q_i(x,y)\land P_i(x,z)\bigr)\Bigr],
\]
where $|G_k|=O(k)$.
$G_k$ asks whether some $z$ satisfies all $k$ conditions:
for each $i$, either $P_i(x,y)$ and $Q_i(x,z)$ hold, or $Q_i(x,y)$ and $P_i(x,z)$ hold.
The reason for this particular pairing of $P_i$ and $Q_i$
will become clear in the lower-bound construction of \S\ref{ov:fold}.

The natural way to translate $G_k$ into \CoR\ is to expand it into disjunctive normal form (DNF),
choosing one of the two disjuncts for each $i$.
Let $[k]\coloneqq\set{1,\ldots,k}$.
For $A\subseteq[k]$, let $A^c \coloneqq [k]\setminus A$ and
$H_A \coloneqq (\bigcap_{i\in A}P_i)\cap(\bigcap_{i\notin A}Q_i)$,
where an empty intersection is $\mathbf{1}$.
\begin{align*}
G_k(x,y)
&\equiv \exists z\!\!\bigvee_{A\subseteq[k]}
\Bigl[
\bigwedge_{\mathclap{i\in A}}(P_i(x,y)\land Q_i(x,z))
\land
\bigwedge_{\mathclap{i\notin A}}(Q_i(x,y)\land P_i(x,z))
\Bigr]
\equiv \exists z\bigvee_{\mathclap{A\subseteq[k]}}
\bigl[H_A(x,y)\land H_{A^c}(x,z)\bigr]
\\
&\equiv \bigvee_{A\subseteq[k]}
\bigl[H_A(x,y)\land\exists z\,H_{A^c}(x,z)\bigr]
\equiv \Bigl(\bigcup_{A\subseteq[k]}
\bigl[H_A\cap(H_{A^c};\mathbf{1})\bigr]\Bigr)(x,y).
\end{align*}
The last step uses row projection:
$(X;\mathbf{1})(x,y)\equiv\exists z\,X(x,z)$.
The DNF gives a \CoR\ term $T_k$ with $\cost{T_k}=2^k$ and $|T_k|=O(k2^k)$.
By sharing partial intersections among the $H_A$,
it also gives a \CoR\ circuit $C_k$ with $\gcost{C_k} =2^k$ and $|C_k|=O(2^k)$.

The same DNF-based quantifier elimination extends to arbitrary \FOthree\ formulas.
\begin{proposition*}[Classical upper bound; {\citep[Theorem~3.1(1)]{nakamura2022}}; Appendix~\ref{app:classical-translation}]
Every \FOthree\ formula $\varphi$ over binary relations with at most two
free variables has an equivalent \CoR\ term, and hence a circuit, of total
size $2^{O(|\varphi|)}$.
\end{proposition*}

\smallskip
\customparagB{Main result: lower bound.}
Could more sophisticated \CoR\ circuits compress $G_k$?
Surprisingly, the straightforward translation is already optimal within a factor of two in the number of compositions.

\begin{restatable}[Exponential lower bounds\proofref{sec:main-proof}]{theorem}{mainbounds}\label{thm:main}\label{main:G}

For every $k\ge 1$, every \CoR\ circuit $C$ such that
$C\equiv G_k$ and every \CoR\ term $T$ such that
$T\equiv G_k$ satisfy
$\gcost{C} \ge 2^{k-1}$ and $\cost{T} \ge 2^{k-1}$.
\end{restatable}
These bounds count only compositions; other operations are unrestricted and uncharged in these counts.
Since $|G_k|=O(k)$, the theorem establishes an unconditional exponential succinctness gap from \FOthree\ to \CoR, even when the target is a circuit rather than a term.
Together with the classical single-exponential upper bound, this settles the worst-case size of the \FOthree-to-\CoR\ translation.

The gap is not an artifact of free variables: the \FOthree\ sentence
$G_k^\forall\coloneqq\forall x\,\forall y\,(x=y\lor G_k(x,y))$
has size $O(k)$, yet
every equivalent \CoR\ equation $T = S$
(asserting equality of the two denoted relations in the input structure)
requires at least $2^{k-1}$ compositions (Corollary~\ref{cor:equations}).
Thus, even for model checking, translating \FOthree\ sentences into \CoR\ equations incurs an unavoidable exponential blowup in size.

\smallskip
\customparagB{Domain size.}
Can \CoR\ circuits specialized to a fixed domain size avoid the exponential lower bound?
The following proposition rules this out for every sufficiently large size.

\begin{restatable}[Domain size\proofref{sec:domain-size}]{proposition}{domainsize}\label{prop:domain-size}
For some constant $c>0$, for every $k\ge1$
and $n\ge ck4^k$, every \CoR\ circuit $C$
computing $G_k$ on all $n$-element structures
satisfies $\gcost C \geq 2^{k-1}$.
\end{restatable}

In contrast, for $n<2^k/(8k)$, Proposition~\ref{prop:small-domains} gives
circuits with only $4kn\,(< 2^{k-1})$ compositions
that compute $G_k$ on all structures with at most $n$ elements.
Whether the domain-size bound $ck4^k$ is
asymptotically optimal remains open.

\smallskip
\customparagB{Over a single binary relation.}
When the vocabulary is restricted to a single binary relation,
we obtain tight \textbf{subexponential} bounds.
The upper bound uses a fixed-vocabulary bound on the number of distinct Boolean inputs (atoms and maximal quantified subformulas) in each quantifier body,
and shares their translations across DNF clauses.
For the lower bound, we encode the relations of $G_k$ into one binary relation and recover each with an $O(\log k)$-size decoder.

\begin{theorem}[Subexponential bounds\proofref{sec:fixed-vocabulary}]
\label{thm:fixed-vocabulary}
Over a single binary relation,
every \FOthree\ query $\varphi(x,y)$ of syntax-tree size $m$ has an equivalent \CoR\ circuit $C$ with $|C|\le 2^{O(m/\log m)}$.

Conversely, for every sufficiently large $m$, there is an \FOthree\ query $\psi$ of size at most $m$ such that every equivalent \CoR\ circuit $C$ satisfies $\gcost{C}\ge 2^{\Omega(m/\log m)}$.


\end{theorem}

Whether a corresponding subexponential upper bound holds for \CoR\ terms remains open.

\smallskip
\customparag{Model-checking consequence.}
The constructive upper bound yields a
$2^{O(m/\log m)}n^\omega$-time algorithm for model checking \FOthree\ sentences of size $m$ on $n$-vertex directed graphs (see~\S\ref{sec:fixed-model-checking}).
This improves the previous $2^{O(m)}n^\omega$-time algorithm~\citep[Theorem~10]{nakamura-uezato2027}.
Moreover,
under the \emph{low-dimensional} 3-Orthogonal Vectors (3-OV) hypothesis,
no algorithm runs in $2^{o(m/\log m)}n^{3-\varepsilon}$ time for any fixed $\varepsilon>0$~\citep[Theorem~72]{nakamura-uezato2027}.

\subsection{Limits of the earlier relational fooling-set argument}\label{sec:natural-approaches}
A natural approach to lower bounds for full \CoR\ is to extend the earlier
exponential circuit lower bound~\citep[Theorems~6.1 and~6.12]{nakamura2022}
for the \defn{existential calculus of relations} $\Sigma_1\CoR$~\citep{nakamura2022,nakamura2023}, which is as expressive as existential \FOthree~\citep[Corollary~3.14]{nakamura2022}.
Here $\Sigma_1\CoR$ allows complement only on input relations and $\id$:
$\overline R;\one$ is permitted, but $\overline{R;\one}$ is not.
Although that lower bound was proved with formula-size games in the
style of Adler and Immerman~\citep{adler-immerman2003}, its proof can be
viewed as a relational fooling-set argument, which we recall first.

The \defn{fooling-set argument}~\citep{yao1979,lipton-sedgewick1981,kushilevitz-nisan1997,rao-yehudayoff2020}
is a standard technique for proving lower bounds in communication complexity.
For a Boolean function $f\colon A\times B\to\{0,1\}$, a
\defn{fooling set} is a set $F\subseteq f^{-1}(1)$ such that any two
distinct $(a,b),(a',b')\in F$ have a crossed input $(a,b')$ or
$(a',b)$ of value $0$. A \defn{1\nobreakdash-rectangle} is a product
$A'\times B'\subseteq f^{-1}(1)$. If a rectangle contains two distinct inputs of $F$,
it also contains their crossed inputs, one of which has value $0$; hence each
1\nobreakdash-rectangle contains at most one input of $F$, and covering
$f^{-1}(1)$ needs at least $|F|$ 1\nobreakdash-rectangles.

Consider Nakamura's query of size $O(k)$,
$\varphi_k(x,y)\coloneqq\exists z\,\bigwedge_{i=1}^k
(R_i(x,z)\lor R_i(z,y))$~\citep{nakamura2022}, where $k\ge2$.
For fixed $x$, $y$, and $z$,
let $I \coloneqq \set{i \in [k] : R_i(x,z)}$ and
$J \coloneqq \set{i \in [k] : R_i(z,y)}$.
Then $z$ is a witness precisely when $I \cup J = [k]$.
The Boolean function $f_k$ defined by $f_k(I,J) = 1 \Leftrightarrow I \cup J = [k]$
has a fooling set
$F_k := \set{(I,I^c) : \varnothing \subsetneq I
\subsetneq [k]}$
of size $2^k-2$.

In~\citep[\S5]{nakamura2022}, each input $(I,I^c)$ is realized by the
three-element structure $\mathcal S_I$ consisting of a witness path
$u_I\xrightarrow{\,I\,}w_I\xrightarrow{\,I^c\,}v_I$, so $\mathcal S_I\models\varphi_k(u_I,v_I)$.
An edge labeled $I$ belongs to exactly the relations $R_i$ with $i \in I$; all other entries, including loops, are false.

\begin{wrapfigure}{l}{0.3\linewidth}
  \centering
  \begin{tikzpicture}[
    >=Stealth,
    every node/.style={font=\small,inner sep=1pt}
  ]
    \node (u) at (0, 1.2) {$u^*$};
    \node (a) at (0, 0) {$a_{I,J}$};
    \node (c) at (1.5, 0.6) {$c_{I,J}$};
    \node (b) at (3.0, 1.2) {$b_{I,J}$};
    \node (v) at (3.0, 0) {$v^*$};

    \draw[->] (u) --
      node[above,sloped] {$I$} (c);
    \draw[->] (c) --
      node[above,sloped] {$I^c$} (b);
    \draw[->,dashed] (a) --
      node[below,sloped] {$J$} (c);
    \draw[->,dashed] (c) --
      node[below,sloped] {$J^c$} (v);
  \end{tikzpicture}
  \caption{\small $S_I\rtimes S_J$.
  Solid edges form $S_I$; dashed edges form $S_J$.}
  \label{fig:crossed-gadget}
\end{wrapfigure}
For each pair $(I,I^c),(J,J^c)\in F_k$ with $f_k(I,J^c)=0$,
let the gadget $S_I\rtimes S_J$ be the five-vertex structure in Figure~\ref{fig:crossed-gadget},
obtained by identifying the middle vertices of
copies of $S_I$ and $S_J$ as $c_{I,J}$,
with the other endpoints labeled as shown.
Let $\mathcal N$ be the union of these gadgets,
identifying all copies of $u^*$ and, separately,
all copies of $v^*$, while keeping all other
vertices distinct.
By construction, $S_I$ and $S_J$ \emph{embed} into
$\mathcal N$, mapping $(u_I,w_I)$ to
$(u^*,c_{I,J})$ and $(w_J,v_J)$ to
$(c_{I,J},v^*)$, respectively.
An embedding $\iota : \mathcal A \hookrightarrow \mathcal B$ is an injective map preserving and reflecting every input relation~\citep{hodges1993,marker2002}.
Every two-edge path from $u^*$ to $v^*$ runs through some $c_{I,J}$ with labels $I$ and $J^c$,
where $f_k(I,J^c) = 0$.
Hence $\mathcal N\not\models\varphi_k(u^*,v^*)$.

$\Sigma_1\CoR$ terms are \emph{preserved under
embeddings}:
\[
\tag{$\diamond$}\label{eq:mono}
\text{For every term}\ V \in \Sigma_1\CoR\ 
\text{and embedding}\ \iota\colon\mathcal A\hookrightarrow\mathcal B,\ 
\text{if}\ 
\mathcal A\models V(u,v)\ 
\text{then}\ 
\mathcal B\models V(\iota(u),\iota(v)).
\]
Now let $T$ be a $\Sigma_1\CoR$ term equivalent to $\varphi_k$.
For a composition $X;Y$ in $T$, let
$A_X\coloneqq\{I:\mathcal S_I\models X(u_I,w_I)\}$ and
$B_Y\coloneqq\{J^c:\mathcal S_J\models Y(w_J,v_J)\}$.
If $\mathcal N\not\models(X;Y)(u^*,v^*)$, then \textbf{$A_X\times B_Y$ is a 1\nobreakdash-rectangle}:
if some $(I,J^c)\in A_X\times B_Y$ had value $0$, then \eqref{eq:mono} applied to
the embeddings of $\mathcal S_I$ and $\mathcal S_J$ into $\mathcal N$ would give
$\mathcal N\models(X;Y)(u^*,v^*)$.
Hence $A_X\times B_Y$ contains at most one input of $F_k$.
For each $I$, $T$ contains a composition $X;Y$ with
$\mathcal S_I\models X(u_I,w_I)\land Y(w_I,v_I)$ and
$\mathcal N\not\models(X;Y)(u^*,v^*)$
(by induction on $T$, with converse pushed to the leaves and \eqref{eq:mono} handling endpoint witnesses).
Its rectangle contains $(I,I^c)$.
Thus $T$ has at least $|F_k|$ compositions.
The exponential separation also holds for circuits~\citep[\S6.3]{nakamura2022}.

\smallskip
\customparagB{What complement breaks.}
The fooling-set argument relies on the embedding-preservation property \eqref{eq:mono},
which fails in full \CoR\ due to unrestricted complement.
In fact, a simple \CoR\ term distinguishes  $\mathcal N$ from every $\mathcal S_I$.
Let $E\coloneqq\bigcup_{i=1}^k(R_i\cup R_i^\smile)$
and consider
$U\coloneqq\overline{(E\cap(E;\overline{\id}));\one}$.
The term $U$ holds at $(x,y)$ precisely when
$x$ has at most one $E$-neighbor.
Thus $U$ holds at $(u_I, w_I)$ in every
$\mathcal S_I$ but not at its image $(u^*, c_{I, J})$ in
$\mathcal N$
because $u^*$ has distinct $E$-neighbors
$c_{I,J}, c_{I',J'}, \ldots$.

\subsection{Technical overview}\label{sec:approach}\label{sec:overview}\label{ov:structure}

Rather than extending the fooling-set argument in~\S\ref{sec:natural-approaches}, we abandon witness crossing and embeddings altogether. Instead, we compare a reference structure with its local modifications using color-based predictions that commute with complement.

\smallskip
\customparagB{Our technique: the two-witness property.}
We use a complete undirected graph $\mathbb{G} = (V,\chi)$ whose edges are colored by the $2^k$ subsets of $[k]$.
We write $\chi(u,v)=\chi(v,u)$ for the color of the edge between distinct vertices $u$ and $v$,
and use $\lambda,\mu,\nu$ for colors.
The key property of $\mathbb{G}$ is the \defn{two-witness property}: for all distinct
$u,v\in V$ and all colors $\mu,\nu$,
\[
\text{at least two vertices $z\notin\{u,v\}$ satisfy $\chi(u,z)=\mu$ and
 $\chi(z,v)=\nu$.}
 \qquad
 \begin{tikzpicture}[baseline={([yshift=-0.55ex]u.center)},
   every node/.style={font=\footnotesize},vtx/.style={circle,draw,inner sep=1.2pt}]
 \node[vtx,label={[font=\small]left:$u$}] (u) at (0,0) {};
 \node[vtx,label={[font=\small]right:$v$}] (v) at (2.4,0) {};
 \node[vtx,label={[label distance=-1pt]above:$z$}] (z) at (1.2,0.46) {};
 \node[vtx,label={[label distance=-1pt]below:$z'$}] (z2) at (1.2,-0.46) {};
 \draw (u) -- node[above left=-2.5pt] {$\mu$} (z);
 \draw (z) -- node[above right=-2.5pt] {$\nu$} (v);
 \draw (u) -- node[below left=-2.5pt] {$\mu$} (z2);
 \draw (z2) -- node[below right=-2.5pt] {$\nu$} (v);
 \end{tikzpicture}
\]
We call such vertices $z$ \defn{witnesses}.
Such colorings exist at every sufficiently large size by a random-coloring argument (Lemma~\ref{lem:palette}), based on constructions of Jipsen, Maddux, and Tuza~\citep{jipsen-maddux-tuza1995}.

\paragraph{Roadmap.}
Using the two-witness coloring, we construct a
reference structure $\mathcal M_0$ and $2^k$ folds,
each changing the value of $G_k$ at a pair that the fold does not recolor (\S\ref{ov:fold}).
We distinguish gates that \defn{notice} a fold
from those that merely propagate its effects (\S\ref{ov:types}).
Boolean operations and converse notice no folds,
whereas each composition notices at most two (\S\ref{ov:twocolors}).
Since every fold must be noticed,
any \CoR\ circuit for $G_k$ requires at least $2^{k-1}$ composition gates.

\subsubsection{Folding the reference structure}\label{ov:outline}\label{ov:fold}

We use the two-witness coloring $\chi$ to define an input structure $\mathcal M_0$ for $G_k$.
On an edge of color $\mu$, we interpret $P_i$ as true exactly when $i\in\mu$, and $Q_i$ exactly when $i\notin\mu$.
No input relation holds on loops; set $\chi(u,u)\coloneqq\id$, a label distinct from all edge colors.
Consequently, for distinct $u,v$, $G_k(u,v)$ holds exactly when $u$
has an incident edge of color $\chi(u,v)^c$:
\[
\mathcal M_0 \models G_k(u,v)\iff
\exists z \notin \set{u, v}.\ \chi(u,z)=\chi(u,v)^c.
\]
Indeed, the $i$-th conjunct holds iff exactly one of $\chi(u,v)$ and $\chi(u,z)$ contains $i$,
and loops satisfy no input relation.
By the two-witness property, every vertex
has edges of every color, so $G_k$ holds
at every non-loop pair of $\mathcal M_0$.

Fix a vertex $\hub \in V$, the \defn{hub}, and for each color $\mu$ choose $v_\mu$ with $\chi(\hub,v_\mu)=\mu$.
For each color $\lambda$, the \defn{fold} of $\lambda$ recolors every $\lambda$-colored edge incident to $\hub$ to $\lambda^c$,
leaving all other edges unchanged.
Let $\chi_\lambda$ be the resulting coloring and $\mathcal M_\lambda$ the induced input structure.
In $\mathcal M_\lambda$, the hub has no edge of color $\lambda$, but the pair $(\hub,v_{\lambda^c})$ remains unchanged since $\lambda^c \neq \lambda$ (Figure~\ref{fig:fold}).

\begin{figure}[H]
\begin{minipage}[c]{0.27\linewidth}
\setlength{\abovecaptionskip}{0pt}
\caption[The fold of $\lambda$]{\raggedright The fold of $\lambda$. Only edges at the hub $\hub$ are drawn.}
\label{fig:fold}
\end{minipage}\hfill
\begin{minipage}[c]{0.71\linewidth}
\centering
\begin{tikzpicture}[scale=0.85,every node/.style={font=\small},
  vtx/.style={circle,draw,inner sep=2pt},
  folded/.style={draw=red!65!black,line width=1.2pt}
]
\foreach \etype/\xs/\col/\val/\name in
  {{}/0/\lambda/1/{reference structure $\mathcal M_0$},   folded/7.6/\lambda^{\color{red}c}/{\color{red}0}/{folded structure $\mathcal M_\lambda$}}{
\begin{scope}[xshift=\xs cm]
\node at (-0.2,2.75) {\name};
\node[vtx,label={[label distance=2pt]below:$\hub$}] (h) at (0,0) {};
\node[vtx] (a1) at (-2.4,0.8) {};
\node[vtx] (a2) at (-1.1,2.0) {};
\node[vtx] (x) at (0.2,2.05) {};
\node[vtx,label=right:$v_{\lambda^c}$] (v) at (2.0,1.4) {};
\draw[\etype] (h)--node[below=1pt,pos=0.6] {$\col$} (a1);
\draw[\etype] (h)--node[left=2pt,pos=0.8] {$\col$} (a2);
\draw (h)--node[right=1pt,pos=0.7] {$\mu\ne\lambda$} (x);
\draw[] (h)--node[below right=0pt,pos=0.62] {$\lambda^c$} (v);
\node at (0,-1.1) {$G_k(\hub,v_{\lambda^c})=\val$};
\end{scope}}
\draw[->,thick] (3.0,0.8)--node[above]{
\begin{tabular}{c}
fold of $\lambda$; \\
{\footnotesize recolor $\lambda$} \\
{\footnotesize to $\lambda^c$}
\end{tabular}} (4.8,0.8);
\end{tikzpicture}
\end{minipage}
\end{figure}

The query $G_k(\hub,v_{\lambda^c})$ requires a $\lambda$-colored edge at $\hub$,
but the fold recolors all such edges.
Hence
\[
\chi_\lambda(\hub,v_{\lambda^c})
 =\chi(\hub,v_{\lambda^c})=\lambda^c,\quad
\mathcal M_0\models G_k(\hub,v_{\lambda^c}),\quad
\mathcal M_\lambda\not\models G_k(\hub,v_{\lambda^c}).
\]
\textbf{The pairing of $P_i$ and $Q_i$ in $G_k$ is designed so that folding $\lambda$ changes the answer at $(\hub,v_{\lambda^c})$}.
We call $(\mathcal M_0,\mathcal M_\lambda,(\hub,v_{\lambda^c}))$ the \defn{test} for $\lambda$.
Thus $G_k$ yields $2^k$ tests, one for each color.

\subsubsection{Circuits see only colors}\label{ov:types}

\paragraph{Color-determined relations.}
A relation $X\subseteq V^2$ is
\defn{color-determined} if
$X(u,v)=X(u',v')$ whenever $\chi(u,v)=\chi(u',v')$.
Here loops carry the separate label $\id$.
Equivalently, $X$ is a union of edge-color
classes, possibly together with the diagonal.
Write $X[\mu]$ for its common value on edges
of color $\mu$,
and $X[\id]$ for its value on loops.

Every gate of an arbitrary \CoR\ circuit computes a color-determined relation on $\mathcal M_0$ (Lemma~\ref{lem:tables}).
The proof is by induction over the gates.
Input relations and constants are
color-determined, while Boolean operations and
converse preserve this property.
The only nontrivial step is composition $X;Y$.

For distinct $u,v$, the two-witness property realizes every pair of colors $\mu,\nu$ along a path $u\xrightarrow{\mu}z\xrightarrow{\nu}v$ with $z\notin\set{u,v}$.
The case $z\in\set{u,v}$ depends only on $\chi(u,v)$ and the diagonal values of $X,Y$.
Hence $(X;Y)(u,v)$ depends only on $\chi(u,v)$.
For a loop $(u,u)$, both edges of a path $u-z-u$ have the same color, and every color occurs at every vertex. 
Together with the case $z=u$, this makes $(X;Y)(u,u)$ independent of $u$, completing the induction.

\paragraph{Predicting and noticing a fold.}
A color-determined relation $X$ assigns a value $X[\mu]$ to each color $\mu$.
To predict the effect of a fold, we keep these values fixed and read them at the new edge colors:
\[
\defn{(Prediction)}\quad
X^{(\lambda)}(u,v) \coloneqq X[\chi_\lambda(u,v)].
\]
Thus a recolored edge reads $X[\lambda^c]$ rather than $X[\lambda]$; all other pairs retain their values.
This predicts what a gate would compute if its color-dependent behavior on $\mathcal M_0$ remained unchanged after the fold;
its actual value on $\mathcal M_\lambda$ may differ.

Input relations and constants behave exactly as predicted.
Moreover, Boolean operations and converse commute with prediction
(Proposition~\ref{prop:prediction-commutation}); in particular,
$\overline{X^{(\lambda)}}=(\overline X)^{(\lambda)}$.
Unlike the preservation under embeddings~\eqref{eq:mono} in \S\ref{sec:natural-approaches},
commutation with prediction holds even for complement.

An operation \defn{notices} a fold if applying it to the predictions of its color-determined operands differs from predicting its result.

Figure~\ref{fig:prediction} illustrates noticing at the unchanged edge $e=(\hub,v_\mu)$,
where $X$ and $Y$ select the $\lambda$- and $\mu$-colored edges, respectively ($\lambda\ne\mu$; loops excluded).
In the table, \textsf{Predicted} means applying
the fold prediction to an operation's output
on $\mathcal M_0$, whereas \textsf{Recomputed}
means applying the operation to the predicted
operands. For composition, these are
$(X;Y)^{(\lambda)}(e)$ and
$(X^{(\lambda)};Y^{(\lambda)})(e)$, respectively.

\begin{figure}[H]
  \centering
  \begin{minipage}[c]{0.58\linewidth}
    \centering
\begin{tikzpicture}[
      >=Latex,
      vertex/.style={circle,draw,fill=white,inner sep=0pt,minimum size=3.2pt},
      edge/.style={line width=.5pt},
      read/.style={line width=1.05pt},
      folded/.style={draw=red!65!black,line width=1.2pt},
      every node/.style={font=\scriptsize},
    ]
      \coordinate (h0) at (0,0);
      \coordinate (v0) at (2.65,0);
      \coordinate (z0) at (1.325,0.98);
      \coordinate (w0) at (1.325,-0.98);
      \draw[edge] (h0)--node[below,pos=.54] {$\mu$} (v0);
      \draw[edge] (h0)--node[above left] {$\lambda$} (z0);
      \draw[edge] (h0)--node[below left] {$\lambda$} (w0);
      \draw[edge] (z0)--node[above right] {$\mu$} (v0);
      \draw[edge] (w0)--node[below right] {$\mu$} (v0);
      \node[vertex,label=left:{$\hub$}] at (h0) {};
      \node[vertex,label=right:{$v_\mu$}] at (v0) {};
      \node[vertex,label=above:{$z$}] at (z0) {};
      \node[vertex,label=below:{$z'$}] at (w0) {};
      \node[font=\small] at (.1, 1.2) {$\mathcal M_0$};
      \draw[-{Latex[length=2mm]},line width=.7pt]
        (3.4,0)--node[above] {fold $\lambda$} (4.85,0);
      \begin{scope}[xshift=5.30cm]
        \coordinate (h1) at (0,0);
        \coordinate (v1) at (2.65,0);
        \coordinate (z1) at (1.325,0.98);
        \coordinate (w1) at (1.325,-0.98);
        \draw[edge] (h1)--node[below,pos=.54] {$\mu$} (v1);
        \draw[folded] (h1)--node[above left,black] {$\lambda^c$} (z1);
        \draw[folded] (h1)--node[below left,black] {$\lambda^c$} (w1);
        \draw[edge] (z1)--node[above right] {$\mu$} (v1);
        \draw[edge] (w1)--node[below right] {$\mu$} (v1);
        \node[vertex,label=left:{$\hub$}] at (h1) {};
        \node[vertex,label=right:{$v_\mu$}] at (v1) {};
        \node[vertex,label=above:{$z$}] at (z1) {};
        \node[vertex,label=below:{$z'$}] at (w1) {};
        \node[font=\small] at (.1, 1.2) {$\mathcal M_\lambda$};
      \end{scope}
    \end{tikzpicture}
  \end{minipage}\hfill%
  \begin{minipage}[c]{0.40\linewidth}
    \centering
    {\footnotesize
    \renewcommand{\arraystretch}{1.16}%
    \begin{tabular}{@{}lccc@{}}
      \toprule
      Op. & \textsf{Predicted} & \textsf{Recomputed} & Notice? \\
      \midrule
      $X;Y$         & $1$ & $0$ & Yes \\
      $X\cup Y$     & $1$ & $1$ & No \\
      $\overline X$ & $1$ & $1$ & No \\
      \bottomrule
    \end{tabular}}
  \end{minipage}
  \caption{
At the unchanged edge
$e=(\hub,v_\mu)$, only composition notices the fold of $\lambda$.
After recoloring, no $X^{(\lambda)}$-edge remains at $\hub$.
Only two witness paths in $\mathcal M_0$ are shown.
}
  \label{fig:prediction}
\end{figure}
In Figure~\ref{fig:prediction}, composition notices
the fold at $e$, since
$e\in(X;Y)^{(\lambda)}$ (\textsf{Predicted})
but
$e\notin X^{(\lambda)};Y^{(\lambda)}$ (\textsf{Recomputed}).
By contrast, union and complement do not notice the fold.

Every fold must be noticed by some composition gate.
Let $C$ be a \CoR\ circuit such that $C\equiv G_k$.
If no composition gate notices the fold of $\lambda$,
then, since $(\hub,v_{\lambda^c})$ is unchanged,
Lemma~\ref{lem:unnoticed} gives
\[
  \mathcal M_0\models C(\hub,v_{\lambda^c})
  \quad\Longleftrightarrow\quad
  \mathcal M_\lambda\models C(\hub,v_{\lambda^c}).
\]
This contradicts $C\equiv G_k$, since
$\mathcal M_0\models G_k(\hub,v_{\lambda^c})$
but $\mathcal M_\lambda\not\models G_k(\hub,v_{\lambda^c})$.

\subsubsection{Each composition notices at most two folds}\label{ov:twocolors}

We now show that any composition $X;Y$ of color-determined relations
notices at most two folds (Lemma~\ref{lem:local}). The two-witness
property first rules out noticing at pairs away from the hub.

\paragraph{Away from the hub (Figure~\ref{fig:notice}).}
At pairs $(u,v)$ with $u,v\ne\hub$, composition agrees with
prediction.
Figure~\ref{fig:notice} uses the same coloring, operands $X,Y$, and fold
as Figure~\ref{fig:prediction}, but observes $e'=(u,v_\mu)$.

\begin{figure}[H]
  \centering
  \begin{tikzpicture}[
    >=Latex,
    every node/.style={font=\scriptsize},
    vtx/.style={circle,draw,fill=white,inner sep=1.5pt},
    edge/.style={line width=.55pt},
    recolored/.style={red!75!black,densely dashed,line width=.85pt},
    surviving/.style={line width=1.15pt},
    read/.style={densely dotted,line width=.55pt},
    folded/.style={draw=red!65!black,line width=1.2pt},    
  ]
    \begin{scope}
      \coordinate (u0) at (0,0);
      \coordinate (v0) at (2.65,0);
      \coordinate (h0) at (1.325,.98);
      \coordinate (w0) at (1.325,-.98);
      \draw[] (u0)--node[below,pos=.57] {$\mu$} (v0);
      \draw[edge] (u0)--node[above left] {$\lambda$} (h0);
      \draw[edge] (h0)--node[above right] {$\mu$} (v0);
      \draw[surviving] (u0)--node[below left] {$\lambda$} (w0);
      \draw[surviving] (w0)--node[below right] {$\mu$} (v0);
      \node[vtx,label=left:{$u$}] at (u0) {};
      \node[vtx,label=right:{$v_\mu$}] at (v0) {};
      \node[vtx,label=above:{$\hub$}] at (h0) {};
      \node[vtx,label=below:{$w$}] at (w0) {};
      \node[font=\small] at (.1, 1.2) {$\mathcal M_0$};
    \end{scope}
    \draw[-{Latex[length=2mm]},line width=.7pt]
      (3.30,0)--node[above] {fold $\lambda$} (4.65,0);
    \begin{scope}[xshift=5.25cm]
      \coordinate (u1) at (0,0);
      \coordinate (v1) at (2.65,0);
      \coordinate (h1) at (1.325,.98);
      \coordinate (w1) at (1.325,-.98);
      \draw[] (u1)--node[below,pos=.57] {$\mu$} (v1);
      \draw[folded] (u1)--node[above left,black] {$\lambda^c$} (h1);
      \draw[edge] (h1)--node[above right] {$\mu$} (v1);
      \draw[surviving] (u1)--node[below left] {$\lambda$} (w1);
      \draw[surviving] (w1)--node[below right] {$\mu$} (v1);
      \node[vtx,label=left:{$u$}] at (u1) {};
      \node[vtx,label=right:{$v_\mu$}] at (v1) {};
      \node[vtx,label=above:{$\hub$}] at (h1) {};
      \node[vtx,label=below:{$w$}] at (w1) {};
      \node[font=\small] at (.1, 1.2) {$\mathcal M_\lambda$};
    \end{scope}
  \end{tikzpicture}
  \caption{Away from the hub, the fold invalidates the path through $\hub$
    but leaves the second witness path through $w$ unchanged.
    Only these two witness paths are shown.}
  \label{fig:notice}
\end{figure}

In $\mathcal M_0$, both paths have colors $(\lambda,\mu)$.
After the fold, the path through $\hub$ is invalidated but the path through $w$ survives.
Hence 
$(X^{(\lambda)};Y^{(\lambda)})(e')
=1
=(X;Y)^{(\lambda)}(e')$.

More generally, between any distinct $u,v\ne\hub$, the two-witness
property preserves every two-color path pattern via a witness other
than $\hub$.
For a diagonal pair $(u,u)$ with $u\ne\hub$, vertex $u$
still has an incident edge of every color.
The endpoint cases $z\in\set{u,v}$ are unchanged.
Thus no composition notices a fold at pairs away from the hub.

\paragraph{At the hub (Figure~\ref{fig:prediction} revisited).}

Only pairs involving $\hub$ remain.
In Figure~\ref{fig:prediction}, $X$ holds only on $\lambda$ and
$Y$ only on $\mu$, where $\lambda\ne\mu$.
Their composition notices the folds of
$\lambda$ and $\mu$ at $(\hub,v_\mu)$ and
$(v_\lambda,\hub)$, respectively.

For arbitrary color-determined $X,Y$, endpoint witnesses contribute equally to recomputation and prediction, even at recolored pairs.
Thus only non-endpoint witnesses can cause a discrepancy.
To notice $\lambda$ at an off-diagonal pair $(\hub,v)$,
$X$ must hold on $\lambda$ and no other off-diagonal color.
Indeed, if $X[\nu]=1$ for some
$\nu\ne\lambda$ and $Y[\tau]=1$ for some
color $\tau$, the two-witness property
supplies a path
$\hub\xrightarrow{\nu}z\xrightarrow{\tau}v$
that survives the fold of $\lambda$.
Thus $z$ remains a witness for
$(X^{(\lambda)};Y^{(\lambda)})(\hub,v)$.
%
Likewise, noticing a fold $\mu$ at an off-diagonal pair $(u,\hub)$ requires $Y$ to hold on $\mu$ and no other off-diagonal color.
At $(\hub,\hub)$, a fold can be noticed only if its color is the
unique off-diagonal color shared by $X$ and $Y$.
If they share a color, all noticed folds must therefore coincide;
otherwise, only the hub row (pairs $(\hub, v)$)
and the hub column (pairs $(u, \hub)$) can notice folds, at most one each.
Thus $X;Y$ notices at most two folds.

\paragraph{Summary.}
Every one of the $2^k$ folds requires a noticing composition gate
(\S\ref{ov:types}), whereas each distinct composition gate notices at
most two (Lemma~\ref{lem:local}). Hence $2^k\le 2\gcost C$,
proving the lower bound of Theorem~\ref{thm:main} even with sharing.

\subsection{Further results}\label{sec:additional-results}

The fold-counting argument of \S\ref{sec:overview} also yields exponential succinctness gaps from \FOthree\ to \FOtwo\ and from \CoR\ to \FOtwo.
It extends to a semiring-valued analogue of the
\FOthree-to-\CoR\ translation:
the translation from annotated relational
algebra with at most three attributes (\textbf{ARA(3)})
to the matrix query language \textbf{MATLANG}.

Folds are indexed by colors $\lambda\in\Lambda$, where $\Lambda\coloneqq2^{[k]}$ is the power set of $[k]$.
For a circuit $C \equiv G_k$,
let $D_g\subseteq\Lambda$ be the set of colors whose folds are noticed by gate $g$.
Every color belongs to some $D_g$: otherwise, the circuit
could not detect the change in $G_k$ at the pair left unchanged by that
fold (\S\ref{ov:types}). Thus $\bigcup_{g\in C}D_g=\Lambda$, and
counting gives
\begin{equation}\tag{$\dagger$}\label{eq:fold-counting}
  \smash{2^k\le \sum_{g\in C}|D_g|}.
\end{equation}
For \CoR, \S\ref{ov:twocolors} gives $\sum_{g\in C}|D_g|\le 2\gcost C$, recovering Theorem~\ref{thm:main}.
The applications below use the same counting argument, with different
bounds on the folds noticed by the target operations and, when needed,
a different reference structure.

\paragraph{\texorpdfstring{{\boldmath\FOthree\ to \FOtwo}}{FO3 to FO2}.}
The same query $G_k$ is also \FOtwo-definable: the DNF expansion in
\S\ref{sec:main-results} gives
\[
 G_k(x,y)\equiv\bigvee_{A\subseteq[k]}
 \bigl[H_A(x,y)\land\exists y\,H_{A^c}(x,y)\bigr]
\qquad
\text{where}\ H_A(x, y) \equiv \bigwedge_{i \in A} P_i(x,y) \land \bigwedge_{i \notin A} Q_i(x,y).
\]
The inner $\exists y$ binds $y$ only in the second conjunct.

\begin{restatable}[Exponential lower bound from \FOthree\ to \FOtwo\proofref{sec:FOtwo}]{theorem}{sentencebounds}\label{cor:sentences}
For every $k\ge1$,
every \FOtwo\ formula $\varphi$ such that $\varphi\equiv G_k$ has at least $2^k$ quantifiers.

\end{restatable}
The displayed expansion uses exactly $2^k$ quantifiers, so the bound is tight.
More generally, DNF expansion gives $2^{O(s)}$-size \FOtwo\ translations
for size-$s$ queries $\exists z\,\theta(x,y,z)$ where $\theta$ is
quantifier-free and, up to exchanging $x$ and $y$, no atom contains
both $y$ and $z$ (\emph{one-sided single-quantifier queries};
Proposition~\ref{prop:one-sided-fo2}).
The gap also holds for sentences (see \S\ref{sec:sentences}).
On pure finite linear orders, by contrast, \FOthree\ sentences admit
polynomial-size \FOtwo\ translations~\citep{grohe-schweikardt2005};
\S\ref{sec:relwork} also compares the separations
of Voigt~\citep{voigt-thesis2019} and Vinall-Smeeth~\citep{vinall-smeeth2024}.

To prove the \FOtwo\ lower bound,
we show that each quantifier notices at most one fold.
The same reference structure $\mathcal M_0$
and folds from \S\ref{sec:overview} suffice.
For each \FOtwo\ subformula $\psi$ and structure $\mathcal M$,
let $X_\psi \coloneqq \sem\psi{\mathcal M}$ denote its binary assignment relation,
consisting of pairs $(u,v)$ satisfying $\psi$ under $x=u$ and $y=v$, with unused coordinates treated as dummies (see \S\ref{sec:prelim-logic}).
Existential quantifiers are then projections:
\[
  X_{\exists y\,\psi}=X_\psi;\one,
  \qquad
  X_{\exists x\,\psi}=\one;X_\psi.
\]
Thus their assignment relations remain color-determined on $\mathcal M_0$,
just like the CoR gate values in \S\ref{sec:overview}.
Since one operand of each composition is $\one$, each quantifier can
notice at most one fold (\S\ref{ov:twocolors}), rather than the two
possible for general composition.
On each folded structure, atomic formulas have their predicted values,
and Boolean connectives notice no folds; universal quantifiers are handled by complement.
The counting inequality~\eqref{eq:fold-counting} therefore
requires at least $2^k$ quantifiers. We also obtain a lower bound of $k2^k$ relation-atom occurrences and
optimal \FOtwo\ formula size $\Theta(k2^k)$ for $G_k$ and its sentence variant $G_k^\forall$ (\S\ref{sec:sentences}).

For \CoR, the optimal term size of $G_k$ is $(2+o(1))k2^k$
(Theorem~\ref{thm:tree}), whereas its optimal circuit size is
$\Theta(2^k)$ (\S\ref{sec:main-proof}).

\paragraph{\texorpdfstring{{\boldmath\CoR\ to \FOtwo}}{CoR to FO2}.}
An exponential succinctness gap also occurs from \CoR\ to \FOtwo.
We need a different query, since $G_k$ already requires exponentially
many \CoR\ compositions. Put $\mathscr P_i^1\coloneqq\id\cap P_i$
and $\mathscr P_i^0\coloneqq\id\cap\overline{P_i}$, and let
$\vartheta_A(u)\coloneqq\bigwedge_{i\in A}P_i(u,u)\land
\bigwedge_{i\notin A}\neg P_i(u,u)$ for $A\subseteq[k]$.
Consider the following linear-size \CoR\ term and its \FOtwo\ expansion:
\[
\defn{\CoR}:
 U_k \coloneqq R;\bigcap_{i=1}^k
 \left[(\mathscr P_i^0;\one;\mathscr P_i^0)
 \cup(\mathscr P_i^1;\one;\mathscr P_i^1)\right]
 \longrightarrow
\defn{\FOtwo}:
 \bigvee_{A\subseteq[k]}
 \left[\vartheta_A(y)\land\exists y\,
    (R(x,y)\land\vartheta_A(y))\right].
\]
Thus $U_k$ asks whether $x$ has an $R$-successor $z$ with the same loop pattern $\set{i : P_i(z,z)}$ as $y$.

\begin{restatable}[Exponential lower bound from \CoR\ to \FOtwo\proofref{sec:reverse}]{proposition}{reverse}\label{prop:reverse}
For every $k\ge1$, every \FOtwo\ formula
$\varphi$ such that $\varphi\equiv U_k$ has at least
$2^k$ quantifiers.
\end{restatable}

Again, the quantifier bound is tight.
For the lower bound, take a new reference structure
$\mathcal W_0$ with a hub $\hub$ and several vertices of each loop pattern.
For every $A\subseteq[k]$, choose vertices $z_A,y_A$ of pattern $A$ with $R(\hub,z_A)$ but not $R(\hub,y_A)$ (Figure~\ref{fig:reverse-fold}).
As in \S\ref{sec:overview},
a fold changes only selected input pairs:
the fold of $A$ deletes all $R$-edges from $\hub$ to vertices of pattern $A$,
leaving all other entries unchanged.
Thus $U_k(\hub,y_A)$ changes from true to false,
although no input relation changes at $(\hub,y_A)$.

\begin{figure}[!htbp]
\centering
\begin{tikzpicture}[
  >=Stealth,
  every node/.style={font=\small},
  patternclass/.style={
    draw=black!50, fill=black!3,
    rounded corners=2pt,
    minimum width=2.5cm, minimum height=.73cm
  },
  vertex/.style={
    circle, draw=black, fill=white,
    inner sep=0pt, minimum size=4pt
  },
  redge/.style={
    ->, semithick, 
  },
  removed/.style={
    ->, very thick, densely dashed, red!75!black,
  },
]
\begin{scope}
  \node at (1.0, 1.0) {Reference $\mathcal W_0$};
  \node[vertex,label={[font=\scriptsize]left:$\hub$}]
    (h0) at (.25,0) {};

  \node[patternclass] at (3.85,.69) {};
  \node[font=\scriptsize] at (3.85,1.18)
    {pattern $A$};
  \node[vertex,label={[font=\scriptsize]right:$z_A$}]
    (za0) at (3.32,.69) {};
  \node[vertex,label={[font=\scriptsize]right:$y_A$}]
    (ya0) at (4.45,.69) {};

  \node[patternclass] at (3.85,-.81) {};
  \node[font=\scriptsize] at (3.85,-.32)
    {pattern $B\ne A$};
  \node[vertex,label={[font=\scriptsize]right:$z_B$}]
    (zb0) at (3.32,-.81) {};

  \draw[redge] (h0) -- (za0);
  \draw[redge] (h0) -- (zb0);
  \node[font=\scriptsize] at (2.7,-1.52)
    {$U_k(\hub,y_A)=1$};
\end{scope}

\begin{scope}[xshift=6.45cm]
  \node at (1.0, 1.0) {Folded $A$: $\mathcal W_A$};
  \node[vertex,label={[font=\scriptsize]left:$\hub$}]
    (h1) at (.25,0) {};

  \node[patternclass] at (3.85,.69) {};
  \node[font=\scriptsize] at (3.85,1.18)
    {pattern $A$};
  \node[vertex,label={[font=\scriptsize]right:$z_A$}]
    (za1) at (3.32,.69) {};
  \node[vertex,label={[font=\scriptsize]right:$y_A$}]
    (ya1) at (4.45,.69) {};

  \node[patternclass] at (3.85,-.81) {};
  \node[font=\scriptsize] at (3.85,-.32)
    {pattern $B\ne A$};
  \node[vertex,label={[font=\scriptsize]right:$z_B$}]
    (zb1) at (3.32,-.81) {};

  \draw[removed] (h1) -- (za1);
  \node[
    font=\large, text=red!75!black,
    fill=white, inner sep=0.5pt
  ] at (1.80,.35) {$\times$};
  \draw[redge] (h1) -- (zb1);

  \node[font=\scriptsize] at (2.7,-1.52)
    {$U_k(\hub,y_A)=0$};
\end{scope}
\end{tikzpicture}
\caption{The fold of $A$ deletes all $R$-edges
(red crossed arrow) from $\hub$ to vertices of loop pattern $A$.}
\label{fig:reverse-fold}
\end{figure}

The reference construction of \S\ref{sec:overview} cannot serve here:
every composition notices at most two folds there, 
whereas $U_k$ has a \CoR\ term with only $O(k)$ compositions.
On $\mathcal W_0$, the outer composition of $U_k$ notices all $2^k$ folds, yet each \FOtwo\ quantifier,
viewed as a projection, notices at most one.
The latter relies on redundancy:
with either endpoint fixed, every off-diagonal atomic pair type
(the truth values of atomic formulas on an ordered pair)
occurring in $\mathcal W_0$ with that endpoint's loop pattern is realized with at least two distinct other vertices.
Applying the same counting argument \eqref{eq:fold-counting} proves the lower bound (\S\ref{sec:reverse}).

\paragraph{Arity elimination in matrix query languages.}
The \FOthree-to-\CoR\ translation has a semiring-valued analogue: ARA(3) plays the role of \FOthree,
allowing intermediate relations with up to three attributes,
while MATLANG plays the role of \CoR,
computing with matrices, vectors, and scalars
(see \S\ref{sec:matrix}).
Brijder, Gyssens, and Van den Bussche showed that,
over commutative semirings,
the two formalisms have the same expressive power for inputs and outputs of arity at most two \citep{brijder2022},
but left open whether the exponential ARA(3)-to-MATLANG translation blowup is unavoidable \citep[Remark~4 and \S4.5]{brijder2022}.

Let $\mathfrak S=(S,\oplus,\otimes,0_S,1_S)$ be a
commutative semiring, and let $P_i,Q_i:V^2\to S$
be input relations over a finite set $V$.
Replacing $\lor$, $\land$, and $\exists z$ in $G_k$
by $\oplus$, $\otimes$, and aggregation over $z$,
respectively, gives
\[
\mathcal G_k(x,y):=
\bigoplus_{z\in V}\bigotimes_{i=1}^k
\left[
(P_i(x,y)\otimes Q_i(x,z))
\oplus
(Q_i(x,y)\otimes P_i(x,z))
\right].
\]
Over the Boolean semiring
($S=\set{0, 1}$, with $\oplus=\lor$ and $\otimes=\land$),
$\mathcal G_k$ specializes to $G_k$.
As in \S\ref{sec:main-results},
the DNF expansion yields a MATLANG expression using $2^k$ matrix--vector products (see \S\ref{sec:weightedG}).

For the lower bound, we use the stronger \emph{typed matrix circuit} model of \S\ref{sec:typed},
which allows shared gates and generalized matrix multiplication.
Let $\#_{\mathrm{mv}}(C)$ count matrix–vector and vector–matrix product gates,
and let $\#_{\mathrm{mm}}(C)$ count matrix–matrix product gates.
All other permitted operations are uncharged.

For $D\in\set{\mathbb N,\mathbb Z}$,
let
$\operatorname{Trop}_D \coloneqq (D\cup\{+\infty\},\min,+,+\infty,0)$
be the min--plus semiring.

\begin{restatable}[Tropical arity-elimination overhead\proofref{sec:matrix}]{theorem}{tropicalmatrixbounds}\label{thm:tropical-matrix}
For every $k\ge1$ and $\mathfrak S\in\{\mathsf{Trop}_{\mathbb N},\mathsf{Trop}_{\mathbb Z}\}$,
the query $\mathcal G_k$ has an ARA$(3)$ expression of size $O(k)$,
but every typed matrix circuit $C$ over $\mathfrak S$ computing it
on all finite inputs satisfies $\mvcount C+3\mmcount C\ge2^k$.
\end{restatable}

An analogous exponential lower bound holds for every fixed nontrivial commutative idempotent semiring of \textit{finite join breadth}, with constants depending on the breadth.
Here finite join breadth means that every finite sum equals the sum of a bounded number of its summands.
Since each fold removes one color from the hub while preserving
the other color patterns, finite join breadth bounds the number
of folds each aggregation gate can notice~(Theorem~\ref{cor:matrix}).
For positive circuits over $(\mathbb N,+,\times)$, with scalar operations restricted to addition and multiplication,
a separate reduction mapping zero to false and positive values to true yields an exponential lower bound (\S\ref{sec:positive}).

\paragraph{Robustness: closure and randomization.}

The lower bound persists even when \CoR\ is extended with reflexive–transitive closure 
$X^* \coloneqq \bigcup_{m\ge 0}X^m$,
which can be evaluated in $O(n^\omega)$ time~\citep{fischer-meyer1971,munro1971,furman1970}.
Each closure gate notices at most one fold.

\begin{restatable}[Closure does not help\proofref{sec:closure}]{proposition}{closure}\label{prop:star}
For every $k \geq 1$,
every \CoR\ circuit $C$ augmented with closure gates $(-)^*$
such that $C \equiv G_k$
satisfies $2 \gcost C + g_*(C) \ge 2^k$,
where $g_*(C)$ counts its distinct closure gates.
\end{restatable}

The lower bound also extends to bounded-error randomized circuits, even when circuits are specialized to each sufficiently large domain size.

\begin{restatable}[Fixed sizes and bounded error\proofref{sec:query-bounds}]{corollary}{randomizedbounds}\label{cor:randomized}

For some constant $c>0$, for every $k\ge1$, $n\ge ck4^k$, and $0\le\varepsilon<1/2$,
let $C$ be drawn from an input-independent distribution of \CoR\ circuits for size $n$.
If, on every $n$-element input and at every pair $(u,v)$, $C$ returns $G_k(u,v)$ with probability at least $1-\varepsilon$, then
\[
\mathbb E[g_{;}(C)]\ge(1-2\varepsilon)2^{k-1}.
\]
\end{restatable}

\subsection{Related work}\label{sec:relwork}

\paragraph{History of \CoR.}
\CoR\ has its roots in 19th-century algebraic logic~\citep[\S6]{burris-legris}, \citep{pratt1992}.
De Morgan introduced (relational) composition and converse in work from the 1850s and 1860s~\citep{demorgan1856,demorgan1864};
Peirce developed the logic of relatives~\citep{peirce1870,peirce1880},
and Schr\"oder devoted a volume to it in 1895~\citep{schroeder1895}.
This was before first-order logic was isolated in Hilbert's lectures and Hilbert and Ackermann's textbook~\citep{hilbert-lectures,hilbert-ackermann1928,moore1988}.
The central expressiveness problem was whether every quantified statement
about relations could be \emph{condensed} into this variable-free calculus.
Schr\"oder claimed that it could~\citep[p.~551]{schroeder1895}, but
Korselt found a counterexample, published by L\"owenheim in
1915~\citep{loewenheim1915,loewenheim1967}. In modern terms, the \FOfour\ sentence
$
\forall a_1\,a_2\,a_3\,a_4.\bigvee_{i<j} a_i = a_j
$
states that the domain has at most three elements, but is expressible
neither in \FOthree\ nor by a \CoR\ equation.

Tarski revisited the expressive power of \CoR\ in 1941~\citep{tarski1941}.
According to Tarski and Givant,
he soon obtained the three-variable characterization, whose full proof appeared in~\citep{tarski-givant1987};
see also~\citep{maddux2006} and the historical
accounts~\citep{pratt1992,maddux1991,burris-legris,brady2000}.
They did not analyze the size of their original recursive translation~\citep[\S3.9]{tarski-givant1987}.
A straightforward bottom-up translation using
DNF conversion~\citep{hilbert-ackermann1928,kleene1952,schoening1989}
and quantifier elimination by relational composition gives a
single-exponential upper bound~\citep[Theorem~3.1(1)]{nakamura2022}
(see also Appendix~\ref{app:classical-translation}).
The reverse \CoR-to-\FOthree\ translation is linear.

\paragraph{Relational computation and quantifier elimination.}
Relational operations have long served as computational primitives,
often outside full \CoR.
Copilowish used matrices over finite domains~\citep{copilowish1948},
and Levien and Maron applied them to retrieval and inference~\citep{levien-maron1965,levien-maron1967,levien1969}.
Codd's relational algebra laid the foundations of database querying~\citep{Codd1970,Codd1972}, followed by work on relational-expression
equivalence, optimization, and efficient join evaluation~\citep{aho-sagiv-ullman1979equiv,aho-sagiv-ullman1979opt,SelingerEtAl1979,yannakakis1981}.
Binary relational-expression algorithms~\citep{hunt-szymanski-ullman1974,hunt-szymanski-ullman1977,szymanski-ullman1975,szymanski-ullman1977} also served parsing
and recursive queries~\citep{sippu-soisalon-soininen1985,grahne-sippu-soisalon-soininen1987,
grahne-sippu-soisalon-soininen1991}.
Related algebras also appear in relational programming, constraint networks, temporal reasoning, database execution, and graph querying~\citep{maclennan1981,montanari1974,allen1983,boncz-kersten1999,cruz-mendelzon-wood1987,FletcherEtAl2015,HellingsEtAl2021,HellingsEtAl2022}.

Matrix multiplication also implements existential quantifier elimination
through composition, as in algorithms for transitive closure
\citep{furman1970,fischer-meyer1971,munro1971,oneil-oneil1973,
galil-margalit1993}, triangles, cliques, cycles,
and small subgraphs
\citep{itai-rodeh1978,nesetril-poljak1985,alon-yuster-zwick1997,
kloks-kratsch-mueller2000,eisenbrand-grandoni2004,yuster-zwick2004}.
Already in 1977, Chandra and Merlin evaluated generalized join--project
operations by rectangular Boolean matrix multiplication~\citep{ChandraMerlin1977}.
Williams's algorithm for prenex first-order graph properties
\citep{Williams2014} can be viewed in the same algebraic framework:
DNF-based quantifier elimination yields ordinary \CoR\ composition
for three quantifiers and typed composition over tuples for larger
prenex formulas, evaluated using rectangular matrix multiplication
(as in \citep{ChandraMerlin1977}).
Gao and Impagliazzo~\citep[Theorem~7 for $k = 3$]{GaoImpagliazzo2019} also noted that general
\FOthree\ formulas, beyond prenex normal form, can be evaluated in matrix multiplication time for each fixed formula.
Related techniques appear in first-order graph algorithms
\citep{BringmannFischerKunnemann2019} and database join--project
and join--aggregate evaluation
\citep{DeepHuKoutris2020,hu-yi2020,Hu2024,AboKhamisHuSuciu2025}.
Most directly, our earlier work used the \FOthree-to-\CoR\ translation
for truly subcubic \FOthree\ model checking
\citep{nakamura-uezato2027}.

\paragraph{Succinctness results.}
Nakamura proved an exponential circuit lower bound for $\Sigma_1\CoR$,
and showed that, over one binary relation, existential-positive \FOthree\ is superpolynomially more succinct than $\Sigma_1\CoR$~\citep[\S6]{nakamura2022}.
In our earlier work~\citep[Corollary~40]{nakamura-uezato2027}, we ruled out
polynomial-time \FOthree-to-\CoR\ translations under the 3-OV
hypothesis, but not the existence of small equivalent representations.
Theorem~\ref{thm:main} gives an unconditional exponential lower bound
even for \CoR\ circuits, while Theorem~\ref{thm:fixed-vocabulary}
gives tight subexponential circuit bounds over one binary relation.

On pure finite linear orders,
Grohe and Schweikardt gave a polynomial-size translation from \FOthree\ sentences to \FOtwo\ sentences~\citep[Theorem~4.3]{grohe-schweikardt2005}.
Since \FOtwo\ translates linearly into \CoR,
this also yields polynomial-size \CoR\ equations.
By contrast, our exponential succinctness gaps from \FOthree\ to both \FOtwo\ and \CoR\ hold even for sentences over general finite structures.
Grohe and Schweikardt also showed that full FO is exponentially more
succinct than \FOthree\ on the same orders, already witnessed by
\FOfour~\citep[Theorem~5.3]{grohe-schweikardt2005}. The short \FOfour\
formulas combine divide-and-conquer for distances~\citep{savitch1970}
with the classical quantifier-based abbreviation trick, which avoids
duplicating recursive subformulas
\citep{fischer-rabin1974,stockmeyer1974,ferrante-rackoff1979,compton-henson1990,vorobyov1996,stockmeyer-meyer2002}.
The \FOthree\ lower bound draws on Adler--Immerman formula-size
games~\citep{adler-immerman2003}.
However, these distance predicates admit linear-size \CoR\ circuits by repeated squaring ($S_{i+1} \coloneqq S_i;S_i$), starting from the directed successor relation.
Thus this distance-based separation does not extend to \CoR\ circuits, whereas our lower bound holds even with sharing.

Voigt proved an exponential gap from \FOfour\ to \FOtwo~\citep[Theorem~3.12.5]{voigt-thesis2019}.
Vinall-Smeeth established exponential gaps from
$\mathrm{FO}^{\ell+1}$ to $\mathrm{FO}^{\ell}$ for every $\ell\ge 3$,
using higher-arity relations~\citep[Theorem~2]{vinall-smeeth2024}.
Our \FOthree-to-\FOtwo\ gap is witnessed already by positive one-sided single-quantifier queries over binary signatures, and the lower bound counts only quantifiers in unrestricted \FOtwo.
The linear \FOtwo-to-\CoR\ translation was already
known~\citep{nakamura2022}, while Proposition~\ref{prop:reverse}
gives an exponential gap in the reverse direction.

\paragraph{Lower-bound methods.}
As explained in \S\ref{sec:natural-approaches}, the witness-crossing
argument underlying Nakamura's formula-size-game proof for $\Sigma_1\CoR$ \citep{nakamura2022}, in the style of Adler and
Immerman~\citep{adler-immerman2003}, relies on preservation under
embeddings, which fails with unrestricted complement. We also considered other standard
lower-bound techniques, but their usual resource measures do not directly
fit our setting. Formula-size games
\citep{grohe-schweikardt2005,french2013,hella-vilander2019}
and Ehrenfeucht--Fra\"{\i}ss\'e, pebble, and multi-structural games~\citep{fraisse1954classif,ehrenfeucht1961,immerman1981,immerman1982,libkin2004,fagin2021,carmosino2024}
address formula size, quantifier depth, variable count, or quantifier number.
Bounds on these resources do not directly yield lower bounds on the number of distinct composition gates when sharing is allowed and Boolean operations are uncharged.
Depth alone is insufficient: $G_k$ admits a \CoR\  expression of composition depth one.
Switching-lemma and random-restriction methods~\citep{hastad1986} typically exploit
shallow Boolean-circuit structure, whereas we impose no restrictions on
Boolean nesting and charge only compositions.
%

Communication complexity suggests further approaches. Rank,
discrepancy, and Karchmer--Wigderson techniques
\citep{mehlhorn-schmidt1982,kushilevitz-nisan1997,karchmer-wigderson1990,jukna2012}
do not directly bound compositions when Boolean operations are uncharged.
For example, the set-disjointness matrix on subsets of $[k]$ has rank $2^k$,
yet its 1-entries form the complement of a union of only $k$ rectangles.
Alternating existential
and universal witness tests also resemble the nondeterministic and
co-nondeterministic levels of the communication polynomial hierarchy
\citep{babai-frankl-simon1986,goos-pitassi-watson2018}, but we know of
no direct transfer of lower bounds to our circuit model.

Our gatewise counting shares the idea of charging local discrepancies to individual gates with 
Razborov's method of approximations~\citep{razborov1985,alon-boppana1987,Razborov1989}
and Haken--Cook's bottleneck counting~\citep{HakenCook1999};
Simon and Tsai~\citep{SimonTsai2000} related these methods for monotone Boolean circuits.
Our fold-based argument charges failures to commute with prediction to individual gates,
without restricting Boolean operations or circuit sharing.

\subsection{Organization}
Section~\ref{sec:aggregation} fixes semantics and resource measures.
Section~\ref{sec:preservation} builds the reference structure, proves
that unnoticed folds leave circuit values as predicted and that each
composition notices at most two folds, and supplies the tests for $G_k$
to prove Theorem~\ref{thm:main}.
Section~\ref{sec:logical}
formulates the counting argument \eqref{eq:fold-counting} as the test–cost theorem
and derives the further bounds for relational circuits: equations, fixed sizes and
bounded error, circuits without converse, domain size, random inputs,
and closure.
Section~\ref{sec:fixed-vocabulary} proves the tight fixed-vocabulary
circuit bounds and derives the model-checking application.
Section~\ref{sec:FOtwo} transfers the bound to quantifiers, asymptotically determines
the minimum term size of $G_k$, and proves Theorem~\ref{cor:sentences} and
Proposition~\ref{prop:reverse}. Section~\ref{sec:matrix} gives a concrete
ARA(3) query, formalizes typed matrix circuits, and extends the fold
analysis to semiring-valued aggregation.
The appendices give explicit two-witness colorings, the complete
term-size charging argument, sharpness and random-input bounds for the
fold tests, the directed construction without converse,
single-exponential translation bounds, and the full small-domain
circuit construction.

\introproofrefsfalse 

\section{Preliminaries}\label{sec:aggregation}
We fix the semantics of finite-variable formulas, relational terms, and
circuits, with examples that will be used throughout the paper. All
structures have nonempty domains. Our lower bounds concern finite
structures, whereas the explicit translations hold on all nonempty
structures.

\paragraph{Common notation.}

Throughout the paper,
we use $\log x\coloneqq\max\{\log_2 x,1\}$ for $x>0$,
while $\ln$ denotes the natural logarithm.
In running-time bounds, $n^\omega$ is shorthand for $n^{\omega+o(1)}$.

We use $\mathbb N\coloneqq\{0,1,2,\ldots\}$. Subscripts on asymptotic notation
indicate the parameters on which hidden constants may depend;
$\operatorname{poly}_\sigma(m)$ also permits the polynomial degree to
depend on $\sigma$.
We use $[k]\coloneqq\{1,\ldots,k\}$ and $2^{[k]}$ for its power set;
$A^c\coloneqq[k]\setminus A$ when $A\subseteq[k]$. Empty intersections and
unions of relations mean $\one$ and $\zero$, respectively; the
analogous Boolean operations mean true and false. Truth values are
$\bool\coloneqq\{0,1\}$, and $[\mathcal P]$ is the truth value of a statement $\mathcal P$.

\subsection{Finite-variable logic and binary queries}\label{sec:prelim-logic}
\paragraph{Signatures and formulas.}
A \emph{signature} (or \emph{vocabulary}) $\sigma$ is a finite set of binary relation symbols;
there are no function or individual constant symbols. A $\sigma$-structure
$\mathcal M$ has domain $\cordom \coloneqq\operatorname{dom}(\mathcal M)\ne\varnothing$
and interprets each $R\in\sigma$ as a relation $R^{\mathcal M}\subseteq \cordom^2$.
Equality is available. For $\ell\ge1$, the \emph{$\ell$-variable fragment}
$\mathrm{FO}^{\ell}[\sigma]$ consists of first-order formulas using only
$\ell$ variable names, with arbitrary reuse and rebinding of those names.
Formally, with $u,v$ ranging over a fixed set of $\ell$ names, its syntax is
\[
 \varphi::=R(u,v)\mid u=v\mid\neg\varphi
 \mid(\varphi\land\varphi)\mid(\varphi\lor\varphi)
 \mid\exists u\,\varphi\mid\forall u\,\varphi.
\]
We use the names $x,y$ for \FOtwo\ and $x,y,z$ for \FOthree.
The restriction is on the \emph{number of variable names}, not on the
number or nesting depth of quantifiers. For example,
\[
 \varphi(x,y)\coloneqq P(x,y)\land\exists y\,Q(x,y)
\]
is an \FOtwo\ formula: the quantifier rebinds $y$ only in its own scope,
leaving the occurrence of $y$ in $P(x,y)$ free. Write
$\operatorname{FV}(\varphi)$ for the set of free variables.
A \emph{sentence} has no free variables. The \emph{quantifier depth}
is the maximum number of nested quantifiers, not their total number.
We abbreviate $\neg\alpha\lor\gamma$ by $\alpha\to\gamma$, and
$(\alpha\land\gamma)\lor(\neg\alpha\land\neg\gamma)$ by
$\alpha\leftrightarrow\gamma$.
A \emph{positive} formula uses atoms, $\land$, $\lor$, and quantifiers,
but no negation; an \emph{existential-positive} formula also uses no
universal quantifiers. A formula is in \emph{prenex form} if all its
quantifiers precede a quantifier-free formula, and in \emph{negation normal form (NNF)} if negation occurs only immediately before atomic formulas.

For an assignment $s$ of variables to elements of $\cordom$,
$\mathcal M,s\models\varphi$ means that $\varphi$ is true in
$\mathcal M$ under $s$. For a sentence, its truth is independent of
$s$, and we write $\mathcal M\models\varphi$;
$\mathcal M\not\models\varphi$ means that it is false.

\paragraph{Assignment relations.}
A binary query has an ordered pair of output coordinates, normally
$(x,y)$. For a formula with at most two free variables using other
names, we first fix their correspondence with these coordinates and
apply a consistent permutation of variable names, including bound
occurrences, so that no variable is captured. Under the $(x,y)$ interface, a formula with
$\operatorname{FV}(\varphi)\subseteq\{x,y\}$ defines the
\emph{binary query}
\[
 \sem\varphi{\mathcal M}
 \coloneqq\{(a,b)\in \cordom^2:\mathcal M\models
       \varphi[x\mapsto a,y\mapsto b]\}.
\]
Here $[x\mapsto a,y\mapsto b]$ specifies an assignment, not a
syntactic substitution of individual constants.
A variable absent from $\operatorname{FV}(\varphi)$ contributes a
\emph{dummy coordinate}: the query is independent of its assigned value.
Thus a sentence $\varphi$ has the binary assignment value
\[
 \sem\varphi{\mathcal M}
 =\begin{cases}
     \cordom^2,&\mathcal M\models\varphi,\\
     \varnothing,&\mathcal M\not\models\varphi.
   \end{cases}
\]
This convention lets formulas, relational terms, and circuits have the
same kind of output. An \FOthree\ subformula may nevertheless have
\emph{three} free variables: for example,
$\theta(x,y,z)\coloneqq R(x,z)\land S(z,y)$ has the ternary assignment
relation $\{(a,b,c)\in \cordom^3:R(a,c)\land S(c,b)\}$. Quantifying $z$
returns a binary query. Accordingly, direct \FOthree\ evaluation can
store three-variable truth tables in $O(|\varphi|n^3)$ time on an
$n$-element domain, whereas every intermediate value of a \CoR\ term
or circuit is binary. Section~\ref{sec:matrix} revisits this arity
contrast for matrix query languages.

\subsection{The calculus of relations and relational circuits}\label{sec:prelim-cor}
\paragraph{Relational terms.}
The \emph{calculus of relations} (\CoR) has terms generated by
\[
 T::=R\mid\zero\mid\one\mid\id
   \mid(T\cup T)\mid(T\cap T)\mid\overline T
   \mid T^\smile\mid(T;T), \qquad R\in\sigma.
\]
Write $\sem T{\mathcal M}$ for the binary relation denoted by $T$ on $\cordom$.
Its constants are $\zero\coloneqq\varnothing$, $\one\coloneqq \cordom^2$, and
$\id\coloneqq\set{(a,a):a\in \cordom}$. Union and intersection are set operations;
for relations $X,Y\subseteq \cordom^2$ the other operations are
\[
  \overline X \coloneqq \cordom^2\setminus X,
  \quad
  X^\smile \coloneqq \set{(b,a):(a,b)\in X},
  \quad
  (X;Y) \coloneqq 
     \set{ (a, b) : \exists c\in \cordom. (X(a,c)\land Y(c,b)) }.
\]
We identify a relation $X$ with its $\{0,1\}$-valued characteristic
function: $X(a,b)=1$ exactly when $(a,b)\in X$. We also use $X(a,b)$
as this membership predicate, and write $X(e)$ for $X(a,b)$ when
$e=(a,b)$. Complement is unrestricted, including around compositions. A \emph{composition} $X;Y$ is Boolean
matrix multiplication, with $c$ its existential \emph{witness}.

Two useful special cases are the \emph{row and column projections}
(Proj):
\begin{equation*}\tag{Proj}\label{eq:prelim-projections}
 \begin{aligned}
  (X;\one)(a,b)&\iff\exists c\in \cordom\ X(a,c),\\
  (\one;X)(a,b)&\iff\exists c\in \cordom\ X(c,b).
 \end{aligned}
\end{equation*}
The first is independent of $b$, and the second of $a$. More generally,
for binary predicates $A,B,C$ with the indicated variable pairs, we have
the following \emph{quantifier-elimination} (QE) identity:
\begin{equation*}\tag{QE}\label{eq:prelim-elimination}
 \exists z\,\bigl(A(x,y)\land B(x,z)\land C(z,y)\bigr)
 \quad\equiv\quad \bigl(A\cap(B;C)\bigr)(x,y).
\end{equation*}
For example, the \FOtwo\ query
$P(x,y)\land\exists y\,Q(x,y)$ above is $P\cap(Q;\one)$ in \CoR.
The position of complement matters: $\overline{R;\one}$ tests that
$x$ has no $R$-successor, whereas $\overline R;\one$ tests that
$x$ has some non-$R$ pair. These queries are not equivalent.

\paragraph{A single-exponential translation from \FOthree.}\label{sec:FOthree-translation}
The \eqref{eq:prelim-elimination} identity eliminates $\exists z$
when its body is a conjunction of binary predicates on the pairs
$(x,y)$, $(x,z)$, and $(z,y)$.
A standard quantifier-elimination construction reduces every \FOthree\ quantifier to this case.
Working from the innermost quantifier outward, it puts the body of
each quantifier into disjunctive normal form and
applies~\eqref{eq:prelim-elimination} to each clause.
Each DNF expansion is exponential, yet nested quantifiers do not
lead to a tower of exponentials:
\begin{proposition*}[Classical upper bound; {\citep[Theorem~3.1(1)]{nakamura2022}}; Appendix~\ref{app:classical-translation}]
Every \FOthree\ formula $\varphi$ over binary relations with at most two
free variables has an equivalent \CoR\ term, and hence a circuit, of total
size $2^{O(|\varphi|)}$.
\end{proposition*}

\paragraph{A linear translation from \FOtwo.}\label{sec:FOtwo-translation}
Under the assignment convention above, the standard translation
$\tau$ from \FOtwo\ into \CoR~\citep[Remark~3.22]{nakamura2022} is given
on atoms and quantifiers by
\[
\begin{array}{c|c@{\qquad}c|c}
 \varphi&\tau(\varphi)&\varphi&\tau(\varphi)\\\hline
 R(x,y)&R&R(y,x)&R^\smile\\
 R(x,x)&(R\cap\id);\one&R(y,y)&\one;(R\cap\id)\\
 x=y\ \text{or}\ y=x&\id&x=x\ \text{or}\ y=y&\one\\
 \exists x\,\psi&\one;\tau(\psi)&
 \exists y\,\psi&\tau(\psi);\one
\end{array}
\]
Conjunction, disjunction, and negation become intersection, union, and
complement; universal quantifiers use
$\forall u\,\psi\equiv\neg\exists u\,\neg\psi$. Structural induction,
using~\eqref{eq:prelim-projections}, shows that
\[
 \sem{\tau(\varphi)}{\mathcal M}
   =\sem\varphi{\mathcal M},\qquad
 |\tau(\varphi)|=O(|\varphi|).
\]
The translation also preserves the number of relation-atom
occurrences, including reflexive atoms.
Section~\ref{sec:FOtwo} bounds the fold cost of the translated term
by the number of quantifiers in the source formula.

\paragraph{Terms versus circuits.}
A \emph{term} is a syntax tree, so repeated occurrences of a subterm
are distinct nodes. A \emph{relational circuit} is a finite directed
acyclic graph (DAG): each gate is an input, a constant, or one of the
same \CoR\ operations, with ordered input ports, and a designated
output gate. A gate may feed multiple later gates, including both
ports of one composition. Its value $\sem g{\mathcal M}$ is computed
once from its predecessors; $\sem C{\mathcal M}$ is the value of the
output gate. For any formula, term, or circuit $\Psi$ with a binary
output interface, $\mathcal M\models\Psi(a,b)$ means
$(a,b)\in\sem\Psi{\mathcal M}$, and
$\mathcal M\not\models\Psi(a,b)$ means the opposite.

For example, the term $(R;R);(R;R)$ has seven nodes and three
composition occurrences. The circuit
\[
 g_0\coloneqq R,\qquad g_1\coloneqq g_0;g_0,\qquad g_2\coloneqq g_1;g_1
\]
has just three distinct gates, two of which are compositions
(Figure~\ref{fig:term-vs-circuit}).

\begin{figure}[h]
\centering
\begin{tikzpicture}[
  comp/.style={draw,circle,fill=blue!7,minimum size=7mm,inner sep=0pt},
  inp/.style={draw,rounded corners=2pt,fill=gray!8,minimum width=8mm,
              minimum height=6mm,inner sep=1pt},
  every node/.style={font=\small},>=Latex]
 \begin{scope}[xshift=-2.25cm]
   \node[font=\small\bfseries] at (0,0.65) {Term (tree)};
   \node[comp] (t0) at (0,0) {$;$};
   \node[comp] (t1) at (-0.95,-0.9) {$;$};
   \node[comp] (t2) at (0.95,-0.9) {$;$};
   \node[inp] (t3) at (-1.45,-1.8) {$R$};
   \node[inp] (t4) at (-0.45,-1.8) {$R$};
   \node[inp] (t5) at (0.45,-1.8) {$R$};
   \node[inp] (t6) at (1.45,-1.8) {$R$};
   \draw (t0)--(t1) (t0)--(t2);
   \draw (t1)--(t3) (t1)--(t4) (t2)--(t5) (t2)--(t6);
   \node[font=\footnotesize] at (0,-2.30) {3 composition nodes};
 \end{scope}
 \begin{scope}[xshift=2.25cm]
   \node[font=\small\bfseries] at (0,0.65) {Circuit (DAG)};
   \node[comp] (c2) at (0,0) {$;$};
   \node[comp] (c1) at (0,-0.9) {$;$};
   \node[inp] (c0) at (0,-1.8) {$R$};
   \draw[-{Latex}] (c0.north west) to[out=130,in=225] (c1.south west);
   \draw[-{Latex}] (c0.north east) to[out=50,in=315] (c1.south east);
   \draw[-{Latex}] (c1.north west) to[out=130,in=225] (c2.south west);
   \draw[-{Latex}] (c1.north east) to[out=50,in=315] (c2.south east);
   \node[font=\footnotesize] at (0,-2.30) {2 composition gates};
 \end{scope}
\end{tikzpicture}
\caption{The term $(R;R);(R;R)$ repeats the same subterm, whereas
 the circuit reuses one gate as both inputs to the next composition.
 Distinct gates, not wires or uses, are counted.}
\label{fig:term-vs-circuit}
\end{figure}

More generally, $g_d$ defined by
$g_{i+1}\coloneqq g_i;g_i$ computes the pairs joined by an $R$-walk of length
$2^d$. Fully unfolding this computation gives a term with $2^d-1$
compositions, while the circuit has only $d$ composition gates.
For another example, a circuit can compute $h\coloneqq R;S$ once and use it
in both branches of $(P\cap h)\cup(Q\cap h^\smile)$, still with only
one composition gate. Thus sharing can save exponentially many
compositions; our lower bounds hold \emph{despite} allowing it.

\paragraph{Equations.}
A \emph{\CoR\ equation} is a statement $T=S$, satisfied by a
structure precisely when
\[
 \mathcal M\models T=S
  \quad\iff\quad
 \sem T{\mathcal M}=\sem S{\mathcal M}.
\]
For example, $\id\cup R=\one$ expresses the \FOtwo\ sentence
$\forall x\,\forall y\,(x=y\lor R(x,y))$. Equivalence between a
\CoR\ equation and an FO sentence means that they hold in exactly
the same structures, rather than that they define the same binary query.
A circuit computing an equation may have two designated output gates
(for $T$ and $S$), with shared intermediate gates counted only once.

\subsection{Resource measures and computational conventions}\label{sec:prelim-resources}
\paragraph{Syntax size and gate counts.}
Formula size $|\varphi|$ and term size $|T|$ count nodes in their
\emph{syntax trees}; all connectives have arity at most two. Negation
and quantification are unary, and displayed iterated connectives are
abbreviations to be expanded, as are $\to$ and $\leftrightarrow$.
Each relation-atom occurrence counts as
one node even for a growing signature (so $|G_k|=8k$); the length of
the written relation name is not charged. This is different from
measuring the encoded bit length of a formula. Write
$\relatoms\varphi$ for the number of relation-atom occurrences
(including $R(x,x)$, but excluding equality atoms), and
$\#_{\exists\forall}(\varphi)$ for its quantifier occurrences.

For a term $T$, write $\relatoms T$ for occurrences of input relation
symbols and $\cost T$ for occurrences of composition. For a circuit
$C$, $|C|$ counts its distinct gates, including input and constant
gates, and $\gcost C$ counts distinct composition gates. All other
operations are uncharged by $\gcost C$, though they do count toward
$|C|$. A shared gate counts once, irrespective of its fan-out. We will
also use the elementary tree inequalities
\[
 |T|\ge2\cost T+1,
 \qquad |T|\ge2\relatoms T-1.
\]

\paragraph{Equivalence and allowed inputs.}
Let $\Psi$ and $\Gamma$ each be a \CoR\ term, a \CoR\ circuit, or
an FO formula with at most two free variables, over a common signature
$\sigma$ and with the same ordered output interface as above. We write
\[
 \Psi\equiv\Gamma
 \quad\iff\quad
 \sem\Psi{\mathcal M}=\sem\Gamma{\mathcal M}
 \quad\text{for every nonempty $\sigma$-structure $\mathcal M$.}
\]
Thus $\equiv$ denotes agreement on arbitrary nonempty structures,
finite or infinite. Our lower bounds are proved using finite test
structures and hold even under the weaker requirement of agreement on
all finite structures. A \emph{size-specific} circuit
$C_n$ may depend on $n=|\cordom|$, but not on the input interpretation;
it must compute the intended query on \emph{every} $n$-element
input. A relational circuit has only the input symbols and the three
constants $\zero,\one,\id$; no built-in order, named vertices,
extra vertex predicates, or arbitrary vertex-dependent constants are
supplied for free. The more permissive typed matrix circuit model of
Section~\ref{sec:typed} is defined separately.

\paragraph{Randomized targets.}
A randomized circuit is an input-independent distribution
$\mathcal D_n$ over deterministic circuits, possibly of different
sizes, for the given domain size $n$. Pointwise error at most
$\varepsilon$ for a query $q$ means that, for every $n$-element input
$\mathcal M$ and every $(a,b)\in \cordom^2$,
\[
 \Pr_{C\sim\mathcal D_n}\!\left[
   \sem C{\mathcal M}(a,b)=\sem q{\mathcal M}(a,b)
 \right]\ge1-\varepsilon.
\]
The entire output relation need not be correct simultaneously.
Randomized resource bounds concern expectations over $\mathcal D_n$;
the circuit distribution is fixed independently of the input.

\section{The main exponential lower bound}\label{sec:preservation}
We prove Theorem~\ref{thm:main} directly, following the technical
mechanism of the overview. We construct the reference structure
(\S\ref{sec:structure}) and define color-determined relations, folds,
and predictions (\S\ref{sec:types}). We then prove that unnoticed folds
leave circuit values as predicted (\S\ref{sec:unnoticed}) and that each
composition notices at most two folds (\S\ref{sec:responses}). Finally,
we supply the tests for $G_k$ and complete the lower bound
(\S\ref{sec:main-proof}). Section~\ref{sec:simultaneous} subsequently
formulates the test--cost theorem for the further applications.

\subsection{The reference structure}\label{sec:structure}
Fix a palette $\Lambda$ of $N\ge2$ colors and a vertex set $V$ of size $n$.
A coloring $\chi:\binom V2\to\Lambda$ assigns a color to every unordered
pair of distinct vertices. We write $\chi(u,v)$ for $\chi(\{u,v\})$.

\Needspace{7\baselineskip}
\begin{lemma}[Two witnesses]\label{lem:palette}
For every integer
\[
 n\ge n_0(N)\coloneqq\left\lceil12N^2\ln(2N)\right\rceil+3,
\]
there is a coloring such that, for all distinct $u,v\in V$ and all
$\mu,\nu\in\Lambda$, at least two distinct vertices
$z\in V\setminus\{u,v\}$ satisfy
$\chi(u,z)=\mu$ and $\chi(z,v)=\nu$.
\end{lemma}
\begin{proof}
Color the edges independently and uniformly. For fixed distinct $u,v,w$
and colors $\mu,\nu$, the probability that no suitable witness occurs
outside $\{u,v,w\}$ is at most $\exp(-(n-3)/N^2)$: the candidate witnesses
use disjoint pairs of edges. A union bound gives failure probability at most
\[
 n^3N^2\exp(-(n-3)/N^2).
\]
At $n=n_0(N)$, the inequality $n_0(N)\le13N^3$ bounds this expression by
$2197/(4096N)<1$. It decreases for $n>3N^2$, proving the assertion at every
larger size. Hence some coloring has, for all distinct $u,v,w$ and all
colors $\mu,\nu$, a witness outside $\{u,v,w\}$; applying this with $w$
equal to one witness for $u,v$ yields a second.
\end{proof}

Fix such a coloring, and choose a vertex $\hub\in V$, called the
hub. The two witnesses of Lemma~\ref{lem:palette} play different
roles. One witness outside the endpoints already makes every computed
relation a union of color classes, possibly together with the diagonal
(Lemma~\ref{lem:tables}). The second ensures that two vertices other
than $\hub$ have, for every pair of colors, a witness that also avoids
$\hub$, so that composition agrees with prediction at pairs away
from the hub (\S\ref{sec:responses}). Lemma~\ref{lem:palette} also implies that
every vertex sees every color, and that every vertex other than $\hub$
sees every color on edges avoiding $\hub$ (take $v=\hub$). For each color
$\mu$, fix a vertex $v_\mu$ with $\chi(\hub,v_\mu)=\mu$. No vertex,
color, or color neighborhood is supplied as an additional input
predicate.

Jipsen, Maddux, and Tuza~\citep{jipsen-maddux-tuza1995} gave
the random and affine-plane colorings underlying our reference
structures. Although they formulated a one-witness condition, their
constructions also yield the two-witness property used here. The
random-coloring argument above establishes it at every sufficiently
large domain size; their affine-plane coloring has it for planes of
order at least four. Appendix~\ref{app:explicit} verifies the latter
and gives an explicit construction on fewer than $16N^2$ vertices.\footnote{ChatGPT drew our attention to this connection; see the \hyperref[sec:ai-methodology]{AI methodology statement} at the end of the paper.}

\subsection{Color-determined relations, folds, and predictions}\label{sec:types}
For $\mu\in\Lambda$, let
$e_\mu\coloneqq\{(u,v)\in V^2:u\ne v,\ \chi(u,v)=\mu\}$ be the set of edges of
color $\mu$, in both orientations. We extend $\chi$ to loops by giving
every loop the extra label $\diagc$, which is not a color. A relation
$X\subseteq V^2$ is \defn{color-determined} if $X(u,v)$ depends only
on $\chi(u,v)$. For such a relation, write $X[\mu]\in\bool$ for its
common value at pairs $(u,v)$ with $\chi(u,v)=\mu$, where
$\mu\in\Lambda\cup\{\diagc\}$. These $N+1$ values form a table, and
\[
 S_X\coloneqq\{\mu\in\Lambda:X[\mu]=1\}
\]
is the \defn{off-diagonal support} of $X$. Thus $X$ is the union of the
classes $e_\mu$ with $\mu\in S_X$, together with the diagonal if
$X[\diagc]=1$. A reference structure interprets every input relation
by a color-determined relation; we write $\mathcal M_0$ for it.

\begin{lemma}[Closure of color-determined relations]\label{lem:tables}
The constants $\zero,\one,\id$ are color-determined, and
color-determined relations are closed under Boolean operations, converse,
and composition. For color-determined $X,Y$, their composition has the
following \emph{color table} (CT), for every color $\mu\in\Lambda$:
\begin{equation*}\tag{CT}\label{eq:table-off}
\begin{aligned}
 (X;Y)[\mu]
  &=\bigl(X[\diagc]\land Y[\mu]\bigr)\lor\bigl(X[\mu]\land Y[\diagc]\bigr)
       \lor\bigl(S_X\ne\varnothing\land S_Y\ne\varnothing\bigr),\\
 (X;Y)[\diagc]
  &=\bigl(X[\diagc]\land Y[\diagc]\bigr)\lor\bigl(S_X\cap S_Y\ne\varnothing\bigr),
\end{aligned}
\end{equation*}
where a condition stands for its truth value.
\end{lemma}
\begin{proof}
Boolean operations act entry by entry, and converse preserves this
property because $\chi$ is symmetric. For composition, consider a pair
$(u,v)$ and a potential witness $z$. An endpoint witness,
$z\in\{u,v\}$, contributes the disjuncts of~\eqref{eq:table-off} that involve
$X[\diagc]$ or $Y[\diagc]$. An outside witness, $z\notin\{u,v\}$,
contributes $X[\chi(u,z)]\land Y[\chi(z,v)]$. If $u\ne v$, then by
Lemma~\ref{lem:palette} the color pairs $(\chi(u,z),\chi(z,v))$ of outside
witnesses range over all of $\Lambda\times\Lambda$, so an outside witness
exists exactly when $S_X$ and $S_Y$ are both nonempty. If $u=v$, both
edges of an outside witness are the edge $\{u,z\}$, and every color
occurs on an edge at $u$, so an outside witness exists exactly when
$S_X\cap S_Y\ne\varnothing$.
\end{proof}
In particular, every value of a relational circuit on $\mathcal M_0$ is
color-determined, and
\[
 e_\mu;e_\nu=
 \begin{cases}\one,&\mu=\nu,\\
 \overline{\id},&\mu\ne\nu.
 \end{cases}
\]
Thus the color-determined relations give a concrete representation of
the symmetric relation algebra $E_{N+1}(1,2,3)$ studied by Jipsen,
Maddux, and Tuza~\citep{jipsen-maddux-tuza1995}.
Its atoms (minimal nonempty elements under inclusion)
are $\id$ and the $e_\mu$, and $\overline{\id}=\bigcup_{\mu\in\Lambda}e_\mu$ is
the diversity relation. This identification follows from the displayed
multiplication table, not from automorphisms of the reference graph.

\paragraph{Folds and predictions.}
From now on, take the palette $\Lambda\coloneqq2^{[k]}$ for some $k\ge1$, so that
$N=2^k$. The complement of a color $\mu$ is the color
$\mu^c=[k]\setminus\mu$, and no color equals its complement. The fold of
a color $\lambda$ recolors every edge of color $\lambda$ at the hub to
$\lambda^c$ and leaves all other edges unchanged: for $u\ne v$,
\[
 \chi_\lambda(u,v)\coloneqq
 \begin{cases}\lambda^c,&\hub\in\{u,v\}\text{ and }\chi(u,v)=\lambda,\\
 \chi(u,v),&\text{otherwise},
 \end{cases}
\]
and every loop keeps the label $\diagc$. For the local sensitivity bounds,
the target color $\lambda^c$ could be replaced by any fixed color
distinct from $\lambda$, while leaving all other edges unchanged. For a color-determined relation $X$,
its prediction $\pred X\lambda$ reads the same values on $\chi_\lambda$:
\[
 \pred X\lambda(u,v)\coloneqq X[\chi_\lambda(u,v)]\qquad(u,v\in V).
\]
Thus $\pred X\lambda$ agrees with $X$ except on the recolored edges,
where it takes the value $X[\lambda^c]$. Replacing every input relation
by its prediction defines the folded structure $\mathcal M_\lambda$ on
the same domain $V$. Predictions are defined only for color-determined relations, in
particular for all values computed on $\mathcal M_0$. We do not assume
that relations computed on $\mathcal M_\lambda$ are color-determined;
in general, they are not.

\subsection{Noticing, and unnoticed folds}\label{sec:unnoticed}
Let $\lambda$ be a color and $\mathbf U=(U_1,\ldots,U_r)$ a tuple of
color-determined relations. For an $r$-ary operation $\operatorname{op}$
such that $\operatorname{op}(\mathbf U)$ is also color-determined, we say
that it notices the fold of $\lambda$ at $\mathbf U$ if
\[
 \operatorname{op}\bigl(\pred{U_1}\lambda,\ldots,\pred{U_r}\lambda\bigr)\ne
 \pred{\operatorname{op}(\mathbf U)}\lambda.
\]
Lemma~\ref{lem:tables} ensures that $\operatorname{op}(\mathbf U)$ is
color-determined for every \CoR\ operation, so its prediction is defined.
Both sides of the inequality are determined by $\mathbf U$ and $\lambda$: noticing is a
property of the operands on $\mathcal M_0$, not of any computation on
$\mathcal M_\lambda$.

\begin{proposition}[Commutation with prediction]\label{prop:prediction-commutation}
For any color $\lambda$ and color-determined relations $X,Y$, prediction
commutes with Boolean operations and converse:
\[
\begin{aligned}
 \pred{(X\cup Y)}\lambda&=\pred X\lambda\cup\pred Y\lambda,&
 \pred{(X\cap Y)}\lambda&=\pred X\lambda\cap\pred Y\lambda,\\
 \pred{(\overline X)}\lambda&=\overline{\pred X\lambda},&
 \pred{(X^\smile)}\lambda&=(\pred X\lambda)^\smile.
\end{aligned}
\]
The constants $\zero,\one,\id$ are also fixed by prediction.
\end{proposition}
\begin{proof}
Boolean operations act pair by pair, so both sides apply the same
operation to the entries at $\chi_\lambda(u,v)$. Converse commutes with
prediction because $\chi_\lambda(u,v)=\chi_\lambda(v,u)$. The constants
are fixed because a fold preserves the diagonal.
\end{proof}

Thus only composition can notice a fold. A composition gate of a
relational circuit notices the fold of $\lambda$ if the composition of
its operand values on $\mathcal M_0$ notices it.
\begin{lemma}[Unnoticed folds]\label{lem:unnoticed}
Let $C$ be a relational circuit and $\lambda$ a color. If no composition
gate of $C$ notices the fold of $\lambda$, then on $\mathcal M_\lambda$
every gate of $C$ computes the prediction of its value on $\mathcal M_0$.
In particular, the outputs of $C$ on $\mathcal M_0$ and on
$\mathcal M_\lambda$ agree at every pair that the fold does not recolor.
\end{lemma}
\begin{proof}
Evaluate $C$ on $\mathcal M_0$. By Lemma~\ref{lem:tables}, every gate
value is color-determined, so its prediction is defined. We show by
induction along a topological order of the gates that every gate computes
on $\mathcal M_\lambda$ the prediction of its value on $\mathcal M_0$. For
input relations this is the definition of $\mathcal M_\lambda$, and
constants are their own predictions by Proposition~\ref{prop:prediction-commutation}.
For Boolean and converse gates, the claim follows from the induction
hypothesis and Proposition~\ref{prop:prediction-commutation}.
If the operands of a composition gate have values $X$ and $Y$ on
$\mathcal M_0$, then by induction they have values $\pred X\lambda$ and $\pred Y\lambda$ on
$\mathcal M_\lambda$. The gate therefore computes
$\pred X\lambda;\pred Y\lambda$ on $\mathcal M_\lambda$, which equals
$\pred{(X;Y)}\lambda$ because the gate does not notice the fold. Finally,
every value $U$ on $\mathcal M_0$ agrees with $\pred U\lambda$ at every
pair that the fold does not recolor.
\end{proof}

\subsection{Each composition notices at most two folds}\label{sec:responses}
The unnoticed-fold lemma makes the circuit lower bound a local question:
how many folds can one composition notice? We answer it by comparing
surviving witnesses after a fold with those accounted for by prediction.

Fix color-determined relations $X,Y$ and a color $\lambda$, and write
$J\coloneqq\Lambda\setminus\{\lambda\}$ for the set of colors at $\hub$
after the fold of $\lambda$.
Endpoint witnesses contribute the same value to
$\pred X\lambda;\pred Y\lambda$ and to $\pred{(X;Y)}\lambda$. At a pair
$(u,v)$ with $u\ne v$ and $\mu\coloneqq\chi_\lambda(u,v)$, they contribute
$(X[\diagc]\land Y[\mu])\lor(X[\mu]\land Y[\diagc])$ to
$\pred X\lambda;\pred Y\lambda$, because $\pred X\lambda$ and
$\pred Y\lambda$ take the values $X[\diagc]$ and $Y[\diagc]$ on the
diagonal and $X[\mu]$ and $Y[\mu]$ at $(u,v)$, and the same value to
$\pred{(X;Y)}\lambda(u,v)=(X;Y)[\mu]$ by~\eqref{eq:table-off}; at a
diagonal pair, they contribute $X[\diagc]\land Y[\diagc]$ to both. The
following table compares the outside witnesses. Its middle column states
when $\pred X\lambda;\pred Y\lambda$ has an outside witness at the pair,
and its right column states the last disjunct of~\eqref{eq:table-off},
which $\pred{(X;Y)}\lambda$ reads at the pair; here $u$ and $v$ range
over vertices other than $\hub$, and $S\ne\varnothing\ne S'$ means that
both sets are nonempty.
\begin{center}
\renewcommand{\arraystretch}{1.15}
\begin{tabular}{lll}
\toprule
Output pair & Outside witness in $\pred X\lambda;\pred Y\lambda$
 & Last disjunct of~\eqref{eq:table-off}\\
\midrule
$(u,v)$, $u\ne v$ & $S_X\ne\varnothing\ne S_Y$ & $S_X\ne\varnothing\ne S_Y$\\
$(u,u)$ & $S_X\cap S_Y\ne\varnothing$ & $S_X\cap S_Y\ne\varnothing$\\
$(\hub,v)$ & $S_X\cap J\ne\varnothing\ne S_Y$ & $S_X\ne\varnothing\ne S_Y$\\
$(u,\hub)$ & $S_X\ne\varnothing\ne S_Y\cap J$ & $S_X\ne\varnothing\ne S_Y$\\
$(\hub,\hub)$ & $S_X\cap S_Y\cap J\ne\varnothing$ & $S_X\cap S_Y\ne\varnothing$\\
\bottomrule
\end{tabular}
\end{center}
For the middle column, consider the outside witnesses $z$ of a pair. For
distinct $u,v\ne\hub$, Lemma~\ref{lem:palette} gives, for every pair of
colors, an outside witness other than $\hub$, whose two edges keep their
colors; the witness $\hub$ contributes one more color pair, which is
already realized. For $(u,u)$ with $u\ne\hub$, every color occurs on an
edge at $u$ that avoids $\hub$. At a pair $(\hub,v)$ of the
\defn{hub row}, the left edge $\{\hub,z\}$ of every outside witness lies
at the hub, so its color becomes $\chi_\lambda(\hub,z)\in J$, while the
right edge $\{z,v\}$ keeps its color; by Lemma~\ref{lem:palette} the
realized color pairs are exactly $J\times\Lambda$. The pairs $(u,\hub)$
of the \defn{hub column} are symmetric, with $\Lambda\times J$. At
$(\hub,\hub)$ both edges of an outside witness are the edge
$\{\hub,z\}$, whose color $\chi_\lambda(\hub,z)$ ranges exactly over $J$.

\begin{lemma}[When a composition notices a fold]\label{lem:local}\label{lem:notice}
Let $X,Y$ be color-determined. The composition $X;Y$
notices the fold of $\lambda$ only if
\[
 S_X=\{\lambda\},\qquad S_Y=\{\lambda\},\qquad\text{or}\qquad
 S_X\cap S_Y=\{\lambda\}.
\]
Hence it notices at most two folds, at most one if $S_X$ or $S_Y$ is
$\varnothing$ or $\Lambda$, and none if both are.
\end{lemma}
\begin{proof}
By the table, $\pred X\lambda;\pred Y\lambda$ and
$\pred{(X;Y)}\lambda$ can differ only in the hub row, in the hub column,
or at $(\hub,\hub)$. In the hub row they can differ only if
$S_X\ne\varnothing$ but $S_X\cap J=\varnothing$, that is,
$S_X=\{\lambda\}$; in the hub column only if $S_Y=\{\lambda\}$; and at
$(\hub,\hub)$ only if $S_X\cap S_Y=\{\lambda\}$. Each condition
determines $\lambda$. If $S_X\cap S_Y\ne\varnothing$, each of the first
two conditions implies the third, so at most one fold is noticed;
otherwise the third condition never holds, and at most two are. If $S_Y$
is $\varnothing$ or $\Lambda$, then $S_Y\ne\{\lambda\}$ because $N\ge2$,
and $S_X\cap S_Y$ is $\varnothing$ or $S_X$, so only the condition
$S_X=\{\lambda\}$ remains; the case of $S_X$ is symmetric. If both are
$\varnothing$ or $\Lambda$, none of the conditions can hold.
\end{proof}
The bound two is attained. The composition $e_\mu;e_{\mu^c}$ has
$S_X=\{\mu\}$ and $S_Y=\{\mu^c\}$, and notices the folds of both $\mu$
and $\mu^c$. What a gate notices depends only on its operand values on
$\mathcal M_0$, so a shared gate counts once, regardless of its fan-out.
We now supply $2^k$ tests that force every fold to be noticed by some
composition gate, completing the lower bound.

\subsection{The tests for \texorpdfstring{$G_k$}{Gk} and the main lower bound}\label{sec:G}\label{sec:main-proof}
The input interpretation below makes each fold of $\lambda$ change
$G_k$ at an unchanged edge of complementary color $\lambda^c$.

On a color $\mu$, interpret $P_i$ as $1$ precisely when $i\in\mu$, and
$Q_i$ as $1$ precisely when $i\notin\mu$. Put both values zero on every
loop. More explicitly,
\[
 P_i^{\mathcal M_0}\coloneqq\bigcup_{\mu\in\Lambda:\,i\in\mu}e_\mu,\qquad
 Q_i^{\mathcal M_0}\coloneqq\bigcup_{\mu\in\Lambda:\,i\notin\mu}e_\mu.
\]
With these inputs, the DNF relation $H_\mu$ of \S\ref{sec:main-results}
is exactly the color class $e_\mu$.
For each color $\lambda$, observe the output at $(\hub,v_{\lambda^c})$,
of color $\lambda^c$ (\S\ref{sec:structure}), before and after the fold
of $\lambda$ (\S\ref{sec:types}). Since $\lambda\ne\lambda^c$, the fold
does not recolor this pair.

\begin{lemma}[The tests for $G_k$]\label{lem:G-tests}
For every color $\lambda$,
\[
 \sem{G_k}{\mathcal M_0}(\hub,v_{\lambda^c})=1,\qquad
 \sem{G_k}{\mathcal M_\lambda}(\hub,v_{\lambda^c})=0.
\]
Moreover, the sentence $G_k^\forall=\forall x\,\forall y\,(x=y\lor
G_k(x,y))$ of \S\ref{sec:main-results} is true in $\mathcal M_0$
and false in every $\mathcal M_\lambda$. Thus $G_k$ and $G_k^\forall$
each supply $2^k$ tests.
\end{lemma}
\begin{proof}
For non-loop pairs $(u,v)$ and $(u,z)$ of colors $\mu,\nu$, respectively,
the body of $G_k(u,v)$ holds exactly when $\nu=\mu^c$.
Since all $P_i,Q_i$ vanish on loops, $z=u$ cannot satisfy the body.
Hence, on $\mathcal M_0$ and on every $\mathcal M_\lambda$, $G_k$ holds
at a pair $(u,v)$ of distinct vertices exactly when $u$ has an edge
whose color is the complement of the color of $(u,v)$.
At $(\hub,v_{\lambda^c})$, $G_k$ therefore requires a $\lambda$-colored
edge incident to $\hub$. That color occurs at $\hub$ in
$\mathcal M_0$ and does not occur there after the fold. In
$\mathcal M_0$, every vertex has edges of every color, so $G_k$ holds at
every pair of distinct vertices, and $x=y$ handles the diagonal; in
$\mathcal M_\lambda$, $G_k$ fails at the pair $(\hub,v_{\lambda^c})$ of
distinct vertices.
\end{proof}

\mainbounds*
\begin{proof}
Let $C\equiv G_k$, or more generally let $C$ agree with $G_k$
on all finite structures. Take the reference structure at any size
$n\ge n_0(2^k)$. For each color $\lambda$, the answers of $G_k$
at $(\hub,v_{\lambda^c})$ on $\mathcal M_0$ and on $\mathcal M_\lambda$ differ
by Lemma~\ref{lem:G-tests}, and the fold of $\lambda$ does not recolor
this pair. By Lemma~\ref{lem:unnoticed}, some composition gate of $C$
notices the fold of $\lambda$. By Lemma~\ref{lem:notice}, each
composition gate notices at most two folds, so $2\gcost C\ge2^k$. A term
is a circuit without sharing, so the same bound holds for its number of
compositions.
\end{proof}
Together with the expansion of $G_k$ in \S\ref{sec:main-results}, these
lower bounds show that the minimum circuit size for $G_k$ is
$\Theta(2^k)$; Theorem~\ref{thm:tree} in \S\ref{sec:trees}
asymptotically determines the minimum term size.

The lower bound does not rest on the difficulty of evaluating $G_k$. As
its proof shows, it holds even for circuits that are correct only on the
reference structure and its folded structures, on which every non-loop
edge has a color from the palette. For distinct $u,v$ in these structures,
$G_k(u,v)$ holds exactly when $u$ has an edge of the color complementary
to that of $(u,v)$; on loops it is false. Thus a trie of the colors
present in each row evaluates $G_k$ on all pairs of these structures in
time $O(kn^2)$, linear in the size of the input.

\paragraph{Why the requirements are crossed.}
For comparison, the superficially similar non-crossed formula collapses:
\[
 \exists z\,\bigwedge_{i=1}^k
 [(P_i(x,y)\land P_i(x,z))\lor(Q_i(x,y)\land Q_i(x,z))]
 \equiv\bigwedge_{i=1}^k(P_i(x,y)\lor Q_i(x,y)).
\]
The forward implication is immediate; for the reverse implication take
$z=y$. In $G_k$, however, the two edges are subject to crossed
requirements, so choosing $z=y$ need not satisfy them.

\section{The test--cost theorem and further relational bounds}\label{sec:logical}
Section~\ref{sec:preservation} proved the main composition lower bound.
We now formulate its counting principle for other operations and
randomized targets (\S\ref{sec:simultaneous}), then derive the bounds
for equations, bounded error, and circuits without converse
(\S\ref{sec:query-bounds}). We also compare large and small domains
(\S\ref{sec:domain-size}), analyze random inputs
(\S\ref{sec:random-inputs}), and allow reflexive--transitive closure
(\S\ref{sec:closure}).

\subsection{Tests and the test--cost theorem}\label{sec:simultaneous}
The proof of Theorem~\ref{thm:main} counts folds that must be noticed.
For randomized circuits and the other targets of
\S\ref{sec:additional-results}, we now formalize this
\emph{fold-counting argument} using only three properties of the tests.
For a color $\lambda$,
the \defn{test} for $\lambda$ consists of the structures $\mathcal M_0$
and $\mathcal M_\lambda$ and the \defn{output observation}
$\rho_\lambda(U)\coloneqq U(\hub,v_{\lambda^c})$, the entry of an output
value $U$ at the pair $(\hub,v_{\lambda^c})$ of \S\ref{sec:structure}.
Since $\lambda^c\ne\lambda$, the fold of $\lambda$ does not recolor this
pair, so the \emph{observation} (Obs) condition
\begin{equation*}\tag{Obs}\label{eq:observation}
 \rho_\lambda(\pred U\lambda)=\rho_\lambda(U)
\end{equation*}
holds for every reference output value $U$, as defined in~(i) below.

The counting argument uses only three properties of this construction.
We call data satisfying them a \defn{test system}:
\begin{enumerate}[label=(\roman*),leftmargin=*]
\item Each type of value has a collection of \defn{reference values},
which contains the input values on $\mathcal M_0$ and is closed under the
permitted operations, and every reference value $U$ has a prediction
$\pred U\lambda$ for each color $\lambda$.
\item Every input value on $\mathcal M_\lambda$ is the prediction of the
same input value on $\mathcal M_0$.
\item Every test satisfies~\eqref{eq:observation}, and the required
answers at its observation on $\mathcal M_0$ and on $\mathcal M_\lambda$
differ.
\end{enumerate}
For relational circuits, the reference values are the color-determined
relations, and Lemma~\ref{lem:tables} supplies~(i). The typed scalar
computations in \S\ref{sec:matrix} will satisfy the same conditions;
predictions carry the type of their argument, which we suppress in the
notation. For a gate $g$ of a deterministic circuit $C$, let
$\operatorname{op}_g$ be its operation and $\mathbf U_g$ the tuple of its
operand values on $\mathcal M_0$. As in \S\ref{sec:unnoticed}, $g$
notices the fold of $\lambda$ if
$\operatorname{op}_g(\pred{\mathbf U_g}\lambda)\ne
\pred{\operatorname{op}_g(\mathbf U_g)}\lambda$, where predictions act
componentwise on tuples. Define
\[
 \begin{split}
 D_g&\coloneqq\{\lambda\in\Lambda:g\text{ notices the fold of }\lambda\},\\
 \tcost(C)&\coloneqq\sum_{g\text{ distinct}}|D_g|.
 \end{split}
\]
For input nodes put $D_g\coloneqq\varnothing$; constant gates are nullary
operations. The definition of $D_g$ compares the entire output of the
gate, and it uses predicted operands computed from $\mathcal M_0$, not
operands recomputed on $\mathcal M_\lambda$.

We call $\tcost(C)$ the \defn{cost} of $C$. It counts the pairs of a
distinct gate and a fold that the gate notices. This integer depends on
the reference structure and the tests, which remain fixed, but is
defined even for an incorrect circuit.

\begin{theorem}[Test--cost theorem]\label{thm:test-cost}
Fix a test system. A deterministic circuit $C$ correct at both
observations of every test satisfies
\[
 2^k\le\tcost(C).
\]
If $C$ is drawn from an input-independent distribution $\mathcal D$ of
circuits, and each observation in each test is incorrect with probability
at most $\varepsilon<1/2$, then
\[
 (1-2\varepsilon)2^k\le
 \mathbb E_{C\sim\mathcal D}[\tcost(C)].
\]
\end{theorem}
\begin{proof}
Fix a deterministic circuit, not assumed correct, and put
$D_C\coloneqq\bigcup_gD_g$. For $\lambda\notin D_C$, no gate notices the
fold of $\lambda$, so the induction in the proof of
Lemma~\ref{lem:unnoticed}, with~(ii) for the input values, gives
\[
 \sem C{\mathcal M_\lambda}=\pred{(\sem C{\mathcal M_0})}\lambda.
\]
The equality concerns the whole output. By the observation
condition~\eqref{eq:observation}, the two observations are equal, whereas
their correct values differ. Hence at least one is wrong. If
$\mathrm{err}_{\lambda,0}(C),\mathrm{err}_{\lambda,1}(C)$ are the
corresponding zero-one error indicators, then
\[
 \tcost(C)+\sum_{\lambda\in\Lambda}
   \bigl(\mathrm{err}_{\lambda,0}(C)+\mathrm{err}_{\lambda,1}(C)\bigr)\ge2^k,
\]
because $|D_C|\le\tcost(C)$. Zero error proves the first assertion.
Take expectations using the same random circuit on the two structures
of each test. Each error indicator has expectation at most $\varepsilon$,
which proves the second assertion. No independence among the tests or
their errors is required.
\end{proof}

For ordinary \CoR, call a composition gate whose operand values on
$\mathcal M_0$ are $X,Y$ \defn{one-sided} if $S_X$ or $S_Y$ is
$\varnothing$ or $\Lambda$. Let
$m_1(C)$ count these gates and $m_2(C)$ the other composition gates,
counting shared gates once.
\begin{lemma}[Cost of compositions]\label{lem:cor-cost}
Every relational circuit $C$ satisfies
\[
 \tcost(C)\le m_1(C)+2m_2(C)\le2\gcost{C}.
\]
\end{lemma}
\begin{proof}
By Lemma~\ref{lem:local}, a one-sided composition notices at most one
fold and any other composition at most two. Other gates notice nothing.
Summing over distinct gates proves the inequalities.
\end{proof}
The row and column projections $X;\one$ and $\one;X$ of
\S\ref{sec:prelim-cor} are examples of one-sided compositions; the
one-sided condition itself is semantic and may hold for other gates.
If a syntactic count is preferred, one may charge all non-projection
compositions at cost two, without changing Lemma~\ref{lem:cor-cost}.

For lower bounds obtained from the color-based reference structures
of \S\ref{sec:structure}, this argument also gives fixed-size and
randomized forms. Each such bound estimates $\tcost(C)$ from above by
a resource count, such as $2\gcost C$ in Lemma~\ref{lem:cor-cost}. First, Lemma~\ref{lem:palette} gives the
reference structure on every domain of size $n\ge n_0(2^k)$, and the
folds keep the domain, so each bound holds for targets that need to be
correct only on $n$-element inputs and may be chosen for $n$. Second, by
the randomized assertion of Theorem~\ref{thm:test-cost}, each bound for
targets that compute a query also holds for an input-independent
distribution of targets with pointwise error at most
$\varepsilon<1/2$, after the resource count is replaced by its
expectation and the bound is multiplied by $1-2\varepsilon$; for scalar
outputs, an error is a failure to return the exact required value.

\subsection{Consequences}\label{sec:query-bounds}
\paragraph{Equations.}
By the semantics of \CoR\ equations in \S\ref{sec:prelim-cor},
the same tests also bound the equational form of Tarski's correspondence.
\begin{corollary}[\CoR\ equations]\label{cor:equations}
If a \CoR\ equation $T=S$ holds on exactly those finite structures that
satisfy $G_k^\forall$, then every relational circuit with gates computing
$T$ and $S$ has at least $2^{k-1}$ composition gates.
\end{corollary}
\begin{proof}
Let $C$ compute $T$ and $S$. By Lemma~\ref{lem:G-tests}, the equation
holds on $\mathcal M_0$ and fails on every $\mathcal M_\lambda$.
For any color $\lambda$, if no composition gate of $C$ noticed its fold,
Lemma~\ref{lem:unnoticed}, applied to both output gates, would give
\[
 \sem{T}{\mathcal M_\lambda}
 =\bigl(\sem{T}{\mathcal M_0}\bigr)^{(\lambda)}
 =\bigl(\sem{S}{\mathcal M_0}\bigr)^{(\lambda)}
 =\sem{S}{\mathcal M_\lambda},
\]
contradicting failure of the equation on $\mathcal M_\lambda$.
Thus every one of the $2^k$ folds is noticed by a composition gate of
$C$, and Lemma~\ref{lem:local} gives $2^k\le2\gcost C$.
\end{proof}
The DNF of \S\ref{sec:main-results} gives the equivalent equation
$\id\cup\bigcup_{A\subseteq[k]}[H_A\cap(H_{A^c};\one)]=\one$
with $2^k$ compositions, so the bound is tight up to a factor of two.

\paragraph{Fixed sizes and bounded error.}
The reference structure and its folded structures share one domain,
which may have any size $n\ge n_0(2^k)$, and the randomized assertion of
Theorem~\ref{thm:test-cost} handles bounded error.
\randomizedbounds*
\begin{proof}
We may take $c=18$, since for every $k\ge1$,
\begin{equation}\label{eq:domain-threshold}
 n_0(2^k)=\lceil12\cdot4^k(k+1)\ln2\rceil+3\le18k4^k.
\end{equation}
At the prescribed size $n\ge18k4^k$, Lemma~\ref{lem:palette} gives
the reference structure, and Lemma~\ref{lem:G-tests} gives all $2^k$
tests on that domain. Apply the expected-cost assertion of
Theorem~\ref{thm:test-cost} and $\tcost(C)\le2\gcost C$.
Setting $\varepsilon=0$ includes every exact size-specific circuit.
\end{proof}

\paragraph{Circuits without converse.}
Without converse, a directed version of the reference structure lets
each composition notice at most one fold, and the bound becomes exact.
\begin{restatable}[Without converse]{corollary}{withoutconverse}\label{cor:without-converse}
For every $k\ge1$, every relational circuit without converse that
defines $G_k$ on finite structures has at least $2^k$ composition gates.
This bound is attained by the DNF.
\end{restatable}
Appendix~\ref{app:directed} proves the corollary with a directed
reference coloring. For the original undirected test family of
Section~\ref{sec:preservation}, Proposition~\ref{prop:promise} gives
a circuit with $2^{k-1}$ compositions computing $G_k$ on the reference
structure and all its folds. Thus this family attains the coefficient
two of Lemma~\ref{lem:notice} and cannot yield a composition lower bound
larger than $2^{k-1}$ for full \CoR. The optimum for $G_k$ on arbitrary inputs remains between
$2^{k-1}$ and $2^k$.

\subsection{Domain size}\label{sec:domain-size}
The main lower bound persists even when a circuit need only be correct
on inputs of one sufficiently large size. We first prove this assertion,
then contrast it with a small-domain upper bound.

\domainsize*
\begin{proof}
Fix $k\ge1$ and $n\ge18k4^k$. By~\eqref{eq:domain-threshold} and
Lemma~\ref{lem:palette}, the reference structure can be chosen on exactly
$n$ vertices. Its folds have the same domain. Any circuit correct on all
$n$-element structures is therefore correct on every test of
Lemma~\ref{lem:G-tests}. By Lemma~\ref{lem:unnoticed}, each of the $2^k$
folds must be noticed by a composition gate, and Lemma~\ref{lem:notice}
allows each gate to notice at most two. Hence $\gcost C\ge2^{k-1}$,
so $c=18$ suffices.
\end{proof}

\paragraph{Circuits for small domains.}
On arbitrary inputs, the \defn{pair profile} is the ordered bit vector
\[
 \pi(x,z)\coloneqq
 ([P_1(x,z)],\ldots,[P_k(x,z)],
  [Q_1(x,z)],\ldots,[Q_k(x,z)])\in\{0,1\}^{2k}.
\]
On the reference structure, the profile of an edge is determined by its color.
A vertex $z$ witnesses $G_k(x,y)$ exactly when, for every $i$,
$\pi(x,y)$ has its $P_i$-bit and $\pi(x,z)$ its $Q_i$-bit set,
or vice versa. Hence $G_k(x,y)$ depends
only on $\pi(x,y)$ and the set $\Pi_x\coloneqq\{\pi(x,z):z\in V\}$ of profiles in
row $x$, which has at most $|V|$ elements.

\begin{restatable}[Small domains]{proposition}{smalldomains}\label{prop:small-domains}
For all $k,n\ge1$ there is a converse-free relational circuit with $4kn$
composition gates, all of them row projections $X;\one$, and $O(k^2n^2)$
other gates, that computes $G_k$ on every structure with at most $n$
elements.
\end{restatable}
\begin{proof}[Construction overview]
For each row $x$, retain the prefixes of length $\ell$ of the profiles
in $\Pi_x$. There are at most $n$ such prefixes at each length, starting
with the empty word. Store them in $n$ slots using row-constant Boolean
relations: their values in row $x$ do not depend on the second
coordinate. Extending every occupied slot by a bit creates at most
$2n$ candidates. One row projection per candidate tests whether some
profile in the row has that prefix, simultaneously in all rows.

At most $n$ candidates survive. Boolean gates alone count the live
candidates in a fixed slot-and-bit order, determine their ranks, and
copy them into the next set of slots. This uses no order on the input
vertices and no additional compositions. After $2k$ levels, the slots
contain all profiles in $\Pi_x$; Boolean gates compare them with the
profile $\pi(x,y)$ to return $G_k(x,y)$.

Thus the construction uses $2k\cdot2n=4kn$ row projections. The
Boolean implementation of the rank and slot updates costs
$O((\ell+1)n^2)$ gates at level $\ell$, for a total of $O(k^2n^2)$.
Appendix~\ref{app:small-domains} gives the complete gate-level
construction, its invariant, and the resource count.
\end{proof}

\subsection{Random inputs}\label{sec:random-inputs}
We now fix the circuit and randomize the input, in contrast to the
randomized targets of \S\ref{sec:prelim-resources}. The same fold tests
give matching success-probability bounds up to the failure probability
of the two-witness condition.
\begin{corollary}[Random inputs]\label{cor:random-inputs}
Fix $k\ge1$ and $n$. Color the edges of the complete graph on $n$
vertices independently and uniformly with the colors in $2^{[k]}$,
determine $P_i,Q_i$ by the colors as in \S\ref{sec:overview}, and
fix a vertex $\hub$. Choose a color $\lambda$ uniformly, a vertex $v$
uniformly among those with $\chi(\hub,v)=\lambda^c$ (any vertex if there
is none), and a fair bit; if the bit is one, apply the fold of
$\lambda$. Then every relational circuit with $g$ composition gates
answers $G_k(\hub,v)$ correctly with probability at most
$\frac12+g/2^k+\delta$,
where $\delta\coloneqq n^34^k\exp(-(n-3)/4^k)$, and for each
$g\le2^{k-1}$ some circuit with $g$ composition gates answers correctly
with probability at least $\frac12+g/2^k-\delta$.
\end{corollary}
\paragraph{Proof idea.}
Except on a coloring event of probability at most $\delta$, the
two-witness property supplies all $2^k$ tests. The error-count
inequality of Theorem~\ref{thm:test-cost} then bounds the success
probability of a circuit with $g$ compositions. For the matching bound,
retain $g$ complementary-color pairs in the circuit of
Proposition~\ref{prop:promise}. Appendix~\ref{app:random-inputs} gives
the full proof immediately after that circuit's construction.

\subsection{Reflexive--transitive closure}\label{sec:closure}
Reflexive--transitive closure can capture paths of arbitrary length
in a single gate. We show that even this extension cannot bypass the
fold-counting lower bound: each closure gate notices at most one fold.

For a relation $X$, define $X^0\coloneqq\id$ and
$X^{m+1}\coloneqq X^m;X$ for $m\ge0$. Its reflexive--transitive
closure is $X^*\coloneqq\bigcup_{m\ge0}X^m$.
\closure*
\begin{proof}
On the reference structure,
\eqref{eq:table-off} gives $X;X=\one$ whenever
the off-diagonal support $S_X$ is nonempty. Hence $X^*=\one$ if
$S_X\ne\varnothing$ and $X^*=\id$ otherwise; in particular, closure
preserves color-determined relations, as condition~(i) of a test system requires. Let
$\lambda$ be a color, and put $J\coloneqq\Lambda\setminus\{\lambda\}$. If
$S_X=\varnothing$, then $\pred X\lambda=X$, and both
$(\pred X\lambda)^*$ and $\pred{(X^*)}\lambda$ equal $\id$. If $S_X$
contains a color in $J$, the table of \S\ref{sec:responses} gives
$\pred X\lambda;\pred X\lambda=\one$, so
$(\pred X\lambda)^*=\one=\pred{(X^*)}\lambda$. Thus a closure gate
notices the fold of $\lambda$ only when $S_X=\{\lambda\}$, so it notices
at most one fold. Theorem~\ref{thm:test-cost} and Lemma~\ref{lem:cor-cost}
give $2^k\le m_1(C)+2m_2(C)+g_*(C)\le2\gcost C+g_*(C)$ with the tests of
Lemma~\ref{lem:G-tests}.
\end{proof}

\section{Fixed vocabularies and model checking}\label{sec:fixed-vocabulary}
We prove Theorem~\ref{thm:fixed-vocabulary} in the stronger form
stated below. The upper bound is constructive and holds for every fixed
finite binary vocabulary, while the lower bound already holds over a
single binary relation. For the latter, we encode the growing vocabulary
of $G_k$ into one relation. Throughout this section, $m$ counts nodes
in the \emph{source syntax tree}; only the target may share
subexpressions.

\begin{theorem}[Fixed-vocabulary succinctness]
\label{thm:fixed-vocabulary:formal}
For every sufficiently large $m$, there is an \FOthree\ query of
syntax-tree size at most $m$ over one binary relation such that every
equivalent relational circuit over finite structures has at least
$2^{\Omega(m/\log m)}$ composition gates.
Conversely, over any fixed finite binary vocabulary $\sigma$, every
\FOthree\ formula of syntax-tree size $m$ with at most two free variables
has an equivalent relational circuit of total size
$2^{O_\sigma(m/\log m)}$, constructible within the same time bound.
The upper-bound translation is correct on all nonempty structures.
\end{theorem}

We first give the constructive upper bound, then transfer the main
lower bound through a single-relation encoding. The final subsection
derives the model-checking consequence.

\subsection{The fixed-vocabulary upper bound}\label{sec:fixed-upper}
Fix a finite binary vocabulary $\sigma$. Since the variable names
are restricted to $x,y,z$, formulas are trees over a fixed finite
alphabet whose symbols have fixed arities. Consequently, for some constant $a_\sigma>1$, there
are at most $a_\sigma^{t+1}$ syntactically distinct formulas of size
at most $t$. It follows that $d$ distinct formulas
$\psi_1,\ldots,\psi_d$ satisfy
\[
 \sum_{j=1}^d|\psi_j|
 \ge c_\sigma d\log(d+1)
\]
for a constant $c_\sigma>0$.
Indeed, for all sufficiently large $d$, at most $d/2$ formulas have
size below a sufficiently small constant times $\log d$.
At least $d/2$ of the chosen formulas therefore have size
$\Omega_\sigma(\log d)$; reducing $c_\sigma$ handles the remaining
values of $d$.

\begin{proof}[Proof of the upper bound in Theorem~\ref{thm:fixed-vocabulary:formal}]
We translate quantified subformulas from the inside out, sharing each
translated child among all clauses that use it. The key is to expand a
Boolean formula in its distinct inputs rather than its input
occurrences.
For a quantified subformula $Qw\,\psi$ with
$Q\in\{\exists,\forall\}$, form the Boolean shell of $\psi$ by replacing
each maximal quantified subformula occurrence by a single Boolean input;
the remaining atoms are its other inputs. Here maximal means not
contained in another quantified subformula of $\psi$.
Let $d$ be the number of \emph{distinct} quantified children.
Choosing one occurrence of each gives disjoint subtrees of the source,
whose total size is at most $m$. The preceding inequality yields
\[
 d\log(d+1)=O_\sigma(m),
 \qquad d=O_\sigma\!\left(\frac{m}{\log m}\right).
\]
There are only $O_\sigma(1)$ possible relation and equality atoms.
Thus the shell is a Boolean function of
$r=d+O_\sigma(1)=O_\sigma(m/\log m)$ distinct inputs.
Repeated occurrences of the same child are assigned the same input.

By the
binary-assignment convention of \S\ref{sec:prelim-logic}, each
quantified child has at most two free variables and already has a
binary representation. A unary child has a dummy coordinate, so its
orientation can be changed by converse. On each structure, a sentence's
representation has value $\zero$ or $\one$ according to its truth on
that structure. Reflexive atoms have constant-size
representations as in \S\ref{sec:FOtwo-translation}.

The truth-table DNF of a shell has at most $2^r$ clauses and $r$
literals per clause. At an existential quantifier $\exists w$, let
$u,v$ be the other two variable names and group the literals of a clause
by their free-variable pairs, obtaining
$A(u,v)\land B(u,w)\land D(w,v)$. Eliminating $w$ gives
$A\cap(B;D)$, with empty groups represented by $\one$.
Thus each clause introduces at most one composition.
Negated literals and universal quantifiers use complement.
Each child circuit is computed once and reused in every clause:
only $O((r+1)2^r)$ new gates, including at most $2^r$ compositions,
are introduced at this shell.

There are at most $m$ quantified occurrences. The unquantified root
shell, whose free variables are contained in a pair, is translated
directly by Boolean operations. All atomic representations together
cost $O_\sigma(1)$ gates if shared. Hence both the total size and the
composition count are bounded by
\[
 |C_\varphi|
 \le \operatorname{poly}_\sigma(m)\,
       2^{O_\sigma(m/\log m)}
 =2^{O_\sigma(m/\log m)}.
\]
Correctness on all nonempty structures follows by induction from the
quantifier-elimination identity~\eqref{eq:prelim-elimination}. Appendix~\ref{app:classical-translation}
gives the corresponding classical translation without sharing.
Syntactic comparison of children, enumeration of shell truth
assignments, and construction of the circuit take at most a polynomial
factor in $m$ times this bound, which is absorbed into the exponent.
\end{proof}

The counting step uses disjoint subtrees of the source. It therefore
does not give the same upper bound when $m$ instead measures a
DAG-represented source formula. Nor does the construction guarantee a
subexponential-size \CoR\ term: unfolding the shared children can
increase its size.

\subsection{The lower bound over one binary relation}\label{sec:single-relation}
We encode the input relations of $G_k$ in a single relation on copies
of the original domain. The following block simulation controls the
increase in the number of composition gates, even with unrestricted
complement and sharing.

\begin{lemma}[Block simulation]\label{lem:block-simulation}
Suppose an encoding maps each structure with domain $V$ to a structure
on $V\times L$, where $|L|=q$, with one input relation $E$.
Assume that each block
$E_{ab}(u,v)\coloneqq E((u,a),(v,b))$ is an original input relation or one of
$\zero,\one,\id$, with the choice depending only on $a,b\in L$.
Every relational circuit $C(E)$ can then be simulated over the
original signature so as to compute any designated output block,
using at most $q^3\gcost C$ composition gates.
\end{lemma}
\begin{proof}
For every gate $X$, compute all $q^2$ relations
$X_{ab}(u,v)\coloneqq X((u,a),(v,b))$.
The blocks of $\zero$ and $\one$ are respectively $\zero$ and $\one$;
the identity has block $\id$ when $a=b$ and $\zero$ otherwise.
Boolean operations act blockwise, including complement relative to
$V^2$, and converse exchanges the layer indices:
\[
 (\overline X)_{ab}=\overline{X_{ab}},
 \qquad
 (X^\smile)_{ab}=(X_{ba})^\smile.
\]
Composition satisfies
\[
 (X;Y)_{ab}=\bigcup_{c\in L}(X_{ac};Y_{cb}).
\]
Thus each composition gate introduces at most $q^3$ compositions
over all output blocks; the other operations introduce none.
The blocks of a shared gate are computed once and reused.
Induction over the circuit proves correctness.
\end{proof}

\paragraph{The encoding.}
List the $2k$ input symbols of $G_k$ as
$R_1,\ldots,R_{2k}$, with $R_i\coloneqq P_i$ and $R_{k+i}\coloneqq Q_i$.
For binary words, $\epsilon$ denotes the empty word, $|t|$ the length
of $t$, and juxtaposition concatenation. Thus $tb$ appends the bit $b$
and $bs$ prepends it. Write $\{0,1\}^{\le\ell}$ for the words of length
at most $\ell$. Put
\[
 \ell\coloneqq\lceil\log(2k)\rceil,\qquad
 L\coloneqq\{0,1\}^{\le\ell},\qquad
 q\coloneqq|L|=2^{\ell+1}-1<8k,
\]
and choose distinct leaves $w_1,\ldots,w_{2k}\in\{0,1\}^{\ell}$.
For a structure $\mathcal M$ on $V$, the encoded structure
$\operatorname{Enc}(\mathcal M)$ has domain $V\times L$.
Its sole relation $E$ consists of the following edges:
\[
\begin{aligned}
 &((u,t),(u,tb))
       &&(u\in V,\ |t|<\ell,\ b\in\{0,1\}),\\
 &((u,t),(u,t))
       &&(u\in V,\ t\ne\epsilon,\ t\text{ ends in }1),\\
 &((u,w_j),(v,\epsilon))
       &&((u,v)\in R_j^{\mathcal M},\ 1\le j\le2k).
\end{aligned}
\]
There are no other edges. Each $u$ has a binary routing tree; loops
mark its $1$-children, and the edges from leaf $w_j$ back to the root
layer store $R_j$.
Every layer-to-layer block is $\zero$, $\id$, or some $R_j$, so
Lemma~\ref{lem:block-simulation} applies.
The encoded graphs are directed and may have loops.

\paragraph{Three-variable decoders.}
Let
\[
 N(u,v)\coloneqq E(u,v)\land u\ne v,\qquad
 B_0(v)\coloneqq\neg E(v,v),\qquad B_1(v)\coloneqq E(v,v).
\]
For distinct variable names $u,v\in\{x,y,z\}$, define
\[
 \theta_\epsilon(u,v)\coloneqq N(u,v),\qquad
 \theta_{bs}(u,v)\coloneqq
 \exists w\bigl(N(u,w)\land B_b(w)\land\theta_s(w,v)\bigr),
\]
where $w$ is the third variable name.
Inside $\theta_s(w,v)$, the old name $u$ may be rebound; all such
bindings are local to that subformula.
For every word $t$, expanding $\theta_t$ gives a genuine
\FOthree\ formula of size $\Theta(|t|+1)$.

Starting at $(u,\epsilon)$, the first $\ell$ non-loop steps of
$\theta_{w_j}$ uniquely follow the word $w_j$ through the routing tree.
Its final non-loop edge returns to the root layer and reads $R_j$.
Consequently,
\[
 \operatorname{Enc}(\mathcal M)\models
 \theta_{w_j}((u,\epsilon),(v,\epsilon))
 \quad\Longleftrightarrow\quad
 \mathcal M\models R_j(u,v).
\]
Furthermore, each selected decoder $\theta_{w_j}$ starting at the
root layer can hold only when its second argument is also in that layer.
The inequalities in $N$ exclude the marking loops from both the
routing steps and the final data step; data edges can still encode
loops of the original relations, since their two layer indices differ.

\begin{proof}[Proof of the lower bound in Theorem~\ref{thm:fixed-vocabulary:formal}]
Replace each relation-atom occurrence $R_j(u,v)$ in $G_k$ by
$\theta_{w_j}(u,v)$, obtaining $\widehat G_k$ over the single symbol $E$.
Each decoder uses only three variable names. Any rebinding of the
third name is local to that atom replacement, so it does not affect
another conjunct or disjunct of $G_k$.
There are $4k$ relation-atom occurrences, hence
\[
 m_k\coloneqq|\widehat G_k|=\Theta(k\log(k+1)).
\]
When $x,y$ lie in the root layer, every successful decoded atom on
$(x,z)$ forces $z$ into the root layer as well. Since $k\ge1$,
\[
 \operatorname{Enc}(\mathcal M)\models
 \widehat G_k((u,\epsilon),(v,\epsilon))
 \quad\Longleftrightarrow\quad
 \mathcal M\models G_k(u,v).
\]

Let $C$ compute $\widehat G_k$ on all finite $E$-structures.
Lemma~\ref{lem:block-simulation} produces a circuit $C^\downarrow$
over the original signature computing its root--root output block,
with $\gcost{C^\downarrow}\le q^3\gcost C$.
For every finite original structure $\mathcal M$, the displayed
equivalence shows that $C^\downarrow$ agrees with $G_k$ on
$\mathcal M$. The finite-structure form of Theorem~\ref{thm:main}
therefore yields
\[
 2^{k-1}\le\gcost{C^\downarrow}\le q^3\gcost C,
 \qquad
 \gcost C\ge\frac{2^{k-1}}{q^3}
       >\frac{2^{k-1}}{(8k)^3}.
\]
Since $k=\Theta(m_k/\log m_k)$, this is
$2^{\Omega(m_k/\log m_k)}$.
For every sufficiently large size bound $m$, choose
$k=\lfloor c m/\log m\rfloor$ with a sufficiently small absolute
constant $c>0$. Then $m_k\le m$ and the same estimate gives the
stated $2^{\Omega(m/\log m)}$ lower bound.
No well-formedness test for the encoding is needed: correctness on
all finite $E$-structures already includes the encoded ones. Moreover,
if $C\equiv\widehat G_k$, then $C^\downarrow\equiv G_k$, since both
the decoder equivalence and the block simulation hold on arbitrary
nonempty structures.
\end{proof}

For any fixed vocabulary containing at least one binary symbol, the
same lower bound follows by fixing all other input relations to
$\zero$. Together with the upper bound, this completes the proof of
Theorem~\ref{thm:fixed-vocabulary:formal} and hence of
Theorem~\ref{thm:fixed-vocabulary}. It also determines both the
worst-case composition count and the worst-case total circuit size
at the scale $2^{\Theta(m/\log m)}$.
The composition lower bound also holds for terms, but does not
determine their optimal general upper bound.

\subsection{Model checking on directed graphs}\label{sec:fixed-model-checking}
Having proved Theorem~\ref{thm:fixed-vocabulary}, we turn to its
model-checking consequence. The constructive upper bound of
Theorem~\ref{thm:fixed-vocabulary:formal} allows an \FOthree\ sentence
to be compiled into a \CoR\ circuit, whose composition gates can be
evaluated by Boolean matrix multiplication.

\begin{corollary}\label{cor:fixed-model-checking}
Fix a finite binary vocabulary $\sigma$.
A tree-represented $\mathrm{FO}^3[\sigma]$ sentence of size $m$ can be
evaluated on an $n$-element structure in
\[
 2^{O_\sigma(m/\log m)}\,n^\omega
\]
time, including compilation. In particular, this holds for directed
graphs with one input edge relation.
\end{corollary}
\begin{proof}
Construct the circuit of Theorem~\ref{thm:fixed-vocabulary:formal} and evaluate
each gate once in topological order.
A composition is a Boolean matrix product, computable in
$O(n^{\omega+\eta})$ time for any fixed $\eta>0$; all other gates take
$O(n^2)$ time. By the running-time convention of
\S\ref{sec:aggregation}, this gives the stated $n^\omega$ factor.
By the assignment semantics of \S\ref{sec:prelim-logic}, a
sentence has value $\one$ or $\zero$ on the given structure, so any
output entry gives its truth value.
The circuit construction time and input-reading time are absorbed
in the displayed bound.
\end{proof}

\paragraph{Conditional tightness.}
For fixed vocabularies, this improves the $2^{O(m)}$ formula-size
factor in~\citep[Theorem~10]{nakamura-uezato2027}.
The 3-OV problem asks whether three $N$-element sets of $d$-bit vectors
contain $a,b,c$ such that $a_jb_jc_j=0$ for every coordinate $j$.
The \emph{low-dimensional 3-OV hypothesis} states that, for every
$\delta>0$, there is a constant $c_0>0$ such that no
$O(N^{3-\delta})$-time algorithm solves this problem when
$d\le c_0\log N$~\citep[Appendices to \S7.5]{nakamura-uezato2027}.
Under this hypothesis, their Theorem~72
(Appendix~T) rules out $2^{o(m/\log m)}n^{3-\delta}$ time
for every fixed $\delta>0$, even for tree-represented
existential-positive sentences over one binary relation.
Their bound counts edges $M$ rather than vertices.
Retaining at most three isolated vertices preserves \FOthree\
sentences and leaves at most $2M+3$ vertices
\citep[Proposition~6]{nakamura-uezato2027}, giving the stated vertex
bound. Building an adjacency matrix takes $O((M+1)^2)$ time,
absorbed after reducing $\delta$ to lie in $(0,1)$.
Thus $m/\log m$ is the conditionally optimal exponent scale for
truly subcubic evaluation, without a claim of optimality for $\omega$.

\section{Two-variable logic and tree size}\label{sec:FOtwo}
The formula $G_k$ is \FOtwo-definable, but its displayed linear-size
form uses three variable names. This section quantifies the difference,
proves the term-size bound of Theorem~\ref{thm:tree}, and separates
\CoR\ from \FOtwo\ (\S\ref{sec:reverse}).

\subsection{Quantifiers and fold sensitivity}\label{sec:quantifier-cost}
We first count the quantifiers needed to express $G_k$ in \FOtwo.
The key is that each quantifier translates into a one-sided
relational composition, which notices at most one fold.

\sentencebounds*

Recall the linear translation $\tau$ from \FOtwo\ formulas into \CoR\
given in \S\ref{sec:FOtwo-translation}. On the reference structure,
the off-diagonal support of $R\cap\id$ is empty. Therefore neither
of the compositions used to represent a reflexive atom,
$(R\cap\id);\one$ or $\one;(R\cap\id)$, notices a fold
(see Lemma~\ref{lem:notice}). Every quantifier, by contrast, contributes
one \emph{one-sided} composition $X;\one$ or $\one;X$; by the same
lemma, it can notice at most one fold. Boolean operations notice none.
This gives a bound in terms of quantifiers, rather than all
composition occurrences in $\tau(\varphi)$.
\begin{lemma}[Quantifier cost]\label{lem:FO-cost}
On the reference structure, every \FOtwo\ formula $\varphi$ satisfies
\[
 \tcost(\tau(\varphi))\le \#_{\exists\forall}(\varphi).
\]
\end{lemma}
The bound holds without restricting the Boolean connectives or the
quantifier alternation of the target formula.

\begin{proof}[Proof of Theorem~\ref{cor:sentences}]
Combining Lemma~\ref{lem:FO-cost} with Theorem~\ref{thm:test-cost} and
the tests of Lemma~\ref{lem:G-tests} gives
\[
 \varphi\equiv G_k\Longrightarrow \#_{\exists\forall}(\varphi)\ge2^k.
\]
In fact, agreement with $G_k$ on all finite structures suffices for
this bound. The \FOtwo\ expansion displayed in \S\ref{sec:additional-results}
uses exactly $2^k$ quantifiers, so the optimum quantifier count for
$G_k$ is exactly $2^k$. No restriction to existential target formulas is
made.
\end{proof}

The query $G_k$ belongs to the one-sided single-quantifier class
of Appendix~\ref{app:one-sided-fo2}.
Proposition~\ref{prop:one-sided-fo2} gives a matching upper bound:
a source query of size $s$ whose body is in negation normal form
with $d$ disjunctions translates into \FOtwo\ with at most $2^d$
quantifiers, quantifier depth at most one, and size $O(s2^d)$.
Since $G_k$ has $k$ disjunctions,
its exact quantifier lower bound attains this class-wide upper bound.

\subsection{Counting input occurrences in trees}\label{sec:trees}
The test--cost theorem also yields a stronger bound for syntax trees.
A \defn{singleton selector} for a color $\lambda$ is a
Boolean formula in the tables of the input relations, with constants
allowed, whose support on the entire palette is exactly $\{\lambda\}$.
\begin{restatable}[Disjoint input-occurrence charging]{lemma}{treecharging}\label{lem:tree-charge}
Fix the reference structure and its folds. Suppose every singleton
selector requires at least $d>0$ input-relation occurrences.
Then every \CoR\ term $T$, whether correct or not, satisfies
\[
 \tcost(T)\le \relatoms{T}/d.
\]
\end{restatable}
\paragraph{Proof idea.}
Call a composition \emph{active} when both operand supports are
nonempty. Its off-diagonal table is the constant true by
\eqref{eq:table-off}; an inactive composition notices no fold and its
off-diagonal table is zero or one operand table. Pruning active
compositions to true and inactive ones to the relevant table therefore
produces a Boolean formula without duplicating input occurrences.

For each active node $g$, prune its two operands and let $\mathcal L_g$
contain the input occurrences that survive. These are not all leaves
below $g$. The sets $\mathcal L_g$ are pairwise disjoint: an active
descendant is replaced by a constant when an ancestor's operands are
pruned, and incomparable nodes have disjoint subtrees. By
Lemma~\ref{lem:local}, each noticed color is selected by one pruned
operand table or by their conjunction. If two colors are noticed,
the two operands are disjoint singleton selectors. The selector
assumption thus gives $d|D_g|\le|\mathcal L_g|$, and
\[
 d\,\tcost(T)=d\sum_{g\text{ active}}|D_g|
 \le\sum_{g\text{ active}}|\mathcal L_g|
 \le\relatoms{T}.
\]
Appendix~\ref{app:trees} gives the full charging argument,
with explicit formulas for the diagonal-bit cases at inactive nodes. The disjointness is
where the tree structure is essential.

For the input relations of \S\ref{sec:G}, a Boolean formula in their
tables that contains neither $P_i$ nor $Q_i$ takes the same value on a
color and on its $i$-th-bit flip. Hence a singleton selector contains
$P_i$ or $Q_i$ for every $i$, so $d=k$, and this is exact, since
$H_\lambda$ selects $\lambda$. Theorem~\ref{thm:test-cost} and
Lemma~\ref{lem:tree-charge} give $\relatoms{T}\ge k2^k$ for every \CoR\
term $T$ agreeing with $G_k$ on all finite structures. By the
tree-size inequality from \S\ref{sec:prelim-resources},
\[
 |T|\ge2k2^k-1\qquad(T\equiv G_k),
\]
which is the lower bound of the following theorem.
\begin{theorem}[Term size]\label{thm:tree}
For every $k\ge1$, every \CoR\ term $T\equiv G_k$ has at least $2k2^k-1$
nodes. The lower bound holds even if $T$ and $G_k$ agree only on
finite structures, and $G_k$ is equivalent on all nonempty structures
to a \CoR\ term with $(2k+6)2^k-5$ nodes.
\end{theorem}
\begin{proof}
The lower bound is the inequality above. For the upper bound, factor
the conditions on the queried pair $(x,y)$ into a complete binary tree, with the
two branches at level $i$ selecting $P_i$ and $Q_i$. At a leaf $A$,
where the $P_i$ branch was taken exactly for $i\in A$, place
$H_{A^c};\one$. A branch with child terms $U_1,U_0$ has the form
\[
 (P_i\cap U_1)\cup(Q_i\cap U_0).
\]
Distributing the branches recovers the DNF for $G_k$ on arbitrary
inputs; no assumption that exactly one of $P_i,Q_i$ holds at each edge
is used. Each of the $2^k$ leaves costs $2k+1$ nodes, and each of the
$2^k-1$ internal branches adds five. Thus the total is
\[
 2^k(2k+1)+5(2^k-1)=(2k+6)2^k-5.
\]
The number of relation-input occurrences is $k2^k+2(2^k-1)$.
\end{proof}

\subsection{Sentences and the variable hierarchy}\label{sec:sentences}
The preceding subsections establish the binary-query bounds of
Theorem~\ref{cor:sentences}. We now prove its sentence extension for
$G_k^\forall=\forall x\,\forall y\,(x=y\lor G_k(x,y))$.

\paragraph{Quantitative and sentence extensions.}
We also prove two refinements. Every \FOtwo\ formula defining $G_k$
has at least $k2^k$ relation-atom occurrences, and the minimum
syntax-tree size of such a formula is $\Theta(k2^k)$. Moreover, $G_k^\forall$ is an
\FOtwo-definable \FOthree\ sentence of size $O(k)$ and quantifier
depth three, but every equivalent \FOtwo\ sentence over finite
structures has at least $2^k+1$ quantifiers and $k2^k$ relation-atom
occurrences; the minimum tree size of such a sentence is $\Theta(k2^k)$.
\begin{proof}[Proof of the quantitative and sentence extensions]
The exact $2^k$-quantifier lower bound for the binary query $G_k$
was established in \S\ref{sec:quantifier-cost}. For any equivalent \FOtwo\
formula $\varphi$, apply Theorem~\ref{thm:test-cost} and
Lemma~\ref{lem:tree-charge} to $\tau(\varphi)$. By the
relation-atom preservation established in \S\ref{sec:FOtwo-translation},
this gives
$\relatoms{\varphi}\ge k2^k$. The \FOtwo\ expansion displayed in
\S\ref{sec:additional-results} has $O(k2^k)$ nodes and exactly $2^k$
quantifiers, proving the tree-size claim for the binary query.

For the sentence $G_k^\forall$, the source has size $8k+4$ and
quantifier depth three. By Lemma~\ref{lem:G-tests}, it is true in
$\mathcal M_0$ and false in every $\mathcal M_\lambda$. On each of
these structures its assignment relation is $\zero$ or $\one$, and
both relations are fixed by prediction. Thus these structures give $2^k$ tests for the
sentence itself. For any equivalent \FOtwo\ sentence $\varphi$,
Theorem~\ref{thm:test-cost} gives $\tcost(\tau(\varphi))\ge2^k$.
Every \FOtwo\ formula with at most one free variable has a constant
assignment relation on $\mathcal M_0$: by Lemma~\ref{lem:tables} and the
translation $\tau$, it is color-determined and independent of one
coordinate, so its value everywhere equals its common diagonal value.
Choose an outermost quantifier occurrence $Qv\,\theta$ in $\varphi$,
that is, one outside the scope of all other quantifiers.
Since $\varphi$ is closed, $\operatorname{FV}(\theta)\subseteq\{v\}$,
so the reference value of its operand is $\zero$ or $\one$.
The corresponding projection notices no fold, also when $Q=\forall$
is handled by complement. All other quantifiers notice at most one
fold, and the compositions introduced for reflexive atoms notice none,
as in Lemma~\ref{lem:FO-cost}. Thus
\[
 2^k\le\tcost(\tau(\varphi))\le\#_{\exists\forall}(\varphi)-1,
\]
proving the $2^k+1$ quantifier lower bound.
Lemma~\ref{lem:tree-charge}, together with exact preservation of
$\#_{\mathrm{rel}}$, also gives $\relatoms{\varphi}\ge k2^k$.
Conversely, substitute the \FOtwo\ expansion of $G_k$ into the
sentence. It has $2^k+2$ quantifiers and $O(k2^k)$ nodes, proving the
sentence upper bounds and the matching order of tree size.
\end{proof}
The quantifier optimum for this sentence is therefore in
$[2^k+1,2^k+2]$; the exact $2^k$ optimum proved earlier concerns the binary
query.

The separation is over growing binary signatures, not a lower bound
for separating a prescribed pair of structures as a function of their
size. It complements the higher-target-variable separations
of~\citep[Theorem~2]{vinall-smeeth2024}; the contrasting polynomial
translation of~\citep[Theorem~4.3]{grohe-schweikardt2005} concerns sentences
on pure finite linear orders. See \S\ref{sec:relwork} for the full comparison.

\subsection{Separating \texorpdfstring{{\boldmath\CoR}}{CoR} from
 \texorpdfstring{{\boldmath\FOtwo}}{FO2}}\label{sec:reverse}
The linear translation of \S\ref{sec:FOtwo-translation} goes from \FOtwo\ to \CoR. On the reference
structure, a composition notices at most two folds and a quantifier at
most one, so tests on the reference structure bound \CoR\ circuits and
\FOtwo\ formulas alike, up to a factor of two. To separate the two
formalisms, we use a structure on which a quantifier still notices at
most one fold, but a single composition notices all of them.

Use input symbols $R,P_1,\ldots,P_k$, and put
\[
 \mathscr P_i^1\coloneqq\id\cap P_i,\qquad \mathscr P_i^0\coloneqq\id\cap\overline{P_i},\qquad
 U_k\coloneqq R;\bigcap_{i=1}^k
 \left[(\mathscr P_i^0;\one;\mathscr P_i^0)\cup(\mathscr P_i^1;\one;\mathscr P_i^1)\right].
\]
The complement in $\mathscr P_i^0$ is taken before intersection with the identity;
it is not the complement of $\mathscr P_i^1$. Expanding these abbreviations gives
an $O(k)$-size tree. For a vertex $v$, let
$\pi(v)\coloneqq\{i\in[k]:P_i(v,v)\}$ be its \defn{loop pattern}.
Here $\pi(v)$ is a local one-vertex notation, distinct from the pair
profile $\pi(u,v)$ of \S\ref{sec:domain-size}.
Then $U_k(x,y)$ holds exactly when some $R$-successor $z$ of $x$
satisfies $\pi(z)=\pi(y)$; equivalently, $U_k$ is defined by the
\FOthree\ formula
\[
 \exists z\left(R(x,z)\land
     \bigwedge_{i=1}^k(P_i(z,z)\leftrightarrow P_i(y,y))\right).
\]

\paragraph{The structures.}
Let $\mathbb Z_{11}$ denote the integers modulo $11$, and put
$V\coloneqq2^{[k]}\times\mathbb Z_{11}$ and
$\Gamma\coloneqq\{1,2,3,4,7,8\}\subseteq\mathbb Z_{11}$. The structure
$\mathcal W_0$ on $V$ has the loops $P_i((A,s),(A,s))$ for $i\in A$, no
other $P_i$-entries, and $R((A,s),(B,s'))$ exactly when $s'-s\in\Gamma$;
thus $\pi((A,s))=A$, and $R$ has no loops. The \defn{atomic type} of a pair
$(u,v)$ of distinct vertices is $(R(u,v),R(v,u),\pi(u),\pi(v))$, and that
of a loop $(u,u)$ is $\pi(u)$; it determines the values of all atomic
formulas at the pair. As $s'-s$ ranges over
$\mathbb Z_{11}\setminus\{0\}$, the truth values of $s'-s\in\Gamma$ and
$s-s'\in\Gamma$ take each of the four combinations at least twice:
$(1,0)$ at $s'-s\in\{1,2\}$, $(0,1)$ at $s'-s\in\{9,10\}$, $(1,1)$ at
$s'-s\in\{3,4,7,8\}$, and $(0,0)$ at $s'-s\in\{5,6\}$. Hence every
vertex $u$ is the first entry of at least two off-diagonal pairs of each
atomic type $(r,r',\pi(u),B)$, and the second entry of at least two
of each type $(r,r',B,\pi(u))$, for $r,r'\in\{0,1\}$ and $B\subseteq[k]$.

For this construction, fix the hub $\hub\coloneqq(\varnothing,0)$. For $A\subseteq[k]$, the fold
of $A$ yields the structure $\mathcal W_A$ obtained from $\mathcal W_0$ by
removing from $R$ the edges $(\hub,v)$ with $v\ne \hub$ and $\pi(v)=A$. It
changes the atomic types of these pairs and of their converses, and of no
other pair. Afterwards, no pair with first entry $\hub$ has a type
$(1,r',\varnothing,A)$, and no pair with second entry $\hub$ has a type
$(r,1,A,\varnothing)$; all other types still occur at $\hub$.

\paragraph{Type tables and predictions.}
As in \S\ref{sec:types}, we use relations determined by a finite
set of pair types, but here the types are atomic pair types rather
than edge colors, and the folds delete edges rather than recolor them.
A relation $X\subseteq V^2$ has a \defn{type table} if $X(u,v)$ depends
only on the atomic type of $(u,v)$ in $\mathcal W_0$; let $S_X$ be the
set of atomic types at which $X$ holds, rather than a set of colors.
Every atomic type of $\mathcal W_A$ occurs in $\mathcal W_0$, so the prediction $\pred XA$ can
read the same table at the atomic types of $\mathcal W_A$. As in
\S\ref{sec:unnoticed}, for operand values $\mathbf U$ having type tables,
an operation $\operatorname{op}$ whose result $\operatorname{op}(\mathbf U)$
also has a type table notices the fold of $A$ if
$\operatorname{op}(\pred{\mathbf U}A)\ne\pred{\operatorname{op}(\mathbf U)}A$,
and Boolean operations never notice a fold.

For each \FOtwo\ formula $\psi$, write $X_\psi$ for its binary
assignment relation in the structure under consideration, as in
\S\ref{sec:prelim-logic}. Existential quantifiers act as the projections
\[
 X_{\exists y\,\psi}=X_\psi;\one,\qquad
 X_{\exists x\,\psi}=\one;X_\psi.
\]
Universal quantifiers are handled by complement.
By induction on the formula, every relation that an \FOtwo\ formula
defines on $\mathcal W_0$ has a type table: if $X$ has one, then
$(X;\one)(u,v)$ holds exactly when $S_X$ contains the atomic type of some
pair with first entry $u$. The set of atomic types realized by pairs
with first entry $u$ depends only on $\pi(u)$; the argument for
$\one;X$ is symmetric.

\begin{lemma}[Projections notice at most one fold]\label{lem:W-projection}
If $X$ has a type table, then $X;\one$ and $\one;X$ each notice the fold
of at most one $A\subseteq[k]$.
\end{lemma}
\begin{proof}
At a pair $(u,v)$, the prediction of $X;\one$ holds exactly when $S_X$
contains the atomic type in $\mathcal W_0$ of some pair with first entry
$u$, and $\pred XA;\one$ holds exactly when $S_X$ contains the atomic type
in $\mathcal W_A$ of some such pair. If $u\ne \hub$, only the
off-diagonal pair $(u,\hub)$ can change. Every off-diagonal atomic type
realized among pairs with first entry $u$ is realized at least twice,
so its old type is still realized after the fold. Its new type was
already realized in $\mathcal W_0$, while the diagonal pair $(u,u)$
is unchanged. Hence the sets of realized types coincide. If $u=\hub$, the fold removes exactly the types
$(1,r',\varnothing,A)$ and adds none. Hence $X;\one$ notices the fold of
$A$ only if the types in $S_X$ of pairs with first entry $\hub$ are all of
the form $(1,r',\varnothing,A)$, and there is at least one; this
determines $A$. The argument for $\one;X$ uses second entries.
\end{proof}

\reverse*
\paragraph{Equational sentence extension.}
The \CoR\ equation $\id\cup U_k=\one$ is \FOtwo-definable as a
sentence, but every equivalent \FOtwo\ sentence over finite
structures has at least $2^k$ quantifiers.
\begin{proof}[Proof of the query and sentence claims]
For the upper bounds, let $\vartheta_A(u)$ be the quantifier-free
formula asserting that $\pi(u)=A$. The \FOtwo\ formula
\[
 \bigvee_{A\subseteq[k]}\left[\vartheta_A(y)\land
       \exists y\,(R(x,y)\land\vartheta_A(y))\right]
\]
defines $U_k$ with exactly $2^k$ quantifiers; the inner $y$ is rebound
only in the second conjunct. The equation $\id\cup U_k=\one$ says that
for all $x$ and all $y\ne x$, some $R$-successor of $x$ has the loop
pattern of $y$, and so does the \FOtwo\ sentence
\[
 \forall x\bigwedge_{A\subseteq[k]}\left[
 \exists y\,(x \neq y\land\vartheta_A(y))\rightarrow
 \exists y\,(R(x,y)\land\vartheta_A(y))\right].
\]

For the lower bounds, every vertex of $\mathcal W_0$ has an $R$-successor
with each loop pattern, so $U_k$ is the full relation on $\mathcal W_0$
and the equation holds there. On $\mathcal W_A$, the hub has no
$R$-successor with loop pattern $A$, so $U_k(\hub,v)$ fails for every
$v\ne \hub$ with $\pi(v)=A$, and the equation fails. Let $\psi$ be an
\FOtwo\ formula agreeing with $U_k$ on all finite structures, or an
\FOtwo\ sentence holding on exactly the finite structures satisfying
the equation. Its value on $\mathcal W_0$ is the full relation, whose
prediction on $\mathcal W_A$ is again full, but its value on
$\mathcal W_A$ is not. If no quantifier of $\psi$ noticed the fold of
$A$, then, as in the proof of Lemma~\ref{lem:unnoticed}, every subformula
would take its predicted value on $\mathcal W_A$, by induction on the
formula: atomic types determine atomic formulas, and Boolean operations
notice nothing. Hence some quantifier notices the fold of each $A$, and
by Lemma~\ref{lem:W-projection} each quantifier notices at most one, so
$\psi$ has at least $2^k$ quantifiers.
\end{proof}

\section{Scalar aggregation and matrix query languages}\label{sec:matrix}
We extend the fold argument from binary relations to semiring-valued
matrix computations. We first describe the short three-attribute
expression, then the matrix circuit model, and finally the fold
sensitivity bounds that force exponential arity-elimination overhead.

\subsection{The three-attribute query and its matrix interpretation}\label{sec:weightedG}
Let $\mathfrak S\coloneqq(S,\oplus,\otimes,0_S,1_S)$ be a commutative
semiring and $V$ a finite nonempty common index set, with inputs
$P_i,Q_i:V^2\to S$. Thus $(S,\oplus,0_S)$ and $(S,\otimes,1_S)$
are commutative monoids, $\otimes$ distributes over $\oplus$, and
$0_S$ is absorbing for $\otimes$.
An \emph{annotated relation} assigns an $S$-value to each tuple.
In positive annotated relational algebra, union $\cup$ adds annotations,
natural join $\bowtie$ multiplies them, and projection $\pi_{xy}$ sums
over eliminated attributes while retaining $x,y$; renaming changes
attribute names.
ARA$(3)$ restricts intermediate relations to at most three attributes,
analogous to the ternary assignment relations of \FOthree\
(\S\ref{sec:prelim-logic}). Over Boolean annotations, these operations
are disjunction, conjunction, and existential quantification.

For an explicit ARA$(3)$ expression, let $P_i^{xy}$ denote $P_i$
with attributes $(x,y)$, and similarly for $P_i^{xz}$, $Q_i^{xy}$,
and $Q_i^{xz}$. Our semiring-valued query is
\[
 \mathcal G_k\coloneqq\pi_{xy}\!\left(
 \mathop{\bowtie}_{i=1}^{k}
 \bigl[(P_i^{xy}\bowtie Q_i^{xz})
       \cup(Q_i^{xy}\bowtie P_i^{xz})\bigr]\right).
\]
Every intermediate relation has at most the three attributes
$(x,y,z)$: the $k$ conditions are combined with the \emph{same} $z$, which is
eliminated only at the final projection. Thus the ARA$(3)$ expression
has size $O(k)$. Equivalently, its scalar value is
\[
 \mathcal G_k(x,y)=\bigoplus_{z\in V}\ \bigotimes_{i=1}^{k}
 \bigl[(P_i(x,y)\otimes Q_i(x,z))
        \oplus(Q_i(x,y)\otimes P_i(x,z))\bigr].
\]
For $S=\mathbb B$ with $\oplus=\lor$ and $\otimes=\land$, this
is precisely $G_k(x,y)$, without negation or universal quantification.
For integer min--plus, specializing $\mathfrak S$ to
$(\mathbb Z\cup\{+\infty\},\min,+,+\infty,0)$ gives
\[
 \mathcal G_k(x,y)=\min_{z\in V}\sum_{i=1}^{k}
 \min\{P_i(x,y)+Q_i(x,z),\,Q_i(x,y)+P_i(x,z)\}.
\]

\paragraph{An explicit MATLANG expression.}
Although MATLANG retains only matrices, vectors, and scalars,
the DNF of \S\ref{sec:main-results} gives an expression for
$\mathcal G_k$ using matrix--vector products. For $A\subseteq[k]$,
let $H_A$ be the pointwise semiring analogue of the relation in
\S\ref{sec:main-results}:
\[
 H_A(x,y)\coloneqq
 \bigotimes_{i\in A}P_i(x,y)\ \otimes
 \bigotimes_{i\in A^c}Q_i(x,y),
\]
where an empty product is $1_S$.
Let $\mathbf 1_S$ be the constant column vector with every entry
$1_S$. Computing the column vector $r_A$ and then broadcasting it
along the rows gives
\begin{equation}\label{eq:matlang-dnf}
 r_A\coloneqq H_{A^c}\mathbf 1_S,\qquad
 \mathcal G_k(x,y)=\bigoplus_{A\subseteq[k]}
       H_A(x,y)\otimes r_A(x).
\end{equation}
Distributivity expands the $k$ clauses and interchanges the finite
aggregations over $A$ and $z$, proving the identity for arbitrary
$S$-valued inputs. Each $r_A$ uses one matrix--vector product;
all other steps are entrywise operations or row broadcasting.
Thus $2^k$ such products suffice, even without matrix--matrix
products. For Boolean and tropical semirings, our lower bounds show
that exponentially many data-index aggregations remain necessary, even
with matrix--matrix products, the additional operations of
\S\ref{sec:typed}, and sharing.

\subsection{Typed matrix circuits}\label{sec:typed}
A semiring is \emph{nontrivial} if $0_S\ne1_S$ and
\emph{idempotent} if $a\oplus a=a$ for every $a\in S$.
In \S\S\ref{sec:typed}--\ref{sec:arity-bounds}, we assume idempotent
addition and write it as $\vee$; \S\ref{sec:positive} treats positive
computations without this assumption. The circuit model below permits
arbitrary binary kernels, rather than only
semiring multiplication; it therefore includes the corresponding
MATLANG operations. This strengthening is useful because the lower
bound charges only data-index aggregations.

We use the dimension symbols $n$ and $1$. Intermediate arrays have
shapes $n\times n$, $n\times1$, $1\times n$, or $1\times1$---matrices,
column vectors, row vectors, and scalars. A data-index contraction
of an $a\times n$ and an $n\times c$ array is
\[
 Z(i,j)\coloneqq\bigvee_{z=1}^n\kappa(X(i,z),Y(z,j)),
 \qquad a,c\in\{1,n\},
\]
where the kernel $\kappa:S\times S\to S$ may vary by gate and
need not be associative, monotone, or zero-preserving.
Contractions over an index of size one apply the kernel without
aggregating over data, and are not charged.
We also allow arbitrary uniform entrywise functions: for arrays
$X_1,\ldots,X_r$ of the same shape and a fixed map $f:S^r\to S$,
\[
 F(X_1,\ldots,X_r)(i,j)\coloneqq
 f(X_1(i,j),\ldots,X_r(i,j)).
\]
Transpose exchanges the two indices. Broadcasting a column $\mathbf c$,
a row $\mathbf r$, or a scalar $s$ gives entries $\mathbf c(u)$,
$\mathbf r(v)$, or $s$, respectively, at position $(u,v)$; a scalar
may also be broadcast to a vector. Diagonal extraction returns the
vector with entry $X(u,u)$, and diagonal formation places the vector
entries on the diagonal and $0_S$ off it. We permit uniform constant
vectors and scalars, and a diagonal mask with any two fixed values,
one on and one off the diagonal. These operations are uncharged.
Standard MATLANG operations are described in
\citep{brijder2019tods,brijder2022}.

All circuits have fixed directed acyclic wiring. Functions, kernels,
and uniform scalar constants are fixed before the input is seen,
though they may depend on $n$ and $k$ in a size-specific circuit.
``Uniform'' means that the same entrywise function is applied at
every position of its type and cannot access vertex names.
A kernel receives only its two scalar arguments, not the output
position $(i,j)$, the contracted index $z$, a third matrix entry,
or an encoding of the entire input unless computed by permitted
gates. Data-dependent values and tests are allowed, but there is
no primitive for data-dependent rewiring or arbitrary indexed
access. The input arrays are square $n\times n$ matrices; no order,
vertex identifiers, arbitrary input vectors, or vertex-dependent
constant arrays are supplied for free.

Write $\mvcount C$ for the number of matrix--vector and
vector--matrix data-index contractions and $\mmcount C$ for the
number of matrix--matrix data-index contractions, counting each
distinct gate once. We call these charged gates \emph{aggregation gates}.
All other operations, including vector inner products, are uncharged
in this resource measure. We next show that these
uncharged operations notice no folds, even though some of them may
require compositions when expressed in CoR.

\subsection{Fold sensitivity and finite join breadth}\label{sec:breadth}
In an idempotent semiring, an aggregation may depend on several
indispensable summands. We use the number of summands needed to
determine a finite sum to bound how many folds one gate can notice.

Let $(S,\vee,0_S)$ be a join-semilattice with least element $0_S$.
For scalar values $a,b\in S$, use the semilattice order
$a\le b$ if and only if $a\vee b=b$. For min--plus this is the reverse
of the usual numerical order. The semilattice's \emph{join breadth} $\beta$ is the least nonnegative integer such
that every finite join equals the join of at most $\beta$ of its
summands (the empty join is $0_S$). We assume $\beta<\infty$.
For a nontrivial totally ordered semiring with min or max as addition,
$\beta=1$: one summand attains the aggregate. In a product with
$b$ independent coordinates, different summands can determine
the result in different coordinates, and breadth can be $b$.

A fold removes one color from the hub. If a join is already determined
by summands of other colors, the fold cannot change its value.
We apply this observation to the hub row, column, and diagonal.

Use the reference coloring and folds of Section~\ref{sec:preservation},
with $\Lambda=2^{[k]}$, on a finite index set $V$ of size
$n\coloneqq|V|\ge n_0(2^k)$. For color-determined square arrays $X,Y$
and any kernel $\kappa:S\times S\to S$, put
\[
 (X\agg\kappa Y)(u,v)\coloneqq
 \bigvee_{z\in V}\kappa(X(u,z),Y(z,v)).
\]
Kernels may differ between gates; no order or monotonicity assumption
is made on $\kappa$. Reference arrays are color-determined as in
\S\ref{sec:types}, now with entries $X[\diagc],X[\mu]\in S$.
Their predictions read the same table at the folded coloring
$\chi_\lambda$. Entrywise functions and transpose notice no fold.
For $J\subseteq\Lambda$, define
\begin{align*}
 \kappa_J^{\mathrm L}&\coloneqq\bigvee_{\mu\in J,\,\nu\in\Lambda}
                          \kappa(X[\mu],Y[\nu]),&
 \kappa_J^{\mathrm R}&\coloneqq\bigvee_{\mu\in J,\,\nu\in\Lambda}
                          \kappa(X[\nu],Y[\mu]),\\
 \kappa_J^{\mathrm D}&\coloneqq\bigvee_{\mu\in J}\kappa(X[\mu],Y[\mu]),&
 \kappa_*&\coloneqq\kappa_\Lambda^{\mathrm L}
                     =\kappa_\Lambda^{\mathrm R}.
\end{align*}
The reference product $Z\coloneqq X\agg\kappa Y$ has the following
\emph{scalar color table} (SCT):
\begin{equation*}\tag{SCT}\label{eq:scalar-color-table}
\begin{aligned}
 Z[\diagc]&=\kappa(X[\diagc],Y[\diagc])
            \vee\kappa_\Lambda^{\mathrm D},\\
 Z[\mu]&=\kappa(X[\diagc],Y[\mu])
          \vee\kappa(X[\mu],Y[\diagc])\vee\kappa_*.
\end{aligned}
\end{equation*}
The calculation of \S\ref{sec:responses} remains valid, with
endpoint contributions formed by $\kappa$ and, away from the hub,
outside-witness contributions $\kappa_*$ off the diagonal and
$\kappa_\Lambda^{\mathrm D}$ on it. For the fold of $\lambda$,
put $J\coloneqq\Lambda\setminus\{\lambda\}$.
For $\pred X\lambda\agg\kappa\pred Y\lambda$, the outside-witness
contributions in the hub row, hub column, and hub diagonal are
$\kappa_J^{\mathrm L}$, $\kappa_J^{\mathrm R}$, and
$\kappa_J^{\mathrm D}$, respectively.
For $\pred{(X\agg\kappa Y)}\lambda$, the corresponding contributions
use $J=\Lambda$. This follows from the same pairs of witness colors,
while idempotence makes their multiplicities irrelevant.

\begin{lemma}[Local bounds for contractions]\label{lem:scalar}
For color-determined square arrays $X,Y$ and an arbitrary kernel
$\kappa:S\times S\to S$, the contraction $X\agg\kappa Y$ notices
at most $3\beta$ folds. If one operand is constant across off-diagonal
colors, it notices at most $\beta$; if both are, it notices none.
Over the Boolean semiring $\bool$, with $\kappa:\bool^2\to\bool$,
the general bound improves to two.
\end{lemma}
\begin{proof}
By the calculation above, the gate notices the fold of $\lambda$
only if one of the joins $\kappa_J^{\mathrm L}$,
$\kappa_J^{\mathrm R}$, and $\kappa_J^{\mathrm D}$ with
$J=\Lambda\setminus\{\lambda\}$ differs from its value at
$J=\Lambda$. Choose at most $\beta$ summands of
$\kappa_\Lambda^{\mathrm L}$ whose join is
$\kappa_\Lambda^{\mathrm L}$, and let $K^{\mathrm L}$ be the
set of their row colors, the row color of
$\kappa(X[\mu],Y[\nu])$ being $\mu$. If
$\lambda\notin K^{\mathrm L}$, all chosen summands also occur in
$\kappa_J^{\mathrm L}$, so
$\kappa_\Lambda^{\mathrm L}\le
\kappa_J^{\mathrm L}\le\kappa_\Lambda^{\mathrm L}$.
Define $K^{\mathrm R}$ in the same way from the column colors of
summands of $\kappa_\Lambda^{\mathrm R}$, and $K^{\mathrm D}$
from the colors of summands of $\kappa_\Lambda^{\mathrm D}$.
Hence the gate notices the fold of $\lambda$ only if
$\lambda\in K^{\mathrm L}\cup K^{\mathrm R}\cup K^{\mathrm D}$,
a set of at most $3\beta$ colors.

If $Y[\mu]$ is independent of $\mu$, then
$\kappa_J^{\mathrm L}=\kappa_J^{\mathrm D}$ and
$\kappa_J^{\mathrm R}=\kappa_*$ for every nonempty $J$;
only $K^{\mathrm L}$ is needed. The case of $X$ is symmetric.
If both are constant, all three joins are unchanged.
Over the Boolean semiring, if $\kappa_\Lambda^{\mathrm D}=1$, a single
color $\mu$ with $\kappa(X[\mu],Y[\mu])=1$ supplies a summand
of value one to all three joins whenever $\mu\in J$, so only the
fold of $\mu$ can be noticed. Otherwise the diagonal join is
always zero, and one color for each of the other two joins suffices.
\end{proof}
The formulas~\eqref{eq:scalar-color-table} for $Z[\diagc]$ and $Z[\mu]$
show that $Z=X\agg\kappa Y$ is color-determined.
More generally, a simultaneous induction
over the typed circuit shows that, on the reference structure, every
matrix is color-determined and every vector is constant over its data
index. For the matrix--vector step, contracting a reference matrix $X$
with a constant column vector of value $c$ gives
\[
 \kappa(X[\diagc],c)\vee\bigvee_{\mu\in\Lambda}\kappa(X[\mu],c)
\]
at every coordinate: each row contains its diagonal entry and every
color, and idempotence removes multiplicities. The vector--matrix case
is symmetric. Uniform entrywise functions, transpose, broadcasting,
diagonal operations, products over an index of size one, and vector
inner products preserve these reference-value classes; the permitted
inputs and constants belong to them. This establishes the closure
required by test-system condition~(i) in \S\ref{sec:simultaneous}.

Encode a column vector
$\mathbf c$ as the square array $(u,v)\mapsto\mathbf c(u)$,
a row vector $\mathbf r$ as $(u,v)\mapsto\mathbf r(v)$, and a
scalar as a constant square array. Encoded reference vectors are
constant across all edge colors and on the diagonal and are
unchanged by prediction, as are scalars. These assertions concern
reference values and their predictions; vectors recomputed on folded
inputs need not be constant. Diagonal formation and extraction,
broadcasting, and permitted constants notice no folds.
Matrix--vector and vector--matrix contractions notice at most
$\beta$ folds by Lemma~\ref{lem:scalar}; a vector inner product
notices none, since both operands are constant.
Products over an index of size one also notice none.

\begin{lemma}[Cost of typed products]\label{lem:typed-cost}
Every typed matrix circuit $C$ in this model satisfies
\[
 \tcost(C)\le\beta(\mvcount C+3\mmcount C).
\]
Over the Boolean semiring, with arbitrary Boolean-valued kernels,
the bound improves to $\mvcount C+2\mmcount C$.
\end{lemma}
\begin{proof}
Sum Lemma~\ref{lem:scalar} over the distinct aggregation gates.
All other permitted gates notice none, as shown above.
\end{proof}
\begin{center}
\renewcommand{\arraystretch}{1.1}
\begin{tabular}{p{0.62\linewidth}l}
\toprule
Operation & Folds noticed, at most\\
\midrule
Uniform entrywise function, transpose, permitted constant & $0$\\
Diagonal formation or extraction, broadcasting & $0$\\
Product over an index of size one & $0$\\
Vector inner product & $0$\\
Matrix--vector or vector--matrix contraction & $\beta$\\
Square matrix--matrix contraction & $3\beta$\\
\bottomrule
\end{tabular}
\end{center}
The cost counts noticed folds, not the time needed to evaluate a gate.

\paragraph{Why join breadth matters.}
The factor $\beta$ in the local matrix--vector bound cannot in
general be replaced by a constant. For $b\ge1$, the finite-set
semiring $(2^{[b]},\cup,\cap)$ has join breadth exactly $b$:
every union is determined by at most $b$ sets, whereas the union
of the $b$ singletons requires all of them. The same is true of $\mathsf{Trop}_{\mathbb Z}^{b}$, the $b$-fold
product of the integer min--plus semiring with coordinatewise
operations: different coordinates can attain their minima at
different summands.
Moreover, when $2^{k-1}\ge b$, choose distinct colors
$\lambda_1,\ldots,\lambda_b$ with no two complementary, and put
$X[\lambda_j]=\{j\}$ and $X[\mu]=\varnothing$ for all other
colors and on the diagonal. Matrix--vector aggregation of $X$
against the constant column $[b]$, using union for aggregation
and intersection as kernel, yields $[b]$ at the hub on the
reference structure. Folding $\lambda_j$ changes this value to
$[b]\setminus\{j\}$. Thus one aggregation notices $b$ folds.
This is local tightness, not a claim that the overall circuit bound
is sharp in $\beta$.

\subsection{Exponential arity-elimination lower bounds}\label{sec:arity-bounds}
We now apply Theorem~\ref{thm:test-cost} to the $2^k$ original fold
tests. The inputs are encoded by the two distinct semiring values
$0_S,1_S$. On every test, idempotence makes $\mathcal G_k$
return $1_S$ or $0_S$ according to the Boolean truth of $G_k$,
so each fold still changes the answer at the unchanged pair $(\hub,v_{\lambda^c})$.
For natural and integer min--plus, $0_S=+\infty$ and $1_S=0$. Intermediate gates, however, may take
arbitrary values of $S$, which is why the scalar analysis of the
previous subsection is needed.

\begin{theorem}[Exponential arity-elimination overhead]\label{cor:matrix}
Let $S$ be a nontrivial commutative idempotent semiring of finite
join breadth $\beta$. For every $k\ge1$, the query $\mathcal G_k$
has a positive ARA$(3)$ expression of size $O(k)$, but every typed
matrix circuit $C$ computing it on all finite inputs satisfies
\[
 \beta(\mvcount C+3\mmcount C)\ge2^k.
\]
In particular, at least $2^k/(3\beta)$ aggregation gates are
necessary. Conversely, $2^k$ matrix--vector products suffice.
\end{theorem}
\begin{proof}
The $0_S,1_S$ encodings above supply $2^k$ fold tests. Theorem~\ref{thm:test-cost}
and Lemma~\ref{lem:typed-cost} give
\[
 2^k\le\tcost(C)\le\beta(\mvcount C+3\mmcount C).
\] Since
$\mvcount C+3\mmcount C\le3(\mvcount C+\mmcount C)$,
there must be at least $2^k/(3\beta)$ aggregation gates.

For the upper bound, the explicit MATLANG expression
\eqref{eq:matlang-dnf} uses exactly $2^k$ matrix--vector
aggregations; all remaining operations are entrywise or broadcasts.
Its correctness for arbitrary semiring-valued inputs follows from
distributivity, as verified in \S\ref{sec:weightedG}.
Reusing partial products in the $H_A$ as in
\S\ref{sec:main-results} gives $O(2^k)$ total nodes.
\end{proof}

\tropicalmatrixbounds*
\begin{proof}
Both tropical semirings are nontrivial, commutative, and idempotent,
and have join breadth $\beta=1$. The claim is therefore the
$\beta=1$ case of Theorem~\ref{cor:matrix}.
\end{proof}

\begin{theorem}[Boolean arity-elimination overhead]\label{thm:boolean-matrix}
Over the Boolean semiring $\mathbb B$, every typed matrix circuit
$C$ computing $G_k$ on all finite inputs satisfies
$\mvcount C+2\mmcount C\ge2^k$.
This holds with arbitrary uniform entrywise functions and
aggregation kernels.
\end{theorem}
\begin{proof}
Use the same tests and counting argument as in Theorem~\ref{cor:matrix},
replacing the general cost estimate by the Boolean refinement of
Lemma~\ref{lem:typed-cost}. This gives
$2^k\le\tcost(C)\le\mvcount C+2\mmcount C$.
\end{proof}

\paragraph{Why the Boolean case is not just a restatement of the CoR bound.}
Boolean matrix multiplication is relational composition, so total-size
lower bounds for ordinary Boolean matrix expressions are closely
related to Theorem~\ref{thm:main}. Here we count only aggregations:
all other permitted matrix operations are uncharged in this measure. For example,
diagonal extraction followed by broadcasting obtains the matrix
$(u,v)\mapsto R(u,u)$ without aggregation, whereas a CoR expression
for it is $(R\cap\id);\one$, which uses a composition.
Consequently, the aggregation count of a matrix circuit cannot be
identified with the composition count of its CoR translation.
We instead bound the fold sensitivity of the typed operations directly;
that analysis also works over non-Boolean semirings.

Examples of breadth-one semirings include
$\mathsf{Trop}_{\mathbb N}$, $\mathsf{Trop}_{\mathbb Z}$,
the integer max--plus semiring
$(\mathbb Z\cup\{-\infty\},\max,+)$, and max--min on $[0,1]$.
All require at least $2^k/3$ aggregations, even in our enlarged
kernel model.

\subsection{Positive computations without idempotence}\label{sec:positive}
A separate transfer handles positive arithmetic, for which join breadth
is not available. Let $(S,+,\otimes,0_S,1_S)$ be a nontrivial commutative
semiring satisfying the \emph{nonvanishing} (NV) condition:
\begin{equation*}\tag{NV}\label{eq:positive-characteristic}
 r\,1_S\ne0_S\qquad\text{for every integer }r\ge1,
\end{equation*}
where $r\,1_S$ is the sum of $r$ copies of $1_S$.
Restrict the target to positive entrywise addition and multiplication,
semiring matrix products, transpose, broadcasting, diagonal operations,
and constants in $\mathbb N1_S\coloneqq\{r\,1_S:r\ge0\}$.
Arbitrary entrywise decoding and negative constants are not allowed here.
Use the query $\mathcal G_k$ of \S\ref{sec:weightedG} with semiring
addition written as $+$.

On Boolean-encoded inputs, every entry of every intermediate array
belongs to $\mathbb N1_S$.
The map sending $0_S$ to zero and every other element of this subsemiring
to one is a homomorphism to the Boolean semiring. Indeed, sums and products
of positive multiples of $1_S$ are again nonzero by the nonvanishing
condition~\eqref{eq:positive-characteristic}. For any target circuit $C$
computing $\mathcal G_k$ on all finite inputs, applying this homomorphism
entrywise to its computation on Boolean-encoded inputs gives a Boolean
matrix circuit computing $G_k$ with the same typed product counts.
The Boolean lower bound therefore gives
$\mvcount{C}+2\mmcount{C}\ge2^k$.
Here the cost is applied to the Booleanized computation; no argument by
predictions is claimed for the original scalar circuit.
The $2^k$-product expansion still works by distributivity. This includes
positive computations over $\mathbb N$ and nonnegative reals.

The restriction on the operations matters. Over $\mathbb N$, arbitrary
size-specific entrywise encoding and decoding can compute the Boolean
query $G_k$ on Boolean-encoded inputs using a single row-sum aggregation.
Indeed, assign each pair profile $p\in\{0,1\}^{2k}$ a distinct exponent
$j(p)\in\{0,\ldots,2^{2k}-1\}$. The sum
\[
 s_x\coloneqq\sum_{z\in V}(n+1)^{j(\pi(x,z))}
\]
encodes every profile count as a base-$(n+1)$ digit, with no carry
because each count is at most $n$. Broadcasting $s_x$ and decoding it
together with the profile $\pi(x,y)$ determines $G_k(x,y)$.
Nonidempotence alone does not guarantee this encoding.
The finite-join-breadth assumption in Theorem~\ref{cor:matrix} prevents
one aggregation from depending on arbitrarily many indispensable colors.

\section*{Acknowledgments}
Yuya Uezato was supported by JST, CREST Grant Number JPMJCR21M3.
Yoshiki Nakamura was supported by JSPS KAKENHI Grant Number JP25K14985.

\appendix
\section{Explicit two-witness colorings}\label{app:explicit}
This appendix gives a deterministic version of Lemma~\ref{lem:palette} on
selected domain sizes. The affine-plane construction follows the
direction-partition method of Jipsen, Maddux, and
Tuza~\citep{jipsen-maddux-tuza1995}. Their assignment of at least two
directions to every color already guarantees two witnesses at plane
order at least four; we include the verification because this is the
property used here.

\begin{proposition}[Explicit two-witness colorings]\label{prop:explicit-graphs}
For every palette size $N\ge2$, a graph with the two-witness property
of Lemma~\ref{lem:palette} can be constructed deterministically on fewer
than $16N^2$ vertices.
\end{proposition}
\begin{proof}
Let $s$ be the smallest power of two with $s\ge2N-1$, and let
$\mathbb F_s$ be the field of $s$ elements. Use vertex set $V\coloneqq\mathbb F_s^2$.
The directions of affine lines are the one-dimensional linear
subspaces of $\mathbb F_s^2$, of which there are $s+1$. Partition these
directions into sets $\operatorname{Dir}_\mu$, $\mu\in\Lambda$, each of
size at least two. Color distinct points $u,v$ by the unique color $\mu$
for which the direction of their joining line belongs to
$\operatorname{Dir}_\mu$. Since $N\ge2$, we have $s\ge4$ and
$|V|=s^2<16N^2$.

Fix distinct $u,v$ and requested colors $\mu,\nu$.
If $\mu=\nu=\chi(u,v)$, all $s-2\ge2$ other points of the line $uv$
are witnesses. Otherwise, choose a direction from
$\operatorname{Dir}_\mu$ at $u$ and a direction from
$\operatorname{Dir}_\nu$ at $v$, neither equal to the direction of $uv$,
and choose the two directions to be distinct. Their two lines intersect
at a unique witness outside $\{u,v\}$. There are at least two choices:
if the colors differ, their direction sets are disjoint and at most one
set loses a direction to $uv$; if they agree but differ from $\chi(u,v)$,
two directions from their common set can be taken in either order.
Distinct direction pairs give distinct intersection points, since a
point outside $\{u,v\}$ determines both directions uniquely. Thus the
coloring has the two-witness property.
\end{proof}
For $N=2^k$, one can choose $s=2N$, giving $4^{k+1}$ vertices.
This construction supplies selected domain sizes, rather than every large
size; induced subgraphs need not retain the two-witness property.

\section{Tree bounds and limits of the method}\label{app:limits}
\subsection{The input-occurrence charging lemma}\label{app:trees}
This subsection proves the tree refinement stated in
\S\ref{sec:trees}.

\treecharging*
\begin{proof}
For a subterm whose value on the reference structure is $X$, let
$S_X\subseteq\Lambda$ be the off-diagonal support of $X$. Boolean operations
act on these sets in the usual way and converse fixes them. By the
identities~\eqref{eq:table-off}, a composition $X;Y$ with both supports nonempty
has off-diagonal support $\Lambda$. Call such a node \defn{active}.
For an inactive node and $\mu\in\Lambda$, \eqref{eq:table-off} gives
\[
\begin{aligned}
 S_X=\varnothing &\implies (X;Y)[\mu]=X[\diagc]\land Y[\mu],\\
 S_Y=\varnothing &\implies (X;Y)[\mu]=X[\mu]\land Y[\diagc].
\end{aligned}
\]
The diagonal bits are fixed Boolean values on the reference structure,
independent of $\mu$. Hence the product's off-diagonal table is zero
or one of the operand tables. Such an inactive node notices no fold: every outside-witness
contribution in the local calculation vanishes.

Construct a Boolean formula for the off-diagonal table of any subterm
by replacing active compositions by the constant true, inactive
compositions by zero or the relevant operand table, and converse by
its argument. This operation never duplicates an input occurrence.
For each active node $g=(X;Y)$, apply this pruning to its two operands.
Let $\mathcal L_g$ be the set of original input occurrences that remain
in the two resulting table formulas. These are the occurrences that
\emph{survive the pruning}, not all input occurrences below $g$.
This distinction prevents repeated charging at successive active nodes.
The sets $\mathcal L_g$ are
pairwise disjoint: an active descendant is replaced by a constant when
an ancestor's operands are pruned, and incomparable nodes have disjoint
subtrees.

By Lemma~\ref{lem:local}, a composition notices the fold of a color
$\lambda$ only if
\[
 S_X=\{\lambda\},\qquad S_Y=\{\lambda\},\qquad\text{or}\qquad
 S_X\cap S_Y=\{\lambda\}.
\]
Let $E_g$ be the set of colors satisfying at least one of these
conditions. Thus $|D_g|\le|E_g|$. If $S_X\cap S_Y$ is nonempty,
$E_g$ has size at most one. Its unique color, if present, is selected
by one operand table or by their conjunction. The selector assumption
therefore gives $|\mathcal L_g|\ge d|E_g|$. If the two supports are
disjoint, only singleton operands contribute to $E_g$. Each contributing
operand has at least $d$ input occurrences, and the operands are disjoint
trees. The same inequality follows. Summing over active nodes proves
\[
 \tcost(T)=\sum_g|D_g|\le\sum_g|E_g|
       \le\frac1d\sum_g|\mathcal L_g|\le \relatoms{T}/d.
\]
Diagonal bits may affect which operand table an inactive node retains, but do not
create additional input occurrences in this argument.
\end{proof}

\subsection{Sharpness of the fold tests}\label{app:promise}
The lower bound of Theorem~\ref{thm:main} uses only the reference
structure and its folds. On this test family, the factor of two
in Lemma~\ref{lem:local} is optimal: one composition per complementary
pair of colors suffices.

\begin{proposition}[Sharpness on the folded structures]\label{prop:promise}
For every $k\ge1$, the minimum number of composition gates in a
relational circuit that computes $G_k$ on the reference structure and on
all its folded structures is $2^{k-1}$.
\end{proposition}
\begin{proof}
For each unordered complementary pair $\{A,A^c\}$ of colors, choose one
representative $A$ and compute, once,
\[
 X_A\coloneqq H_A;H_{A^c}.
\]
Use the shared result in
\[
 \bigcup_{\{A,A^c\}}
 \left[(H_{A^c}\cap X_A)\cup(H_A\cap X_A^\smile)\right].
\]
This circuit has exactly $2^{k-1}$ compositions. We check that it
computes $G_k$ on $\mathcal M_0$ and on every $\mathcal M_\lambda$. Let
$J$ be the set of colors at the hub, so that $J=\Lambda$ on
$\mathcal M_0$ and $J=\Lambda\setminus\{\lambda\}$ on
$\mathcal M_\lambda$. On each of these structures, $H_A$ is the set of
edges of color $A$. The relation $X_A$ is true on all pairs of distinct
non-hub vertices, has zero diagonal, and has hub-row and hub-column
values $[A\in J]$ and $[A^c\in J]$, respectively.
The masks $H_{A^c}$ and $H_A$ restrict the two terms to pairs of colors
$A^c$ and $A$, respectively. Every non-hub vertex still has edges of
every color, so $G_k(u,v)$ is true whenever $u\ne\hub$ and $v\ne u$.
At the hub, $G_k(\hub,v)$ holds for $v\ne\hub$ precisely when the color
complementary to that of $(\hub,v)$ lies in $J$. These are exactly
the displayed circuit's values. The lower bound is
Theorem~\ref{thm:main}, whose proof uses only these structures.
\end{proof}

\subsection{Random-input success bounds}\label{app:random-inputs}
We prove Corollary~\ref{cor:random-inputs}. The circuit is fixed and
the input is sampled from the color-and-fold distribution specified
there; this is distinct from the randomized targets of
\S\ref{sec:prelim-resources}.

\begin{proof}[Proof of Corollary~\ref{cor:random-inputs}]
If $\delta\ge1$, both bounds are trivial. Hence assume $\delta<1$, which implies $n\ge4$.
Let $\mathcal W$ be the event that the coloring has the two-witness
property. The proof of Lemma~\ref{lem:palette}, with $N=2^k$, gives
$\Pr[\lnot\mathcal W]\le\delta$. Fix a coloring in $\mathcal W$. Then
$\hub$ has edges of every color, so $v$ is uniform among the vertices
with $\chi(\hub,v)=\lambda^c$. For each color $\lambda$, let
$\mathrm{err}_{\lambda,0},\mathrm{err}_{\lambda,1}$ indicate an error of
$C$ at $(\hub,v_{\lambda^c})$ on $\mathcal M_0$ and on $\mathcal M_\lambda$.
The error-count inequality in the proof of Theorem~\ref{thm:test-cost},
with the tests of Lemma~\ref{lem:G-tests} and Lemma~\ref{lem:cor-cost},
gives
\[
 \sum_\lambda(\mathrm{err}_{\lambda,0}+\mathrm{err}_{\lambda,1})
 \ge2^k-\tcost(C)\ge2^k-2g
\]
for every choice of the vertices $v_\mu$ of \S\ref{sec:structure}, hence
also in expectation when each $v_\mu$ is uniform among the vertices with
$\chi(\hub,v_\mu)=\mu$. Given the coloring, the probability of
an error is the average of $\mathrm{err}_{\lambda,b}$ over uniform
$\lambda$ and $b$, so it is
at least $(2^k-2g)/2^{k+1}=\frac12-g/2^k$. Averaging over the colorings
proves the upper bound.

For the lower bound, choose $g$ of the pairs $\{\lambda,\lambda^c\}$ and keep, in the
circuit of Proposition~\ref{prop:promise}, only the terms for these pairs.
This circuit has $g$ composition gates. On a coloring in $\mathcal W$, it
is correct on both structures of the $2g$ tests whose colors lie in the
chosen pairs, and for every other color $\lambda$ it answers zero at
$(\hub,v_{\lambda^c})$, which is correct when the bit is one. Its success probability is
therefore at least $(1-\delta)(\frac12+g/2^k)\ge\frac12+g/2^k-\delta$.
\end{proof}

\section{Circuits without converse}\label{app:directed}
We prove the sharper bound for circuits without converse with a directed
version of the reference structure.

\withoutconverse*
\begin{proof}
Color each ordered pair of distinct vertices independently from the
palette $\Lambda=2^{[k]}$; write $\chi(u,v)$ for its color, with no
relation between $\chi(u,v)$ and $\chi(v,u)$. Require that for every
$u,v$, now allowing $u=v$, and every $\mu,\nu\in\Lambda$, at least two
vertices $z$ outside $\{u,v\}$ realize
\[
 (\chi(u,z),\chi(z,v))=(\mu,\nu).
\]
Candidate witnesses use disjoint ordered arcs, even when $u=v$, so the
proof of Lemma~\ref{lem:palette} applies with the same threshold
$n_0(N)$ and at most $n^3N^2$ bad events. Interpret $P_i$ and $Q_i$ on
each arc by its color as in \S\ref{sec:G}, with all loop entries zero.

Reference relations have a diagonal value $X[\diagc]$ and one value
$X[\mu]$ per arc color $\mu$; let $S_X$ be the set of colors on which $X$
holds. These tables are closed under Boolean operations and
composition: endpoint witnesses contribute as in~\eqref{eq:table-off},
and every pair, including a diagonal pair, has an outside witness
exactly when $S_X$ and $S_Y$ are both nonempty. The fold of $\lambda$
now recolors only the arcs of color $\lambda$ that leave $\hub$, to
$\lambda^c$, and leaves all arcs entering $\hub$ unchanged; predictions
read the tables on this new coloring. Outside the row of $\hub$, the
composition of the predicted operands equals the prediction of the
composition, since for every pair of colors some outside witness other
than $\hub$ uses two unchanged arcs. In the row of $\hub$, including
$(\hub,\hub)$, an outside witness exists exactly when
$S_X\setminus\{\lambda\}$ and $S_Y$ are both nonempty. Hence a
composition notices the fold of $\lambda$ only if $S_X=\{\lambda\}$, and
so it notices at most one fold.

Choose $v_\mu$ so that the arc $(\hub,v_\mu)$ has color $\mu$, and read
the output at $(\hub,v_{\lambda^c})$. The query $G_k$ inspects only arcs
leaving $x$, so the first assertion of Lemma~\ref{lem:G-tests} holds
with the same proof, reading the edges at a vertex as the arcs leaving
it, and Theorem~\ref{thm:test-cost} gives $\gcost C\ge2^k$ for
relational circuits without converse. The DNF of $G_k$ uses no converse
and attains this bound on arbitrary inputs. Converse is excluded because
the color of an arc does not determine the color of its reverse arc; the
coefficient-one argument consequently does not apply to full \CoR.
\end{proof}

\section{Single-exponential translation bounds}\label{app:classical-translation}
\subsection{The classical translation into \texorpdfstring{\CoR}{CoR}}\label{app:classical-cor}
We include the size analysis for the classical expressive equivalence
\citep{tarski-givant1987,maddux2006},
which gives another proof of the single-exponential translation
of~\citep[Theorem~3.1(1)]{nakamura2022}, and our proof is reused in \S\ref{sec:fixed-upper}. The relevant point is that successive
quantifier eliminations need not create a tower of exponentials: an
already translated subformula is treated as one leaf of the next Boolean
expansion, without expanding its algebraic representation again.

\begin{proposition}[{Cor.\ of~\citep[Theorem~3.1(1)]{nakamura2022}}]\label{prop:upper-translation}
Every \FOthree\ formula $\varphi$ over binary relations with at most two
free variables has an equivalent \CoR\ term of size $2^{O(|\varphi|)}$
on all nonempty structures, constructible in time $2^{O(|\varphi|)}$.
Here size counts syntax-tree nodes.
\end{proposition}
\begin{proof}
We use the binary-assignment convention of \S\ref{sec:prelim-logic}:
represent each subformula with at most two free variables on an
ordered pair containing them. An unused coordinate is a dummy, so
moving the representation to another suitable pair requires only
converse (if the coordinates are exchanged). On each structure, the
representation of a sentence has value $\zero$ or $\one$, according
to its truth value on that structure.

The atomic translations are those of \S\ref{sec:FOtwo-translation},
with the two relevant variables renamed or reordered: distinct-variable
atoms use the input relation or its converse, reflexive atoms use a
diagonal test followed by a projection, and equality uses $\id$ or
$\one$. All have constant size. Boolean operations on a fixed
variable pair become intersection, union, and complement.

Process quantified subformulas from the inside out. Consider
$\exists w\,\psi$, and let $u,v$ be the other two variable names, in
that order. Its free variables are contained in $\{u,v\}$.
Form the Boolean \defn{shell} of $\psi$ by replacing each maximal
quantified subformula occurrence by a single leaf; the remaining atoms
are its other leaves. Here maximal means not contained in another
quantified subformula of $\psi$.
Each leaf has at most two free variables: this is immediate for binary
atoms, and a quantifier binds one of the three available names.
The quantified leaves have already been translated. Crucially, the
shell expansion treats each of their translations as an opaque term.

If the shell has $r$ leaf occurrences, a truth-table DNF has at most
$2^r$ clauses, each using at most $r$ literals. Repeated or logically
dependent leaves cause no problem; the propositional identity remains
valid after substituting their actual interpretations. In each clause,
assign its literals to three groups according to the free-variable pair:
$(u,v)$, $(u,w)$, or $(w,v)$. A literal with only one free variable is
assigned to any group containing it; a closed literal can be put in the
first group. Use converse to orient each translated literal for its
chosen pair, and use complement for a negative literal.
Writing the conjunctions of the groups as $A(u,v)$, $B(u,w)$, and
$C(w,v)$, apply the identity~\eqref{eq:prelim-elimination}:
\[
 \exists w\,(A(u,v)\land B(u,w)\land C(w,v))
       \quad\longmapsto\quad A\cap(B;C).
\]
An empty group is represented by $\one$. This identity also covers
vacuous quantification because domains are nonempty.
Take the union of the translated clauses. Universal quantification uses
$\forall w\,\psi\equiv\neg\exists w\,\neg\psi$, which introduces no
subformula duplication. At the root, any remaining Boolean shell with
at most two free variables is translated directly, after orienting its
leaves to a common ordered pair. This proves correctness on all nonempty
structures by induction.

For the size bound, cut the original syntax tree at
quantifiers into the shells just used. Count each child quantified
subformula as one leaf of its parent's shell. The sum of the shell sizes,
including these boundary leaves, is $O(|\varphi|)$: each original node is counted
once, with only one extra boundary occurrence per quantified subformula.
A shell of size $s$ creates at most $2^{O(s)}$ local nodes and at most
$2^{O(s)}$ copies of each child translation. Along any path from an
original leaf or a newly introduced local node to the final output,
the resulting duplication factor is therefore at most
\[
 \prod_{\text{shells on the path}}2^{O(s)}
       \le 2^{O(|\varphi|)}.
\]
There are $O(|\varphi|)$ original nodes and shell locations, and the local
exponential creation factor is included in this same bound. Summing
their contributions gives $|\varphi|\,2^{O(|\varphi|)}=2^{O(|\varphi|)}$ nodes.
No step expands the internal Boolean structure of a translated child,
which is what prevents an iterated-exponential recurrence.
Enumerating the shell assignments and writing the resulting term
takes a polynomial factor in $|\varphi|$ times the same exponential
bound, hence also $2^{O(|\varphi|)}$ time.
\end{proof}

\subsection{One-sided single-quantifier queries}\label{app:one-sided-fo2}
A \defn{one-sided single-quantifier query} is a formula
$\varphi(x,y)\coloneqq\exists z\,\theta(x,y,z)$ with $\theta$
quantifier-free such that, after possibly exchanging $x$ and $y$,
no atom contains both $y$ and $z$.
In this orientation, each atom uses only variables from $\{x,y\}$
or only variables from $\{x,z\}$; equality atoms are included.
Thus $\{y\}$ and $\{z\}$ are separated in the sense of the separated fragment~\citep{sturm-voigt-weidenbach2016}, \citep[Definition~2.0.1]{voigt-thesis2019}, while $x$ may occur with either.
This syntactic condition on queries is distinct from the semantic
one-sidedness of composition gates in \S\ref{sec:simultaneous}.

\begin{proposition}[Tight bounds for one-sided queries]\label{prop:one-sided-fo2}
Let $\varphi(x,y)\coloneqq\exists z\,\theta(x,y,z)$ be a one-sided
single-quantifier query, with $\theta$ in negation normal form and containing
$d$ occurrences of $\lor$.
There is an equivalent \FOtwo\ formula $\psi$ over the same signature,
of quantifier depth at most one, such that
\[
 \#_{\exists\forall}(\psi)\le 2^d,
 \qquad
 |\psi|=O\!\left(|\varphi|2^d\right).
\]
If $\theta$ is positive, $\psi$ can be chosen existential-positive.
The quantifier bound is optimal for every $d\ge1$, already for
positive source formulas over binary signatures and unrestricted
\FOtwo\ targets.
\end{proposition}
\begin{proof}
Exchange $x,y$ if necessary, so that no atom contains both $y$ and $z$.
Distributing conjunction over disjunction gives a DNF with at most
$2^d$ clauses, each of size $O(|\varphi|)$.
For the clause count, a literal gives one clause; conjunction
multiplies the two clause counts, whereas disjunction adds them.
Structural induction gives the bound: at a disjunction with $d_1,d_2$
disjunctions in its children,
$2^{d_1}+2^{d_2}\le2^{d_1+d_2+1}$.
No clause uses any original literal occurrence more than once.

In each clause, place all literals containing $z$ in $\beta_j(x,z)$
and the remaining literals in $\alpha_j(x,y)$.
The one-sided condition ensures that $\beta_j$ does not contain $y$.
Thus, for some $t\le2^d$,
\[
 \varphi(x,y)
 \equiv\bigvee_{j=1}^{t}
       \exists z\,\bigl(\alpha_j(x,y)\land\beta_j(x,z)\bigr)
 \equiv\bigvee_{j=1}^{t}
       \bigl[\alpha_j(x,y)\land\exists y\,\beta_j(x,y)\bigr].
\]
Here $\beta_j(x,y)$ is obtained by renaming $z$ to $y$;
the quantifier binds $y$ only in the second conjunct.
Empty groups denote true, represented by $x=x$ if needed;
vacuous quantification is harmless on nonempty domains.
Each disjunct has one quantifier and $O(|\varphi|)$ nodes, proving
the bounds. The construction preserves positivity, and any initial
exchange of $x,y$ can be undone.

For optimality, take $G_d$ from \S\ref{sec:main-results}.
Its positive body has exactly $d$ disjunctions, while
Theorem~\ref{cor:sentences} requires at least $2^d$ quantifiers in
every equivalent \FOtwo\ formula, even if equivalence is required
only on finite structures.
\end{proof}

Negation normal form has linear size and preserves the one-sided
condition. Consequently, among queries in this class of source
syntax-tree size at most $s$, the worst-case minimum \FOtwo\ size
is $2^{\Theta(s)}$, with the binary signature allowed to grow with
the input. For $G_d$, the size bound is also tight up to a constant
factor: $|G_d|=8d$, and its minimum \FOtwo\ tree size is
$\Theta(d2^d)$ by \S\ref{sec:sentences}.
This is a bound for the stated syntactic class, not for arbitrary
\FOtwo-definable \FOthree\ formulas.

\begin{remark}[Why the one-sided restriction matters]
Allowing the witness to interact with both output variables can
already destroy \FOtwo-definability: the query
$\exists z\,(E(x,z)\land E(z,y))$ is not definable in \FOtwo,
even over finite undirected graphs viewed as symmetric directed graphs.
It holds on an edge of two disjoint triangles $C_3\sqcup C_3$,
but not on an edge of the six-cycle $C_6$.
Nevertheless, these pointed graphs are indistinguishable in the
two-pebble game for any number of rounds: relative to the fixed pebble,
each graph always offers an equal vertex, an adjacent vertex, and a
distinct nonadjacent vertex. Matching these three atomic pair types
preserves a partial isomorphism after every move~\citep{libkin2004}.
The same pairs separate the binary triangle query
$\exists z\,(E(x,y)\land E(y,z)\land E(z,x))$.
\end{remark}

\section{Small-domain circuits}\label{app:small-domains}
We give the complete construction underlying
Proposition~\ref{prop:small-domains}. As in \S\ref{sec:domain-size},
the pair profile is
\[
 \pi(x,z)\coloneqq
 ([P_1(x,z)],\ldots,[P_k(x,z)],
  [Q_1(x,z)],\ldots,[Q_k(x,z)]),
 \qquad \Pi_x\coloneqq\{\pi(x,z):z\in V\}.
\]
The circuit must work simultaneously on all nonempty structures of
size at most $n$; the slots below do not name input vertices.

\smalldomains*
\begin{proof}
Fix an arbitrary input structure with domain $V$, where $1\le |V|\le n$.
A relation is \defn{row-constant} if its value at $(x,y)$ does not
depend on $y$. By~\eqref{eq:prelim-projections}, each row projection $S;\one$
is row-constant and holds in row $x$ exactly when some pair
$(x,z)$ lies in $S$; Boolean
combinations of row-constant relations are row-constant. For
$1\le b\le2k$, let $E_b$ be the input relation giving bit $b$ of the
profile, and put $E_b^1\coloneqq E_b$ and $E_b^0\coloneqq\overline{E_b}$.

Call a word $w\in\{0,1\}^\ell$ \defn{live} in row $x$ if some profile in $\Pi_x$
begins with $w$. Live words of one length begin distinct profiles, so
there are at most $n$ of each length, and the empty word is live because
$\pi(x,x)\in\Pi_x$. The circuit stores the live words of length $\ell$ in
$n$ slots: row-constant relations $L_{\ell,r}$ and $W_{\ell,r,b}$, for
$1\le r\le n$ and $1\le b\le\ell$,
state, respectively, that slot $r$ is occupied in row $x$ and that the $b$-th bit of its stored word is $1$.
The invariant is that, in each row, the
occupied slots store every live word of length $\ell$ exactly once;
bits in unoccupied slots are irrelevant. Initially $L_{0,1}\coloneqq\one$ and
$L_{0,r}\coloneqq\zero$ for $2\le r\le n$; the only stored word
is empty, so there are no word bits. For $0\le\ell<2k$, each slot $r$,
and each $\sigma\in\{0,1\}$, the relation
\[
 S_{\ell,r,\sigma}\coloneqq L_{\ell,r}\cap
 \bigcap_{b=1}^{\ell}\bigl((E_b\cap W_{\ell,r,b})\cup
     (\overline{E_b}\cap\overline{W_{\ell,r,b}})\bigr)\cap E_{\ell+1}^\sigma
\]
holds at $(x,z)$ exactly when slot $r$ is occupied in row $x$ and
$\pi(x,z)$ begins with its word followed by $\sigma$. The row projection
$S_{\ell,r,\sigma};\one$ therefore tells, in every row at once, whether
this extension is live. The at most $n$ live extensions are placed into the slots of length
$\ell+1$ in the fixed order of their pairs $(r,\sigma)$; no order on
the input vertices is used. Index the $2n$ candidate pairs by
$(r_i,\sigma_i)$, for $1\le i\le 2n$, and put
$f_i\coloneqq S_{\ell,r_i,\sigma_i};\one$. Using only Boolean gates,
compute, for $0\le i\le 2n$ and $0\le j\le n$, the row-constant
relation $C_{i,j}$ that holds in a row exactly when $j$ of
$f_1,\ldots,f_i$ hold:
\[
 C_{0,0}\coloneqq\one,\qquad C_{0,j}\coloneqq\zero\ (j>0),\qquad
 C_{i,j}\coloneqq(\overline{f_i}\cap C_{i-1,j})
       \cup(f_i\cap C_{i-1,j-1}),
\]
where $1\le i\le 2n$, $0\le j\le n$, and
$C_{i,-1}\coloneqq\zero$. For $1\le j\le n$, put
\[
 A_{i,j}\coloneqq f_i\cap C_{i-1,j-1}\qquad(1\le i\le 2n).
\]
In each row, $A_{i,j}$ holds exactly if candidate $i$ is the $j$th
live candidate. In particular, at most one $A_{i,j}$ holds for a given
slot $j$. Thus occupancy and every bit of the new slots are computed by
\[
\begin{aligned}
 L_{\ell+1,j}
   &\coloneqq\bigcup_{i=1}^{2n} A_{i,j},\\
 W_{\ell+1,j,b}
   &\coloneqq\bigcup_{i=1}^{2n}
       (A_{i,j}\cap W_{\ell,r_i,b})
       &&(1\le b\le\ell),\\
 W_{\ell+1,j,\ell+1}
   &\coloneqq\bigcup_{\substack{1\le i\le 2n\\\sigma_i=1}} A_{i,j}.
\end{aligned}
\]
Distinct live prefixes have distinct extensions, and every live word
of length $\ell+1$ extends a live prefix. Thus these formulas preserve
the slot invariant. They use only Boolean operations on row-constant
relations, so no further compositions are needed. At each length, the
$2n$ projections producing
the $f_i$ are the only composition gates. The candidate tests use
$O((\ell+1)n)$ Boolean gates, the prefix counts and ranks use
$O(n^2)$, and updating all slots and their $\ell+1$ bits uses
$O((\ell+1)n^2)$. Summing over the $2k$ lengths, together with the
final output gates below, yields $O(k^2n^2)$ non-composition gates.
After length $2k$ the occupied slots
hold exactly the profiles in $\Pi_x$, and
\[
 G_k=\bigcup_{r=1}^{n}\Bigl(L_{2k,r}\cap\bigcap_{i=1}^k
 \bigl((P_i\cap W_{2k,r,k+i})\cup(Q_i\cap W_{2k,r,i})\bigr)\Bigr).
\]
For each occupied slot $r$, the intersection over $i$ tests the $k$
clauses of $G_k$ using $\pi(x,y)$ and the profile stored in that slot.
Since every stored profile occurs in row $x$, the union holds exactly
when some vertex $z$ witnesses $G_k(x,y)$. The circuit
makes $2n$ row projections
at each of the $2k$ lengths, $4kn$ in all, and uses no converse.
\end{proof}

\section{AI methodology}
\label{sec:ai-methodology}

The authors used ChatGPT (OpenAI) to assist with translating
author-written Japanese drafts into English, improving the exposition,
and reviewing proofs for potential errors.
It also assisted with literature searches.
In particular, ChatGPT drew our attention to similarities between
our reference constructions and those of Jipsen, Maddux, and
Tuza~\citep{jipsen-maddux-tuza1995}, leading us to cite their work
and revise the attribution and discussion in
\S\ref{sec:overview}, \S\ref{sec:structure},
\S\ref{sec:types}, and Appendix~\ref{app:explicit}.
The authors take full responsibility for the paper, including
its mathematical claims, proofs, and references.

\clearpage
\bibliographystyle{alphaurl}
\bibliography{bib}

\end{document}